\documentclass[12pt]{article}

\usepackage{iftex}
\ifLuaTeX
\usepackage[haranoaji]{luatexja-preset}
\fi

\newcommand{\maybeusepackage}[2][]{%
  \IfFileExists{#2.sty}{\usepackage[#1]{#2}}{}%
}

\usepackage{graphicx}
\usepackage{amsmath,amssymb,amsfonts}
\maybeusepackage{mathabx}
\maybeusepackage{ascmac}
\maybeusepackage{arydshln}
\usepackage{booktabs}
\maybeusepackage{algorithm}
\maybeusepackage{algpseudocode}
\maybeusepackage{mathrsfs}
\maybeusepackage{stmaryrd}
\maybeusepackage{lipsum}
\maybeusepackage{titling}
\maybeusepackage{comment}
\usepackage{setspace}
\usepackage[breaklinks=true]{hyperref}
\maybeusepackage{breakcites}
\maybeusepackage{algorithmicx}
\maybeusepackage{latexsym}
\maybeusepackage{cases}
\usepackage{url}
\maybeusepackage{pifont}
\usepackage{subcaption}
\maybeusepackage{makecell}
\maybeusepackage{multirow}
\usepackage{tikz}
\usetikzlibrary{patterns}
\usetikzlibrary{patterns.meta}
\usetikzlibrary{arrows.meta}
\usepackage{here}
\usepackage{natbib}

\usepackage{microtype}
\maybeusepackage{babel}
\maybeusepackage{threeparttable}
\usepackage[a4paper]{geometry}
\usepackage{amsthm}
\hypersetup{
bookmarkstype=none,
colorlinks,
linkcolor=blue,
citecolor=blue,
urlcolor=blue}
\maybeusepackage{prettyref}
\allowdisplaybreaks[3]

\tikzstyle{rednode} = [shape=rectangle, fill=red!50, line width=3]
\tikzstyle{bluenode} = [shape=rectangle, fill=blue!50, line width=3]
\tikzstyle{yellownode} = [shape=rectangle, fill=yellow!50, line width=3]
\tikzstyle{dotnode} = [dashed, pattern={Lines[angle=90,distance=3pt]}, pattern color=gray!150]
\tikzstyle{backslashnode} = [dashed, pattern={Lines[angle=45,distance=3pt]}, pattern color=gray!150]
\tikzstyle{slashnode} = [dashed, pattern={Lines[angle=-45,distance=3pt]}, pattern color=gray!150]

\newtheorem{theorem}{Theorem}[section]

\newtheorem{proposition}[theorem]{Proposition}
\newtheorem{lemma}[theorem]{Lemma}

\newtheorem{corollary}[theorem]{Corollary}
\newtheorem{assumption}[theorem]{Assumption}

\theoremstyle{definition}
\newtheorem{definition}[theorem]{Definition}
\newtheorem{remark}[theorem]{Remark}

\DeclareMathOperator*{\argmin}{arg\,min}

\newcommand{\Reg}{\mathrm{Reg}}

\newcommand{\safeincludegraphics}[2][]{%
  \IfFileExists{#2}{\includegraphics[#1]{#2}}{%
    \fbox{\parbox[c][0.22\textheight][c]{0.88\linewidth}{\centering Missing figure file\\[3pt]\texttt{\detokenize{#2}}}}%
  }%
}
\renewcommand{\epsilon}{\varepsilon}
\renewcommand{\phi}{\varphi}
\renewcommand{\hat}{\widehat}

\makeatletter
\newcommand{\mdformalname}{}
\newcommand{\mdlabel}[2]{\phantomsection\def\@currentlabel{#1}\label{#2}}
\makeatother
\theoremstyle{plain}

\theoremstyle{definition}

\begin{document}

\setstretch{1.25}
\title{
How Much to Experiment? Admissible Sample Sizes in Test-and-Roll Experiments\thanks{We thank Toru Kitagawa, Ryo Okui, Katsumi Shimotsu, and Kohei Yata for helpful comments. The authors gratefully acknowledge financial support from JSPS KAKENHI Grant Number 24K16342.}
\medskip
}

\author{
Kentaro Kawato\thanks{Graduate School of Economics, The University of Tokyo, 7-3-1 Hongo, Bunkyo-ku, Tokyo 113-0033, Japan. Email:~kawato-kentaro380@g.ecc.u-tokyo.ac.jp.}
\  and \ Shosei Sakaguchi\thanks{Faculty of Economics, The University of Tokyo, 7-3-1 Hongo, Bunkyo-ku, Tokyo 113-0033, Japan. Email:~sakaguchi@e.u-tokyo.ac.jp.}
}

\date{\today}

\maketitle

\vspace{-0.3cm}

\setstretch{1.2}

\begin{abstract}

This paper studies welfare-based sample-size choice in finite-population test-and-roll experiments. A decision-maker conducts a balanced randomized experiment on $m$ of $N$ units and assigns the remaining $N-m$ units to the arm with the higher sample mean. Regret combines the welfare cost of experimentation with the risk of an inferior rollout. Direct minimax regret is uninformative when treatment-effect magnitudes or noise scales are unrestricted, while minimax-regret choices under parameter bounds can depend materially on those bounds. We therefore derive a sharp characterization of admissible sample sizes, defined as those not dominated in regret by any other sample size under the same rollout rule. Under a Gaussian approximation and continuous relaxation, the admissible set for two treatments is exactly $(0,N/3]$, with $N/3$ strictly dominating every larger sample size. This characterization requires neither a prior nor externally specified parameter bounds, while partial information on the signal-to-noise ratio further narrows the admissible set. 
We extend the analysis to covariate targeting and multiple treatment arms. Under covariate targeting with a plug-in rule, the largest admissible sample size depends on the convergence rate of the conditional average treatment effect (CATE) estimator and falls below $N/3$ for slower-than-root-$m$ convergence. With three or more treatments, even arbitrarily small variance heterogeneity makes every sample size admissible, whereas under common variance the largest admissible sample size lies in $(N/3,N)$ and depends on the number of treatments.

\bigskip

\begin{spacing}{1.2}
\end{spacing}
\end{abstract}

\setstretch{1.5}

\newpage

\section{Introduction}\label{sec:introduction}

Randomized controlled trials (RCTs) are widely used to learn about treatment effects, but choosing the experimental sample size remains a central theoretical and practical challenge. Conventional approaches based on Type I and Type II error control are designed primarily for statistical inference and do not directly account for the welfare consequences of experimentation and subsequent treatment choice. This paper studies sample-size design from a welfare perspective, focusing on the exploration--exploitation tradeoff that arises when an RCT is conducted within a finite population.

We consider a ``test-and-roll'' setting with a finite population of \(N\) units. A decision-maker first conducts an RCT with \(m\) units, randomly assigning half to treatment and half to control. 
Based on the experimental outcomes, the decision-maker then assigns the remaining $N-m$ units to the arm with the higher sample mean, following the empirical success rule \citep{manski2004statistical_treatment_rules}.
The objective is to maximize expected welfare over the entire finite population, measured by the expected sum of outcomes across both the experimental and rollout units. The central question is how large the experimental sample size \(m\) should be. Figure~\ref{fig:1} illustrates this test-and-roll design.

\smallskip
\begin{figure}[H]
\centering
\includegraphics[width=0.95\linewidth]{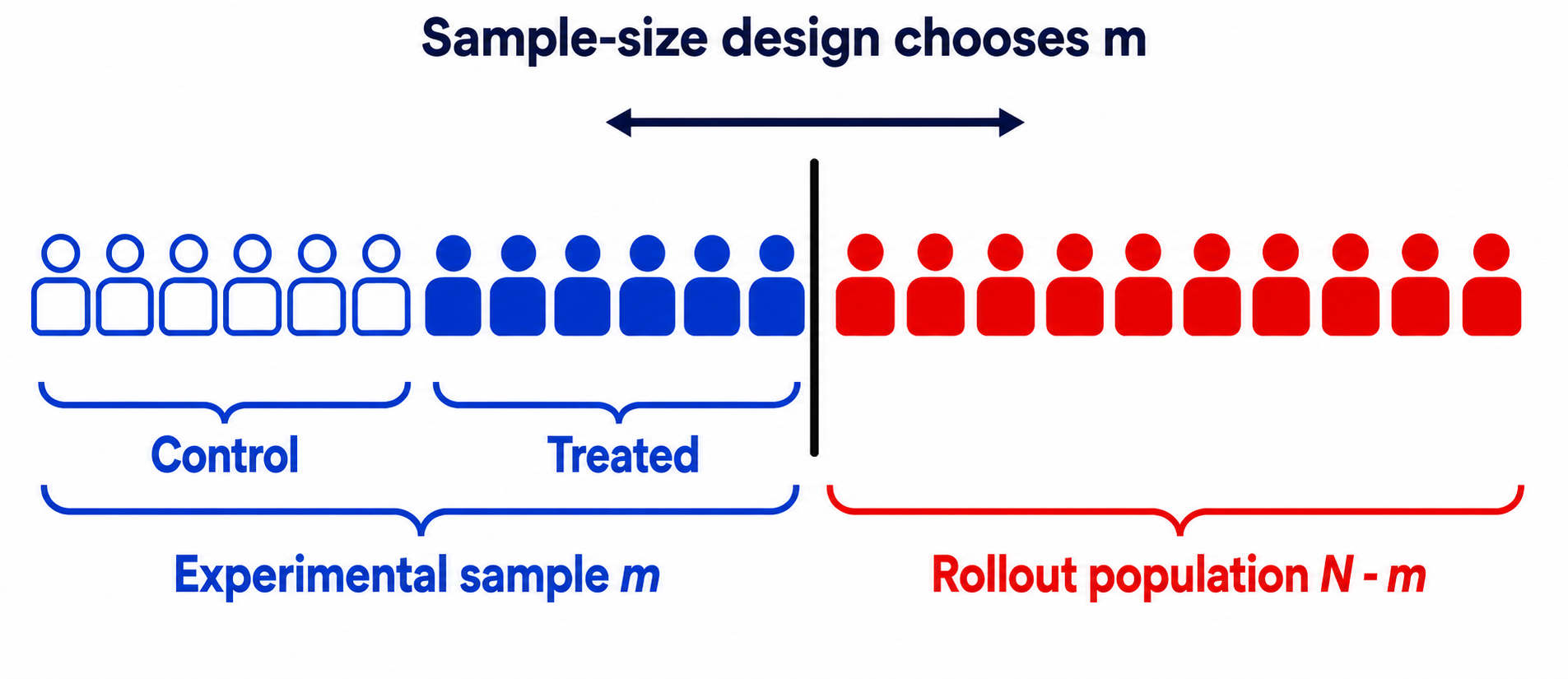}
\begin{minipage}{0.95\linewidth}
\footnotesize
\vspace*{-0.25cm}
\caption{Test-and-roll sample-size design. A balanced experiment of size \(m\) informs the treatment assigned to the remaining \(N-m\) units.}
\label{fig:1}
\end{minipage}
\vspace*{-0.3cm}
\end{figure}

This problem involves a fundamental exploration--exploitation tradeoff. A larger experiment lowers the probability of an inferior rollout, in which the remaining population is assigned to the inferior treatment. At the same time, a larger experiment exposes more units to random assignment and leaves fewer units to benefit from the treatment selected after the experiment. Thus, excessive exploration is costly, while insufficient exploration can lead to an inferior rollout and welfare losses for the rollout population.

The test-and-roll tradeoff arises in a range of experimental settings in which alternative interventions are first tested on a subset of a target population and the resulting evidence is used to select an intervention for the remaining population. Examples include digital A/B testing \citep{feit2019test_and_roll}, field experiments in which organizations pilot new practices \citep{friebel2017team}, and evaluations of alternative medical treatments \citep[e.g.,][]{wang2013clopidogrel}. In such settings, the experimental sample size determines both how much is learned and how much of the population remains for rollout, making sample-size choice inherently a welfare allocation problem.

A central challenge in optimizing the experimental sample size is that the treatment effect and outcome variances are unknown ex ante. A prior can convert this uncertainty into a Bayesian design \citep{cheng2003choosing_sample_size_decision_analysis, stallard2017optimal_sample_size_population_size, feit2019test_and_roll}, but the resulting sample-size choice depends on the prior specification. A direct minimax-regret approach avoids a prior, but provides no informative guidance when either the treatment-effect magnitude or the noise scale is unrestricted. Under a Gaussian approximation to the difference-in-means estimator, we show that imposing externally specified parameter bounds can restore informative guidance, but the minimax-regret sample size may depend materially on these bounds.

We therefore take a different approach. Rather than selecting a single minimax-regret sample size, we ask which sample sizes can be ruled out uniformly over the parameter space. 
Specifically, we derive a sharp characterization of admissible sample sizes under the empirical success rule, namely those not dominated in regret by any other sample size.
Our key tool is the \emph{envelope sample size}, defined as the supremum of the regret-minimizing sample sizes over all nonzero treatment effects and positive, finite outcome standard deviations in the binary-treatment model. Through a marginal cost-benefit analysis under the Gaussian approximation model, we show that the sharp envelope is $m^{\mathrm{Env}}=N/3$. Moreover, $N/3$ strictly dominates every larger sample size, while every sample size in $(0,N/3]$ is admissible. Thus, under the continuous relaxation, the admissible set is exactly $(0,N/3]$, without requiring a prior or externally specified parameter bounds. As the largest admissible sample size, $N/3$ is a conservative choice that guards against under-experimentation rather than an optimum for every parameter value.

The characterization of the envelope \(N/3\) and the admissible set \((0,N/3]\) use no quantitative information about the unknown parameters. We next show how partial information sharpens both. Let \(\rho\) denote the signal-to-noise ratio, which measures the magnitude of the treatment effect relative to outcome variability, and let \(m^*(\rho)\) denote the corresponding regret-minimizing sample size. The parameter-specific optimum \(m^*(\rho)\) is strictly decreasing in \(\rho\). If \(\rho\) is known to lie in \([\underline{\rho},\overline{\rho}]\), the set of admissible sample sizes shrinks from \((0,N/3]\) to exactly \([m^*(\overline{\rho}),\,m^*(\underline{\rho})]\). Thus, a lower bound on \(\rho\) rules out excessively large experiments, while an upper bound rules out excessively small ones.

We extend the baseline analysis in two important directions: covariate targeting and multiple treatment arms. In both settings, we characterize the envelope and the set of admissible sample sizes, showing how admissibility depends on the statistical structure of the problem.

Under covariate targeting, we consider a plug-in rollout rule that assigns treatments based on conditional average treatment effects (CATEs) estimated from the experimental data.
We show that the key determinant of the envelope is the convergence rate of the CATE estimator. For an estimator whose standard deviation decreases at rate \(m^{-\alpha}\), where \(0<\alpha\le1/2\), a Gaussian approximation yields the sharp envelope \(m_{\mathrm{cov}}^{\mathrm{Env}} = \alpha N/(1+\alpha)\). The admissible sample sizes are exactly those in \((0,m_{\mathrm{cov}}^{\mathrm{Env}}]\). A root-\(m\) estimator recovers the \(N/3\) envelope, while slower convergence yields a smaller envelope.

With $K\ge3$ treatment arms, the variance structure determines whether a nontrivial uniform upper bound on admissible sample sizes exists. We show that allowing even arbitrarily small variance heterogeneity across arms precludes such a bound. Every sample size is then admissible, so regret comparisons alone rule out no sample size. Under common variance, by contrast, a nontrivial boundary is restored: $m_K^{\mathrm{CBR}} = \beta_K N/(1+\beta_K) \in (N/3,N)$, where $\beta_K$ is a finite constant depending only on $K$. This boundary strictly dominates every larger sample size, and the admissible sample sizes are exactly those in $(0,m_K^{\mathrm{CBR}}]$.

We also illustrate the binary-treatment results using simulations calibrated to the randomized Job Training Partnership Act (JTPA) experiment. The simulations show that the \(N/3\) envelope is close to the Gaussian plug-in optimum, while larger experiments reduce the probability of an inferior rollout but increase total regret.

\subsection{Related Literature}\label{sec:related-literature}

Our paper relates first to the literature on sample-size choice and experimental design. Conventional sample-size calculations are typically based on power analysis, which chooses the experimental sample size to control Type I and Type II error probabilities \citep[e.g.,][]{Lachin1981, DufloGlennersterKremer2007}. Although useful for statistical inference, these calculations do not directly account for the welfare losses from random assignment during the experiment or the consequences of subsequent treatment choice. 

Recent studies have incorporated welfare considerations into experimental design and sample-size selection \citep{manski2016sufficient, azevedo2020ab_testing_fat_tails, azevedo2023ab_testing_gaussian_priors, narita2021incorporating, hu2024minimax_regret_sample_selection}. We share this welfare-based perspective but focus on a one-shot, finite-population test-and-roll problem in which the same target population supplies both the experimental and rollout units.

The studies most closely related to ours examine welfare-based test-and-roll designs. Bayesian decision-theoretic approaches have been developed for staged binary-response trials in finite patient populations \citep{cheng2003choosing_sample_size_decision_analysis}, extended to more general outcome models and utility functions \citep{stallard2017optimal_sample_size_population_size}, and applied to finite-population A/B testing with normally distributed outcomes and normal priors \citep{feit2019test_and_roll}. These studies share our welfare-based perspective, but their sample-size choices depend on specified prior distributions. By contrast, we characterize the admissible sample sizes without requiring a prior distribution, pilot estimates, or externally specified parameter bounds.

Our setting is also related to explore-then-commit methods \citep{garivier2016etc, perchet2016batched}, the broader literature on multi-armed bandits and adaptive treatment assignment \citep{LattimoreSzepesvari2020BanditAlgorithms, kasy2021adaptive, kock2023treatment}, and pilot-assisted two-stage designs \citep{HahnHiranoKarlan2011AdaptiveExperimental, cai2024performance}. Bandit models typically allow sequential or adaptive treatment allocation and evaluate cumulative rewards or regret over time. By contrast, our analysis focuses on choosing the sample size of a single balanced experiment ex ante, followed by a one-time treatment decision for the remaining population.

This paper also contributes to the literature on statistical treatment choice, which studies how experimental or observational data can be used to select treatments to maximize welfare \citep{manski2004statistical_treatment_rules, stoye2009minimax_regret_finite_samples, kitagawa2018who_should_be_treated, athey2021policy_learning}. A prominent strand of this literature studies covariate-based treatment targeting. These studies typically take the sample size as given and focus on learning treatment rules from data. By contrast, our analysis incorporates the welfare cost of experimentation into sample-size choice. In our covariate-targeting extension, we take the plug-in targeting rule as given and characterize how the convergence rate of CATE estimation determines the envelope and the set of admissible sample sizes.

\subsection{Structure of the Paper}\label{sec:structure-of-paper}

Section~\ref{sec:setup-minimax-regret} introduces the finite-population test-and-roll setting and examines direct minimax-regret sample-size choice. Section~\ref{sec:envelope-admissible} derives the sharp \(N/3\) envelope through a marginal cost-benefit analysis, shows that the admissible sample sizes are exactly those in \((0,N/3]\), and studies how partial signal-to-noise information narrows the admissible set. Section~\ref{sec:simulation-study} presents simulations calibrated to the JTPA experiment. Sections~\ref{sec:covariate-targeting} and~\ref{sec:multiple-treatments} extend the analysis to covariate targeting and multiple treatment arms, respectively. Section~\ref{sec:conclusion} concludes the paper. The remaining proofs are provided in Appendices~\ref{sec:appendix-A} and~\ref{sec:appendix-B} and the Supplementary Appendix.

\section{Setup and Direct Minimax Regret}\label{sec:setup-minimax-regret}

This section defines the one-shot test-and-roll design and the associated welfare and regret criteria. We then examine direct minimax regret for sample-size choice. We first show that this criterion provides no informative guidance when either the treatment-effect magnitude or the noise scale is left unrestricted. We then illustrate that sufficiently strong externally specified parameter bounds can restore informative guidance, but the minimax-regret sample size may depend materially on these bounds.

\subsection{Test-and-Roll Setup, Welfare, and Regret}\label{sec:test-and-roll-setup}

We consider a finite target population of \(N\) units, indexed by \(i=1,\ldots,N\). Unit \(i\) has potential outcomes \(Y_i(1)\) under treatment and \(Y_i(0)\) under control. The pairs \((Y_i(1),Y_i(0))\) are independent and identically distributed draws from a common distribution \(P\). For \(t\in\{0,1\}\), define
\[
\mu_t := \mathbb{E}\left[Y_i(t)\right],
\qquad
\sigma_t^2 := \operatorname{Var}(Y_i(t))\in(0,\infty).
\]
Let \(\tau:=\mu_1-\mu_0\) denote the average treatment effect and \(\sigma:=(\sigma_1,\sigma_0)\) the vector of arm-specific standard deviations.

Before observing any outcomes, the decision-maker chooses an experimental sample size \(m\in\mathcal{M}_N\), where \(\mathcal{M}_N := \{m\in\{2,\ldots,N-1\}: m \text{ is even}\}\) denotes the set of feasible sample sizes under the balanced two-arm design. Units \(i=1,\ldots,m\) form the experimental sample in the test phase. We consider a balanced randomized design in which exactly \(m/2\) units are assigned to treatment and \(m/2\) to control. Let \(T_i\in\{0,1\}\) denote the resulting treatment assignment. The corresponding observed outcome is \(Y_i=T_iY_i(1)+(1-T_i)Y_i(0)\) for \(i=1,\ldots,m\). The experimental data are \(\omega_{m}:=\{(T_i,Y_i):i=1,\ldots,m\}\).

After observing the experimental data \(\omega_m\), the decision-maker enters the rollout phase and selects a single arm for all remaining \(N-m\) units. The selection may be randomized conditional on \(\omega_m\). Let \(\delta(\omega_m)\in[0,1]\) denote a general rollout rule, representing the conditional probability of selecting arm 1 given \(\omega_m\). We refer to the pair \((m,\delta)\) as a \textit{policy}. Unless otherwise stated, expectations are taken over the potential-outcome draws from \(P\), the experimental randomization, and any randomization in the rollout rule.

Define the arm-specific sample means and their difference by
\[
\widehat\mu_t
:=
\frac{2}{m}\sum_{i=1}^m\mathbf 1\{T_i=t\}Y_i,
\qquad
\widehat\tau_m:=\widehat\mu_1-\widehat\mu_0.
\]
Throughout the main analysis, we adopt the empirical success rule for rollout:
\[
\widehat\delta_m(\omega_m)
:=
\mathbf 1\{\widehat\mu_1>\widehat\mu_0\}
=
\mathbf 1\{\widehat\tau_m>0\}.
\]
This rule assigns the rollout population to the arm with the larger estimated mean. 

Under the Gaussian sampling model with bounded treatment effects and variances, the empirical success rule is minimax-regret optimal for each fixed \(m\), even when the arm-specific variances are unknown; see Proposition~\ref{prop:fixed-m-empirical-success} in the Supplementary Appendix.\footnote{
\citet{tetenov2012statistical} studies minimax-regret treatment choice under a Gaussian sampling model with known variance.}
We take the empirical success rule as given and focus on the choice of the experimental sample size \(m\).

The expected welfare of policy \((m,\delta)\) over the experimental and rollout units is
\[
\begin{aligned}
W_P(m,\delta)
&:=
\mathbb{E}\left[
\sum_{i=1}^{m}\{T_iY_i(1)+(1-T_i)Y_i(0)\}
+
\sum_{i=m+1}^{N}
\bigl[
\delta(\omega_{m})Y_i(1)+\{1-\delta(\omega_{m})\}Y_i(0)
\bigr]
\right] \\
&=
\frac{m}{2}(\mu_1+\mu_0)
+
(N-m)\Bigl[
\mu_1\mathbb{E}\left[\delta(\omega_{m})\right]
+
\mu_0\{1-\mathbb{E}\left[\delta(\omega_{m})\right]\}
\Bigr].
\end{aligned}
\]
The first term is the expected welfare from the balanced experiment, while the second is the expected welfare from the rollout population. Here, \(\mathbb{E}\left[\delta(\omega_m)\right]\) is the ex ante probability that treatment 1 is selected for rollout.

An oracle that knows \(\mu_1\) and \(\mu_0\) assigns every unit to a welfare-maximizing arm, achieving welfare \(W_P^{\mathrm{oracle}} = N\max\{\mu_1,\mu_0\}\). The finite-population regret of policy \((m,\delta)\) relative to this oracle is defined as
\[
\mathrm{Reg}_P(m,\delta)
:=
W_P^{\mathrm{oracle}}-W_P(m,\delta).
\]

Let \(e_P(m,\delta)\) denote the probability of an inferior rollout, that is, of rolling out the inferior arm:
\[
e_P(m,\delta)
:=
\begin{cases}
1-\mathbb{E}\left[\delta(\omega_m)\right] & \tau>0\\[3pt]
\frac12 & \tau=0\\[3pt]
\mathbb{E}\left[\delta(\omega_m)\right] & \tau<0
\end{cases},
\]
where the value $1/2$ at $\tau=0$ is a convenient convention. Substituting the welfare expressions above gives
\begin{align}
\begin{aligned}
\mathrm{Reg}_P(m,\delta)
&=
m\left[
\max\{\mu_1,\mu_0\}
-\frac{\mu_1+\mu_0}{2}
\right] \\
&\quad
+(N-m)\left[
\max\{\mu_1,\mu_0\}
-\mu_1\mathbb{E}\left[\delta(\omega_m)\right]
-\mu_0\{1-\mathbb{E}\left[\delta(\omega_m)\right]\}
\right] \\
&=
\frac{m}{2}|\tau|
+
(N-m)|\tau|e_P(m,\delta).
\end{aligned}
\label{eq:1}
\end{align}
The first term, \(m|\tau|/2\), is the experimentation cost: under balanced randomization, half of the experimental units receive the inferior arm. The second term, \((N-m)|\tau|e_P(m,\delta)\), is the rollout risk from selecting the inferior arm for the remaining population. Increasing \(m\) therefore raises the experimentation cost but can reduce the rollout risk by lowering the probability of an inferior rollout.

For the analytical sample-size calculations, we work with the following Gaussian sampling model for the difference-in-means estimator.

\smallskip
\begin{assumption}[Gaussian sampling model]
\label{asm:2.1}
For every \(m\in(0,N)\),
\[
\widehat\tau_m
\sim
\mathcal N\left(
\tau,
\frac{2(\sigma_1^2+\sigma_0^2)}{m}
\right).
\]    
\end{assumption}
\smallskip

Under the iid finite-variance conditions above, the difference-in-means estimator is asymptotically normal, with the variance given above. Motivated by this result, we use the corresponding Gaussian approximation as a working finite-sample model. This allows us to focus on the fundamental trade-off between experimentation and rollout in sample-size choice.

Let \(\Phi\) and \(\varphi\) denote the standard normal cumulative distribution function and density, respectively. Under the empirical success rule, the inferior arm is selected if \(\widehat\tau_m < 0\) when \(\tau > 0\), or if \(\widehat\tau_m>0\) when \(\tau<0\). Therefore, the probability of an inferior rollout is
\[
e(m,\tau,\sigma)
:=
\Phi\left(
-
\frac{|\tau|}{\sqrt{2(\sigma_1^2+\sigma_0^2)/m}}
\right),
\]
which equals \(1/2\) at \(\tau=0\), consistent with the convention above. The corresponding regret is
\begin{align}
\mathrm{Reg}(m,\tau,\sigma)
:=
\frac{m}{2}|\tau|
+
(N-m)|\tau|e(m,\tau,\sigma).
\label{eq:2}
\end{align}

The regret criterion reflects both the welfare cost of assigning the inferior arm and the statistical difficulty of distinguishing the two arms. A larger \(|\tau|\) increases the welfare loss from assigning the inferior arm but also makes it easier to select the better arm, while larger outcome variances increase the probability of an inferior rollout. Henceforth, \(e\) and \(\mathrm{Reg}\) without the subscript \(P\) refer to the quantities under Assumption~\ref{asm:2.1} and the empirical success rule.

\subsection{Minimax Regret and Parameter-Space Restrictions}\label{sec:minimax-parameter-restrictions}

Minimax regret is a standard criterion in statistical decision theory and has been widely used in statistical treatment-choice problems \citep[e.g.,][]{manski2004statistical_treatment_rules, stoye2009minimax_regret_finite_samples}. This section examines the behavior of the sample-size choice under the minimax-regret criterion.

Using the regret function in \eqref{eq:2}, we first ask whether the minimax-regret criterion provides informative guidance for sample-size choice when either the treatment-effect magnitude or the noise scale is unrestricted.

\smallskip
\begin{proposition}[Direct minimax regret on unrestricted parameter spaces]
\label{prop:2.2}
Suppose Assumption~\ref{asm:2.1} holds and the rollout rule is the empirical success rule.
\begin{enumerate}
    \item[(i)] For any fixed \(\sigma\in(0,\infty)^2\), \(\sup_{\tau\in\mathbb{R}} \mathrm{Reg}(m,\tau,\sigma) = \infty\) for every \(m\in\mathcal{M}_N\).
    \item[(ii)]  For any \(0<\bar{\tau}<\infty\), \(\sup_{\substack{|\tau|\leq\bar{\tau}\\ \sigma\in(0,\infty)^2}} \mathrm{Reg}(m,\tau,\sigma) = \frac{N\bar{\tau}}{2}\) for every \(m\in\mathcal{M}_N\). Hence every feasible sample size is minimax-regret optimal over the parameter space in part (ii).
\end{enumerate}
\end{proposition}

\begin{proof}
For part (i), the experimentation-cost term gives \(\mathrm{Reg}(m,\tau,\sigma)\geq m|\tau|/2\), which diverges as \(|\tau|\to\infty\) for every feasible \(m>0\).

For part (ii), since \(e(m,\tau,\sigma)\leq1/2\), \(\mathrm{Reg}(m,\tau,\sigma)\leq N|\tau|/2\leq N\bar{\tau}/2\). Setting \(|\tau|=\bar{\tau}\) and letting \(\sigma_1^2+\sigma_0^2\to\infty\) makes \(e(m,\tau,\sigma)\to1/2\), so regret approaches this upper bound for every feasible \(m\). 
\end{proof}

Part~(i) shows that worst-case regret is infinite even when the arm variances are known, while part~(ii) shows that bounding the treatment effect alone makes worst-case regret finite but constant across sample sizes. Thus, direct minimax regret provides no informative guidance for sample-size choice without sufficiently strong parameter bounds.

We next illustrate the sensitivity of the minimax-regret sample size to an externally specified bound on the treatment effect. Suppose \(Y_i(1)\sim\mathcal{N}(\mu_1,4)\) and \(Y_i(0)\sim\mathcal{N}(\mu_0,16)\), so that the arm-specific standard deviations are known and given by \(\sigma=(2,4)\). We impose the bound \(|\tau|\leq\bar{\tau}\) on the treatment effect. Under the empirical success rule, define the corresponding minimax-regret sample size by
\[
m^{\mathrm{MR}}(\bar{\tau})
\in
\mathop{\mathrm{arg\,min}}_{m\in\mathcal{M}_N}
\max_{|\tau|\leq\bar{\tau}}
\mathrm{Reg}(m,\tau,(2,4)).
\]

Table~\ref{tab:1} reports the minimax-regret sample size for several values of \(\bar{\tau}\) and \(N\). With the arm-specific variances held fixed, the minimax-regret sample size varies substantially with the imposed upper bound on \(|\tau|\). For example, when \(N=500\), it decreases from \(100\) for \(\bar{\tau}=1\) to \(56\) for \(\bar{\tau}=3\) and to \(24\) for \(\bar{\tau}=10\).

\medskip
\begin{table}[H]
\centering
\caption{Minimax-regret sample sizes under bounded treatment effects.}
\label{tab:1}
\begin{tabular}{@{}lrrrr@{}}
\toprule
\(\bar\tau\) & \(N=200\) & \(N=500\) & \(N=1000\) & \(N=5000\) \\
\midrule
1 & 50 & 100 & 180 & 534 \\
3 & 30 & 56 & 86 & 244 \\
10 & 12 & 24 & 36 & 106 \\
\bottomrule
\end{tabular}
\begin{minipage}{0.7\linewidth}
\footnotesize
\smallskip
\textit{Note:} The table reports \(m^{\mathrm{MR}}(\bar\tau)\) in the Gaussian example with \(\sigma=(2,4)\).
\end{minipage}
\end{table}

These results highlight a general limitation of direct minimax regret in the one-shot test-and-roll setting. Without sufficiently restrictive parameter bounds, the criterion may fail to provide informative guidance for sample-size choice. Even with such bounds, the minimax-regret sample size can be sensitive to their values.

The next section therefore takes a different approach. Rather than selecting a single minimax-regret sample size, we characterize the admissible sample sizes over the unrestricted parameter space, without relying on externally specified parameter bounds.

\section{The Envelope Sample Size and Admissible Sample Sizes}\label{sec:envelope-admissible}

This section characterizes admissible sample sizes without relying on externally specified parameter bounds. To this end, we define the \emph{envelope sample size} as the supremum of parameter-specific regret-minimizing sample sizes. Our main result shows that the sharp envelope equals \(N/3\) and that the admissible sample sizes are exactly those in \((0,N/3]\). We further show that partial ex ante information about the signal-to-noise ratio can narrow the admissible set.

Section~\ref{sec:main-theorem} presents the main result and its implications for admissibility. Section~\ref{sec:envelope-proof} provides the proof. Section~\ref{sec:partial-signal-information} examines admissibility under partial signal-to-noise information.

\subsection{Main Theorem}\label{sec:main-theorem}

For each fixed parameter value \((\tau,\sigma)\), regret is minimized by a parameter-specific experimental sample size. Although this optimum depends on the unknown treatment effect and outcome variances, its supremum over all nonzero treatment effects and positive finite outcome standard deviations has a simple form that does not depend on these unknown parameters. We call this supremum the \emph{envelope sample size}. The theorem below shows that it equals \(N/3\).

\medskip
\begin{theorem}[Envelope sample size]
\label{thm:3.1}
Suppose the rollout rule is the empirical success rule \(\widehat{\delta}_m\) and Assumption~\ref{asm:2.1} holds. Under the continuous relaxation \(m\in(0,N)\), for every \(\tau\ne0\) and \(\sigma_0,\sigma_1\in(0,\infty)\), \(\mathrm{Reg}(m,\tau,\sigma)\) has a unique minimizer, denoted by \(m^*(\tau,\sigma) := \arg\min_{m\in(0,N)}\mathrm{Reg}(m,\tau,\sigma).\) Define the \emph{envelope sample size} by
\[
m^{\mathrm{Env}}
:=
\sup_{\tau\ne0,\;\sigma_0,\sigma_1\in(0,\infty)}
m^*(\tau,\sigma).
\]
Then
\[
m^{\mathrm{Env}}=\frac{N}{3}.
\]
\end{theorem}
\medskip

Theorem~\ref{thm:3.1} implies that the regret-minimizing experimental sample size never exceeds one-third of the target population, regardless of the treatment effect and outcome variances. However, this bound on parameter-specific optima does not by itself characterize admissibility. We therefore examine the admissibility of sample sizes \(m\in(0,N)\) by comparing their regret while holding the rollout rule fixed.

Given a parameter space, a sample size \(m'\) \emph{weakly dominates} \(m\) if its regret is no greater than that of \(m\) at every parameter value and strictly lower at some parameter value. A sample size \(m'\) \emph{strictly dominates} \(m\) if its regret is strictly lower at every parameter value. A sample size is \emph{admissible} if no other sample size weakly dominates it.

A set of sample sizes is a \emph{complete class} if every sample size outside the set is weakly dominated by some sample size in the set. A complete class is \emph{minimal} if no proper subset of it is a complete class.

The following corollary combines the envelope result with direct regret comparisons to characterize the admissible sample sizes exactly.

\medskip
\begin{corollary}[Admissible sample sizes]
\label{cor:3.2}
Under the conditions of Theorem~\ref{thm:3.1}, every \(m\in(N/3,N)\) is strictly dominated by \(N/3\). Moreover, the admissible sample sizes are exactly those in \((0,N/3]\), and this set is the unique minimal complete class.
\end{corollary}
\medskip

Corollary~\ref{cor:3.2} establishes that, under the Gaussian approximation and continuous relaxation, every sample size above \(N/3\) is strictly dominated by \(N/3\), while every sample size in \((0,N/3]\) is admissible. For example, \(m=N/2\), which assigns half of the population to the experiment, and \(m=2N/3\), which divides the population equally among treatment, control, and rollout, may appear natural from a balancing perspective. However, both choices exceed the envelope and are therefore strictly dominated by \(m=N/3\). This uniform dominance provides a simple guideline for experimental design, ruling out sample sizes above \(N/3\) without requiring a prior distribution or externally specified parameter bounds. As the largest admissible sample size, \(N/3\) serves as a conservative benchmark that guards against under-experimentation.

\smallskip
\begin{remark}[Discrete sample sizes]
\label{rem:3.3}
Theorem~\ref{thm:3.1} and Corollary~\ref{cor:3.2}
are established under the continuous relaxation $m\in(0,N)$.
If $N$ is divisible by six, then Corollary~\ref{cor:3.2} implies that the admissible
sample sizes over the feasible set $\mathcal M_N$ are
exactly those in
$\{m\in\mathcal M_N:m\le N/3\}$.
\end{remark}

\subsection{Proof}\label{sec:envelope-proof}

This section provides the proof of Theorem~\ref{thm:3.1}. The proof proceeds in two steps. We first characterize the envelope through a marginal cost-benefit comparison and then locate the resulting boundary under the Gaussian approximation. We conclude by briefly explaining how this boundary and direct regret comparisons establish the admissibility result in Corollary~\ref{cor:3.2}. Throughout, we maintain Assumption~\ref{asm:2.1}, fix the rollout rule to the empirical success rule \(\widehat\delta_m\), and work with the continuous relaxation \(m\in(0,N)\).

\paragraph{Cost-benefit representation.}

The finite-population regret decomposes as
\[
\mathrm{Reg}(m,\tau,\sigma)
=
\underbrace{\frac m2|\tau|}_{C(m,\tau)}
+
\underbrace{(N-m)|\tau|e(m,\tau,\sigma)}_{R(m,\tau,\sigma)}.
\]
Here \(C(m,\tau)\) is the experimentation cost and \(R(m,\tau,\sigma)\) is the rollout risk. 

For \(|\tau|>0\), the first-order condition for an interior regret minimizer can be written as
\[
\partial_m\mathrm{Reg}(m,\tau,\sigma)=0
\quad\Longleftrightarrow\quad
-\frac{\partial_m R(m,\tau,\sigma)}{\partial_m C(m,\tau)}=1.
\]
This condition equates the marginal benefit from additional experimentation with its marginal cost and motivates the following ratio.

\smallskip
\begin{definition}[Marginal cost-benefit ratio]\label{def:3.4}
For \(|\tau|>0\), the \emph{marginal cost-benefit ratio} is
\begin{align}
\mathrm{CBR}(m,\tau,\sigma)
:=
-\frac{\partial_m R(m,\tau,\sigma)}{\partial_m C(m,\tau)}.
\label{eq:3}
\end{align}
\end{definition}
\smallskip

If \(\mathrm{CBR}(m,\tau,\sigma)>1\), the marginal benefit from additional experimentation exceeds its marginal cost; if \(\mathrm{CBR}(m,\tau,\sigma)<1\), the opposite holds. Indeed, by the regret decomposition and Definition~\ref{def:3.4},
\begin{align}
\frac{\partial\mathrm{Reg}(m,\tau,\sigma)}{\partial m}
=
\frac{|\tau|}{2}\{1-\mathrm{CBR}(m,\tau,\sigma)\}.
\label{eq:4}
\end{align}

Under the Gaussian approximation, the ratio can be expressed in terms of the signal-to-noise ratio, defined as
\[
\rho:=\frac{|\tau|}{\sqrt{\sigma_1^2+\sigma_0^2}}>0.
\]
Since \(e(m,\tau,\sigma)=\Phi(-\sqrt{m/2}\,\rho)\), direct differentiation gives
\begin{align}
\mathrm{CBR}(m,\tau,\sigma)
=
2\Phi\Bigl(-\sqrt{\frac m2}\,\rho\Bigr)+\frac{N-m}{m}\sqrt{\frac m2}\,\rho\,\varphi\Bigl(\sqrt{\frac m2}\,\rho\Bigr).
\label{eq:5}
\end{align}
The ratio depends on \((\tau,\sigma)\) only through \(\rho\), so we also write it as \(\mathrm{CBR}(m,\rho)\).

For each fixed \(\rho>0\), differentiating \eqref{eq:5} gives
\begin{align}
\frac{\partial\mathrm{CBR}(m,\rho)}{\partial m}
=
-\frac{\sqrt{m/2}\,\rho\,\varphi(\sqrt{m/2}\,\rho)}{2m^2}
\left\{N+3m+\frac{(N-m)m\rho^2}{2}\right\}<0.
\label{eq:6}
\end{align}
Thus the ratio is strictly decreasing in the experimental sample size. Intuitively, the value of information from additional experimentation is diminishing. Moreover,
\[
\lim_{m\downarrow0}\mathrm{CBR}(m,\rho)=\infty,
\qquad
\lim_{m\uparrow N}\mathrm{CBR}(m,\rho)
=
2\Phi\left(-\sqrt{\frac N2}\,\rho\right)<1.
\]

Continuity and strict monotonicity therefore imply that the ratio crosses one exactly once. By \eqref{eq:4}, regret decreases before the crossing and increases after it. Hence the crossing is the unique regret-minimizing sample size \(m^*(\tau,\sigma)\). Since this minimizer depends on the parameters only through \(\rho\), we also write it as \(m^*(\rho)\).

By the strict monotonicity of \(\mathrm{CBR}(m,\rho)\) in \(m\) and the crossing at \(m^*(\rho)\),
\begin{align}
\mathrm{CBR}(m,\rho)\le1
\quad\Longleftrightarrow\quad
m\in[m^*(\rho),N).
\label{eq:7}
\end{align}
Thus the marginal comparison no longer justifies additional experimentation once the sample size reaches the parameter-specific optimum.

This single-crossing property yields a cost-benefit characterization of the envelope. As \(\rho\) spans \((0,\infty)\) over the parameter space, \(m^{\mathrm{Env}}=\sup_{\rho>0}m^*(\rho)\). The set of sample sizes at which no parameter value justifies additional experimentation is
\begin{align}
\begin{aligned}
\left\{
 m\in(0,N):\sup_{\rho>0}\mathrm{CBR}(m,\rho)\le1
\right\}
&=
\bigcap_{\rho>0}
\left\{m\in(0,N):\mathrm{CBR}(m,\rho)\le1\right\}\\
&=
\bigcap_{\rho>0}[m^*(\rho),N)\\
&=
\left[\sup_{\rho>0}m^*(\rho),N\right)
=[m^{\mathrm{Env}},N).
\end{aligned}
\label{eq:8}
\end{align}

Therefore,
\[
m^{\mathrm{Env}}= \min\left\{
 m\in(0,N):\sup_{\rho>0}\mathrm{CBR}(m,\rho)\le1
\right\}.
\]

We call the right-hand side the \emph{cost-benefit threshold}.

\smallskip
\begin{definition}[Cost-benefit threshold]\label{def:3.5}
The \emph{cost-benefit threshold} is
\begin{align}
m^{\mathrm{CBR}}
:=
\min\left\{
 m\in(0,N):\sup_{\rho>0}\mathrm{CBR}(m,\rho)\le1
\right\}.
\label{eq:9}
\end{align}
\end{definition}
\smallskip

For each candidate \(m\), Definition~\ref{def:3.5} takes the parameter supremum of the \emph{marginal cost-benefit ratio} rather than the supremum of regret. The threshold is the smallest \(m\) at which no parameter value justifies further experimentation.

\paragraph{The \(N/3\) boundary.}

We now characterize \(m^{\mathrm{CBR}}\) by fixing \(m\) and varying \(\rho\). Although the ratio in Definition~\ref{def:3.4} is \(0/0\) at \(\tau=0\), the Gaussian expression in \eqref{eq:5} admits the continuous extension
\[
\mathrm{CBR}(m,0):=\lim_{\rho\downarrow0}\mathrm{CBR}(m,\rho)=1,
\]
which we use to characterize its behavior at the zero-signal boundary.

Differentiating \eqref{eq:5} with respect to \(\rho\) gives
\begin{align}
\frac{\partial\mathrm{CBR}(m,\rho)}{\partial\rho}
=
\left[
\frac{N-m}{m}\left(1-\frac m2\rho^2\right)-2
\right]
\sqrt{\frac m2}\,
\varphi\left(\sqrt{\frac m2}\,\rho\right).
\label{eq:10}
\end{align}

If \(m\ge N/3\), then \((N-m)/m\le2\), so the derivative is strictly negative for every \(\rho>0\). The ratio therefore decreases from its continuous zero-signal value of one. Hence \(\mathrm{CBR}(m,\rho)<1\) for every \(\rho>0\), and \(\sup_{\rho>0}\mathrm{CBR}(m,\rho)=1\), approached as \(\rho\downarrow0\). If \(m<N/3\), the right derivative at zero is positive, so the ratio exceeds one for all sufficiently small \(\rho>0\).

Together, these two cases establish
\begin{align}
\sup_{\rho>0}\mathrm{CBR}(m,\rho)\le1
\quad\Longleftrightarrow\quad
m\ge\frac N3.
\label{eq:11}
\end{align}

Thus the set in \eqref{eq:8} is \([N/3,N)\), which has minimum \(N/3\). Consequently,
\[
m^{\mathrm{CBR}}=m^{\mathrm{Env}}=\frac N3,
\]
proving Theorem~\ref{thm:3.1}. \hfill\(\qed\)

\medskip

Figure~\ref{fig:2} illustrates the parameter-specific crossings and the resulting cost-benefit threshold. It highlights how the regret-minimizing sample size varies with the signal-to-noise ratio, while its upper envelope remains fixed at \(N/3\).

\begin{figure}[htbp]
\centering
\includegraphics[scale=0.65]{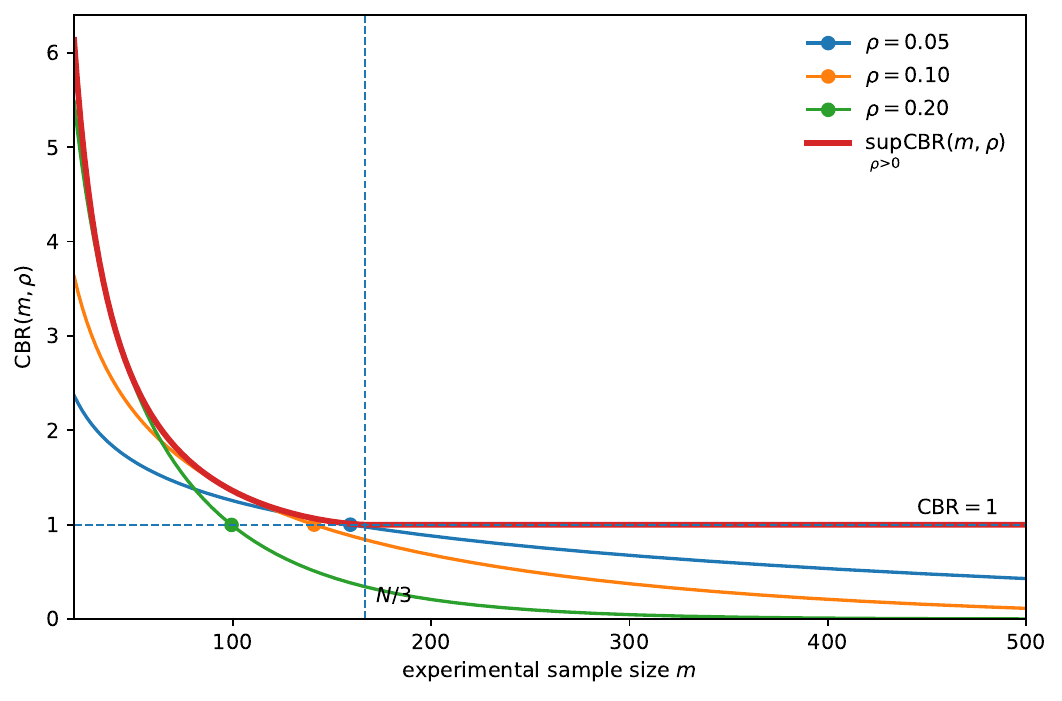}
\caption{Parameter-specific marginal cost-benefit ratios and the uniform cost-benefit boundary for \(N=500\). The thin curves correspond to signal-to-noise ratios $\rho=0.05$, $0.10$, and $0.20$. Their intersections with the horizontal line at one, marked by circles, give the corresponding parameter-specific regret-minimizing sample sizes \(m^*(\rho)\). The thick curve plots \(\sup_{\rho>0}\mathrm{CBR}(m,\rho)\), and the vertical line marks \(m^{\mathrm{Env}}=m^{\mathrm{CBR}}=N/3\). For \(m\ge N/3\), the supremum equals one, although \(\mathrm{CBR}(m,\rho)<1\) for every \(\rho>0\); the supremum is approached as \(\rho\downarrow0\).}
\label{fig:2}
\end{figure}

The proof of Corollary~\ref{cor:3.2}, given in Appendix~\ref{sec:appendix-A}, combines the envelope result with direct regret comparisons to characterize admissible sample sizes. Every sample size above the envelope is strictly dominated by the envelope, while sufficiently small and large signal-to-noise ratios yield opposite regret rankings for any two distinct sample sizes at or below the envelope. These comparisons also rule out dominance of any sample size below the envelope by a sample size above it, establishing the exact characterization of admissible sample sizes.

\subsection{Admissibility under Partial Signal-to-Noise Information}\label{sec:partial-signal-information}

The \(N/3\) envelope provides a parameter-free upper bound on admissible sample sizes. This robustness comes at the cost of leaving a relatively broad set of admissible choices. In practice, however, researchers may have partial ex ante information about the plausible range of the signal-to-noise ratio, even without knowing the treatment effect and outcome variances separately. We now show that such information can substantially narrow the admissible set.

The discussion in Section~\ref{sec:envelope-proof} shows that the parameter-specific regret-minimizing sample size \(m^{\ast}(\tau,\sigma)\) depends only on the signal-to-noise ratio \(\rho\). The following result shows that this optimum is decreasing in \(\rho\) and characterizes how bounds on \(\rho\) narrow the admissible set.

\medskip
\begin{corollary}[Admissible sample sizes under signal-to-noise restrictions]
\label{cor:3.6}
Under the conditions of Theorem~\ref{thm:3.1}, \(m^*(\rho)\) is continuous and strictly decreasing on \((0,\infty)\). Suppose \(0<\underline{\rho}\le\overline{\rho}<\infty\) and restrict attention to \(\rho\in[\underline{\rho},\overline{\rho}]\). Then
\[
0<m^*(\overline{\rho})
\le m^*(\rho)
\le m^*(\underline{\rho})
<\frac N3.
\]
Every \(m<m^*(\overline{\rho})\) is strictly dominated by \(m^*(\overline{\rho})\), whereas every \(m>m^*(\underline{\rho})\) is strictly dominated by \(m^*(\underline{\rho})\). The interval \([m^*(\overline{\rho}),\,m^*(\underline{\rho})]\) is exactly the set of admissible continuously relaxed sample sizes and hence is the unique minimal complete class over the restricted parameter space.
\end{corollary}
\medskip

A positive lower bound on \(\rho\) rules out excessively large experiments: if \(\rho\ge\underline{\rho}\), no parameter-specific optimum can exceed \(m^*(\underline{\rho})\). Conversely, an upper bound on \(\rho\) rules out excessively small experiments: if \(\rho\le\overline{\rho}\), no parameter-specific optimum can fall below \(m^*(\overline{\rho})\). Thus, two-sided information \(\rho\in[\underline{\rho},\overline{\rho}]\) shrinks the set of admissible sample sizes from \((0,N/3]\) to \([m^*(\overline{\rho}),\,m^*(\underline{\rho})]\). Although such information need not determine a unique sample size, it provides more informative guidance by ruling out both excessively small and excessively large experiments.

Table~\ref{tab:2} illustrates the quantitative magnitude of these restrictions by reporting \(m^*(\rho)\) for several values of \(\rho\) and \(N\). Two patterns are apparent. First, for a fixed \(N\), the regret-minimizing sample size \(m^*(\rho)\) decreases substantially as \(\rho\) increases and can be well below the \(N/3\) envelope. Second, for a fixed \(\rho\), \(m^*(\rho)\) increases with \(N\) in absolute terms, while its population share, \(m^*(\rho)/N\), decreases.

The patterns also illustrate how the extent to which bounds on \(\rho\) narrow the admissible interval varies with population size. For smaller populations, an upper bound on \(\rho\) mainly tightens the lower end of the admissible interval, whereas for larger populations, a positive lower bound on \(\rho\) can substantially reduce the upper end below \(N/3\). 
For example, when \(N=20{,}000\) and \(\rho\in[0.05,0.20]\), the admissible interval shrinks from approximately \((0,6667]\) to \([402,2622]\), corresponding to 2.0\%--13.1\% of the population.

\medskip
\begin{table}[H]
\centering
\caption{Parameter-specific regret-minimizing sample sizes by signal-to-noise ratio.}
\label{tab:2}
\begin{tabular*}{0.9\linewidth}
{@{\extracolsep{\fill}}c*{4}{c@{\,}c}@{}}
\toprule
\(\rho\)
& \multicolumn{2}{c}{\(N=500\)}
& \multicolumn{2}{c}{\(N=1000\)}
& \multicolumn{2}{c}{\(N=5000\)}
& \multicolumn{2}{c}{\(N=20000\)} \\
\midrule
0.02 & 165 & (33.0\%) & 328 & (32.8\%) & 1553 & (31.1\%) & 5184 & (25.9\%) \\
0.05 & 159 & (31.8\%) & 305 & (30.5\%) & 1159 & (23.2\%) & 2622 & (13.1\%) \\
0.10 & 141 & (28.2\%) & 246 & (24.6\%) & 655 & (13.1\%) & 1115 & (5.6\%) \\
0.20 & 99 & (19.8\%) & 147 & (14.7\%) & 279 & (5.6\%) & 402 & (2.0\%) \\
0.50 & 38 & (7.6\%) & 48 & (4.8\%) & 71 & (1.4\%) & 91 & (0.5\%) \\
\bottomrule
\end{tabular*}
\begin{minipage}{0.9\linewidth}
\footnotesize
\smallskip
\textit{Note:} Each entry reports the parameter-specific regret-minimizing sample size \(m^*(\rho)\), rounded to the nearest integer. The corresponding percentage of the population assigned to the experiment, calculated using the rounded sample size, is reported in parentheses.
\end{minipage}
\end{table}

\section{Simulation Study}\label{sec:simulation-study}

This section evaluates the envelope sample size against alternative sample sizes using simulations calibrated to data from the randomized Job Training Partnership Act (JTPA) experiment. The sample consists of 9,872 individuals, of whom 6,620 were randomly assigned access to job training and 3,252 were assigned to the control group. See \citet{bloom1997benefits} for a detailed description of the experiment. Following \citet{kitagawa2018who_should_be_treated}, we define the welfare outcome as 30-month earnings net of the \$774 treatment-assignment cost.\footnote{This is the average cost associated with treatment assignment, estimated from Table 5 of \citet{bloom1997benefits}.} The resulting treatment-control difference in sample means is \$356.46, which we treat as the true treatment effect \(\tau\) in the simulation data-generating process. The sample standard deviations of the welfare outcome are \$16,194.56 and \$17,013.89 for the control and treatment-assignment groups, respectively. These values imply a signal-to-noise ratio of \(\rho\approx0.0152\).

In each Monte Carlo replication, experimental and rollout outcomes are drawn with replacement separately from the observed control and treatment-assignment groups in the JTPA data. We consider target population sizes \(N\in\{1000,5000,20000\}\). The JTPA data are used to calibrate the outcome distributions, while \(N\) denotes the target-population size in the test-and-roll decision.

For each \(N\), we compare four sample-size policies. The \emph{Envelope} policy uses the largest feasible even sample size not exceeding \(N/3\), and the \emph{Half} and \emph{Thirds} policies use the feasible even sample sizes closest to \(N/2\) and \(2N/3\), respectively. The infeasible \emph{Gaussian plug-in} policy uses the calibrated signal-to-noise ratio \(\rho\) and chooses the corresponding regret-minimizing sample size \(m^*(\rho)\), adjusted to the feasible even grid. All four policies use the same empirical success rule.

Table~\ref{tab:3} reports results from 10,000 paired Monte Carlo replications. Normalized regret is defined as \(100\times\widehat{\mathrm{Reg}}(m)/(N|\tau|)\), where \(\widehat{\mathrm{Reg}}(m)\) denotes the Monte Carlo estimate of expected regret. The regret difference relative to the Gaussian plug-in policy is defined as \(100\times\bigl(\widehat{\mathrm{Reg}}(m)-\widehat{\mathrm{Reg}}(m^*(\rho))\bigr)/(N|\tau|)\) and is expressed in percentage points (pp). The final column reports the probability of an inferior rollout.

\begin{table}[htbp]
\centering
\caption{JTPA empirical-bootstrap test-and-roll simulation.}
\label{tab:3}
\setlength{\tabcolsep}{2.5pt}
\begin{tabular*}{\linewidth}{@{\extracolsep{\fill}}lcccc@{}}
\toprule
Policy & \(m\) (share) & Normalized regret & Regret diff. (pp) & Inferior rollout \\
\midrule
\multicolumn{5}{l}{\(N=1000\)} \\
\emph{Envelope} (\(N/3\)) & 332 (33.2\%) & 42.9\% & 0.1 & 42.4\% \\
\emph{Half} (\(N/2\)) & 500 (50.0\%) & 45.9\% & 3.0 & 41.3\% \\
\emph{Thirds} (\(2N/3\)) & 666 (66.6\%) & 47.9\% & 5.0 & 39.7\% \\
\emph{Gaussian plug-in} & 330 (33.0\%) & 42.8\% & 0.0 & 42.7\% \\
\addlinespace
\multicolumn{5}{l}{\(N=5000\)} \\
\emph{Envelope} (\(N/3\)) & 1666 (33.3\%) & 38.0\% & 0.0 & 32.6\% \\
\emph{Half} (\(N/2\)) & 2500 (50.0\%) & 39.8\% & 1.8 & 29.7\% \\
\emph{Thirds} (\(2N/3\)) & 3334 (66.7\%) & 41.8\% & 3.9 & 27.2\% \\
\emph{Gaussian plug-in} & 1598 (32.0\%) & 38.0\% & 0.0 & 32.8\% \\
\addlinespace
\multicolumn{5}{l}{\(N=20000\)} \\
\emph{Envelope} (\(N/3\)) & 6666 (33.3\%) & 28.2\% & -0.1 & 18.7\% \\
\emph{Half} (\(N/2\)) & 10000 (50.0\%) & 31.9\% & 3.5 & 14.5\% \\
\emph{Thirds} (\(2N/3\)) & 13334 (66.7\%) & 36.1\% & 7.8 & 10.6\% \\
\emph{Gaussian plug-in} & 5712 (28.6\%) & 28.4\% & 0.0 & 20.6\% \\
\bottomrule
\end{tabular*}
\par\smallskip
\begin{minipage}{\linewidth}
\footnotesize
\emph{Note:} Results are based on 10,000 paired Monte Carlo replications.
Expected regret is estimated as the difference between the oracle's expected welfare and average realized policy welfare across replications.
The Gaussian plug-in policy minimizes regret under the Gaussian approximation at the calibrated signal-to-noise ratio, but not necessarily under the empirical-bootstrap distribution.
\end{minipage}
\end{table}

For each \(N\), the Half and Thirds sample sizes reduce the probability of an inferior rollout relative to the Envelope sample size, but their additional experimentation leads to higher regret, consistent with Corollary~\ref{cor:3.2}. The Envelope sample size achieves regret close to the Gaussian plug-in benchmark across the population sizes considered.
\medskip

\section{Admissible Sample Sizes under Covariate Targeting}
\label{sec:covariate-targeting}

This section extends the analysis in Section~\ref{sec:envelope-admissible} to covariate targeting, in which rollout treatment assignments depend on observed covariates rather than being uniform across the remaining units. We characterize the envelope and admissible sample sizes under a plug-in targeting rule. We show that the envelope depends on the convergence rate of the CATE estimator and that the admissible sample sizes are exactly those at or below the envelope.

Data-driven treatment targeting based on observed covariates has been studied extensively \citep{HiranoPorter2009,kitagawa2018who_should_be_treated,athey2021policy_learning}. In the binary-treatment setting, the optimal targeting rule assigns treatment to individuals with a positive CATE, \(\mathbb{E}[Y(1)-Y(0)\mid X]\), and control otherwise. Since the CATE is unknown, a natural plug-in rule uses an estimate obtained from the experimental data to determine rollout assignments. We take this estimation procedure and the resulting targeting rule as given ex ante, focusing instead on the choice of experimental sample size.

Let \(X\in\mathbb{R}^p\) denote the vector of observed covariates. For \(t\in\{0,1\}\), define the conditional mean outcome as \(\mu_t(x)=\mathbb{E}[Y(t)\mid X=x]\), and let \(\sigma_t(x)<\infty\) denote the corresponding conditional standard deviation. The CATE is \(\tau(x)=\mu_1(x)-\mu_0(x)\). Let \(f\) denote the density or probability mass function of the known distribution of \(X\).\footnote{When $X$ is discrete, integrals are understood as sums over its support.} As in the baseline setting, \(N\) denotes the total population size, and the design variable is the total experimental sample size \(m\), rather than a covariate-specific sample size.

After observing the experimental data, the decision-maker constructs a CATE estimator \(\widehat{\tau}_m(x)=\widehat{\mu}_1(x)-\widehat{\mu}_0(x)\) and uses the plug-in targeting rule \(\widehat{\delta}_m(x)=\mathbf{1}\{\widehat{\tau}_m(x)>0\}\) to determine rollout assignments. Let \(v(x)>0\) denote a scale parameter governing the estimation uncertainty of \(\widehat{\tau}_m(x)\), which may depend on \(\sigma_0(x)\), \(\sigma_1(x)\), and the estimation procedure. We impose the following Gaussian estimation model with convergence-rate exponent \(\alpha>0\).

\medskip
\begin{assumption}[Gaussian CATE estimation model]
\label{asm:5.1}
For \(f\)-almost every \(x\) and every \(m\in(0,N)\),
\[
\widehat{\tau}_m(x)\mid \tau(x),v(x)
\sim
\mathcal N\left(\tau(x),v^2(x)m^{-2\alpha}\right).
\]
\end{assumption}
\medskip

In this model, \(\alpha\) governs the convergence rate of the CATE estimator. For correctly specified parametric models, the estimator typically achieves the root-\(m\) rate, corresponding to \(\alpha=1/2\). Nonparametric estimators may converge more slowly. For example, a kernel estimator with \(p\)-dimensional covariates and smoothness \(s\) has the standard root-MSE rate \(m^{-s/(2s+p)}\) under suitable regularity conditions, corresponding to \(\alpha=s/(2s+p)\). Since smoothing bias is generally non-negligible at the MSE-optimal bandwidth, the centered Gaussian model can be interpreted as an approximation for suitably bias-corrected estimators that are asymptotically normal at rate \(m^\alpha\).\footnote{Alternatively, undersmoothing can make the smoothing bias asymptotically negligible, but generally at the cost of a slower convergence rate than the MSE-optimal rate.}

Under Assumption~\ref{asm:5.1}, the probability that the plug-in targeting rule selects the inferior arm at \(x\) is
\[
e_x(m,\tau,v)
=
\Phi\left(-\frac{m^\alpha|\tau|}{v}\right).
\]

For fixed \(\tau(\cdot)\) and \(v(\cdot)\), the resulting covariate-integrated regret is
\[
\mathrm{Reg}_{\mathrm{cov}}(m,\tau,v)
=
\int
f(x)|\tau(x)|
\left[
\frac{m}{2}
+
\left( N - m \right)e_x(m,\tau(x),v(x))
\right]dx.
\]
The first term represents the welfare cost of experimentation arising from balanced random assignment, while the second represents the rollout risk from assigning the inferior treatment due to CATE estimation error.

The following proposition extends the envelope result to covariate targeting.

\medskip
\begin{proposition}[Covariate-integrated envelope]
\label{prop:5.2}
Under Assumption~\ref{asm:5.1}, let \(0<\alpha\le 1/2\). In the continuous relaxation \(m\in(0,N)\), consider parameter values \((\tau,v)\) consisting of measurable functions satisfying \(0<v(x)<\infty\) for \(f\)-almost every \(x\) and \(0<\int f(x)|\tau(x)|\,dx<\infty\). For each such parameter value, let \(m^*(\tau,v)\) denote a sample size that minimizes \(\mathrm{Reg}_{\mathrm{cov}}(m,\tau,v)\), and define the covariate-integrated \emph{envelope sample size} by
\[
m_{\mathrm{cov}}^{\mathrm{Env}}
:=
\sup_{\substack{\tau(\cdot),\,0<v(\cdot)<\infty\ f\text{-a.e.}\\
0<\int f(x)|\tau(x)|\,dx<\infty}}m^*(\tau,v).
\]

Then
\[
m_{\mathrm{cov}}^{\mathrm{Env}}
=
\frac{\alpha}{1+\alpha}N.
\]
\end{proposition}
\medskip

For a root-\(m\) CATE estimator, such as a regular parametric estimator, Proposition~\ref{prop:5.2} recovers the \(N/3\) envelope. By contrast, for a bias-corrected kernel CATE estimator with smoothness \(s\) and covariate dimension \(p\) achieving \(\alpha=s/(2s+p)\), the envelope becomes
\[
m_{\mathrm{kernel}}^{\mathrm{Env}}
=
\frac{s}{3s+p}N.
\]
Since \(p>0\), this envelope is strictly below \(N/3\), illustrating how slower CATE estimation leads to a smaller envelope sample size.

The following corollary builds on Proposition~\ref{prop:5.2} to provide a complete characterization of the admissible sample sizes under covariate targeting.

\medskip
\begin{corollary}[Admissible sample sizes under covariate targeting]
\label{cor:5.3}
Under the conditions of Proposition~\ref{prop:5.2}, every \(m\in(m_{\mathrm{cov}}^{\mathrm{Env}},N)\) is strictly dominated by \(m_{\mathrm{cov}}^{\mathrm{Env}}=\alpha N/(1+\alpha)\). Moreover, under the plug-in targeting rule, the admissible sample sizes are exactly those in \((0,m_{\mathrm{cov}}^{\mathrm{Env}}]\), and this set is the unique minimal complete class.
\end{corollary}
\medskip

Corollary~\ref{cor:5.3} establishes that slower convergence of the CATE estimator leads to a smaller set of admissible sample sizes, ruling out larger experiments through strict dominance.
The mechanism is that slower CATE estimation reduces the marginal benefit of additional experimentation: rollout risk declines more slowly with \(m\), while the welfare cost of experimentation continues to increase with \(m\). It is the convergence rate of CATE estimation, rather than the presence of covariates itself, that determines whether the envelope falls below \(N/3\).

\section{Multiple Treatment Arms}\label{sec:multiple-treatments}

This section extends the analysis in Section~\ref{sec:envelope-admissible} to balanced experiments with \(K\ge3\) treatment arms. We show that when variance heterogeneity is allowed, every experimental sample size \(m\in(0,N)\) is admissible, even when the permitted variance heterogeneity is arbitrarily small. Under common variance, by contrast, we show that a nontrivial uniform upper bound on admissible sample sizes is restored. The admissible sample sizes are exactly those at or below a fixed fraction of \(N\) determined solely by \(K\), and every larger sample size is strictly dominated.

Let \(\mu=(\mu_1,\ldots,\mu_K)\) denote the vector of mean potential outcomes, and define \(\mu_*:=\max_{1\le k\le K}\mu_k\). Let \(\Delta_k:=\mu_*-\mu_k\ge0\) denote the mean gap between arm \(k\) and the best arm. Thus, \(\Delta_k=0\) for any best arm and \(\Delta_k>0\) for any strictly suboptimal arm. The balanced experiment assigns \(m/K\) units to each arm. Hence, the average welfare loss per experimental unit is \(\bar\Delta:=K^{-1}\sum_{k=1}^K\Delta_k\). We restrict attention to \(\bar\Delta>0\). We impose the following Gaussian sampling model for the sample means.

\begin{assumption}[Gaussian vector of sample means]
\label{asm:6.1}
For every \(m>0\),
\[
(\widehat\mu_1,\ldots,\widehat\mu_K)^\top
\sim
\mathcal N_K\left(
\mu,
\frac Km\operatorname{diag}(\sigma_1^2,\ldots,\sigma_K^2)
\right).
\]
\end{assumption}

The heterogeneous-variance case permits \(\sigma=(\sigma_1,\ldots,\sigma_K)\in(0,\infty)^K\), whereas the common-variance case imposes \(\sigma_1=\cdots=\sigma_K>0\).
The empirical success rule selects \(\widehat k_m\in\mathop{\mathrm{arg\,max}}_{1\le k\le K}\widehat\mu_k\) and assigns the remaining \(N-m\) units to arm \(\widehat k_m\). Ties have probability zero under the Gaussian model.

Expected welfare and oracle welfare are
\[
W_K(m,\mu,\sigma)
=
\frac{m}{K}\sum_{k=1}^K\mu_k
+(N-m)\mathbb E[\mu_{\widehat k_m}],
\qquad
W^{\mathrm{oracle}}=N\mu_*.
\]
Because \(\mathbb E[\mu_{\widehat k_m}]=\sum_{j=1}^K\mu_j \Pr_{\mu,\sigma}(\widehat k_m=j)\), the associated regret is
\[
\begin{aligned}
\mathrm{Reg}_K(m,\mu,\sigma) :=W^{\mathrm{oracle}}-W_K(m,\mu,\sigma) =\underbrace{m\bar\Delta}_{C_K(m,\mu)}+\underbrace{(N-m)r_K(m,\mu,\sigma)}_{R_K(m,\mu,\sigma)},
\end{aligned}
\]
where
\[
r_K(m,\mu,\sigma)
:=
\mathbb E[\Delta_{\widehat k_m}]
=
\sum_{j=1}^K\Delta_j\Pr_{\mu,\sigma}(\widehat k_m=j).
\]
Here \(C_K(m,\mu)\) is the experimentation cost and \(R_K(m,\mu,\sigma)\) is the rollout risk. Because the balanced experiment assigns \(m/K\) units to each arm, the experimentation cost equals \((m/K)\sum_{k=1}^K\Delta_k=m\bar\Delta\). The rollout risk equals the size of the rollout population times the expected loss from the arm selected by the empirical success rule, where the loss from selecting arm \(j\) is its gap \(\Delta_j\) from the best arm.

As in Definition~\ref{def:3.4}, we define the \emph{marginal cost-benefit ratio} by
\[
\mathrm{CBR}_K(m,\mu,\sigma)
:=
-\frac{\partial_mR_K(m,\mu,\sigma)}{\partial_mC_K(m,\mu)}
=
\frac{r_K(m,\mu,\sigma)-(N-m)\partial_mr_K(m,\mu,\sigma)}{\bar\Delta},
\]
which measures the reduction in rollout risk from a marginal increase in the experimental sample size relative to the associated increase in experimentation cost.
We have \(\partial_m\mathrm{Reg}_K(m,\mu,\sigma)=\bar\Delta(1-\mathrm{CBR}_K(m,\mu,\sigma))\).

\subsection{Heterogeneous Variances}
\label{sec:heterogeneous-variances}

We first consider balanced experiments with \(K\ge3\) treatment arms and heterogeneous variances. We show that even arbitrarily small variance heterogeneity prevents the existence of a nontrivial uniform upper bound on admissible sample sizes.

For \(1\le C<\infty\), define the parameter space
\[
\Theta_K(C)
:=
\left\{
(\mu,\sigma)\in\mathbb R^K\times(0,\infty)^K:
\bar\Delta>0,\quad
\frac{\max_{1\le k\le K}\sigma_k^2}
     {\min_{1\le k\le K}\sigma_k^2}
\le C
\right\}.
\]
The constant \(C\) bounds the degree of variance heterogeneity across treatment arms, with \(C=1\) corresponding to common variance. 
By allowing \(C>1\) arbitrarily close to one, we examine whether small departures from common variance affect the existence of a nontrivial uniform upper bound on admissible sample sizes.

As in Section~\ref{sec:envelope-admissible}, we examine the uniform cost-benefit condition and the envelope of parameter-specific regret-minimizing sample sizes. 
Because regret may not attain its minimum on \((0,N)\) for some parameter values, we define the envelope over \(\Theta_K(C)\) as
\[
m_K^{\mathrm{Env}}(C)
:=
\sup\left\{
m\in(0,N):
m\in\mathop{\mathrm{arg\,min}}_{0<t<N}
\mathrm{Reg}_K(t,\mu,\sigma)
\text{ for some }(\mu,\sigma)\in\Theta_K(C)
\right\}.
\]

The following theorem shows that, for any variance-ratio bound \(C>1\), no sample size \(m\in(0,N)\) satisfies the uniform cost-benefit condition, and the envelope sample size equals \(N\).

\medskip
\begin{theorem}[Heterogeneous variances]\label{fm:6.2}
Suppose Assumption~\ref{asm:6.1} holds and \(K\ge3\). Under the empirical success rule, for every fixed \(C>1\) and every \(m\in(0,N)\),
\[
\sup_{(\mu,\sigma)\in\Theta_K(C)}
\mathrm{CBR}_K(m,\mu,\sigma)>1.
\]
Moreover, the envelope sample size over \(\Theta_K(C)\) satisfies
$
m_K^{\mathrm{Env}}(C)=N.
$
\end{theorem}
\medskip

The following corollary characterizes admissibility over the same parameter space.

\medskip
\begin{corollary}[Admissibility under heterogeneous variances]
\label{fm:6.3}
Under the conditions of Theorem~\ref{fm:6.2}, for every \(C>1\), every \(m\in(0,N)\) is admissible over \(\Theta_K(C)\).
\end{corollary}
\medskip

Theorem~\ref{fm:6.2} and Corollary~\ref{fm:6.3} show that, when variance heterogeneity is allowed, no sample size below \(N\) satisfies the uniform cost-benefit condition, the envelope equals \(N\), and every sample size is admissible. Importantly, restricting the variance ratio to any \(C>1\), however close to one, does not restore a nontrivial uniform upper bound on admissible sample sizes.

The key difference from the binary case lies in how variance heterogeneity affects treatment selection. When \(K=2\) and the two arms have equal means, the empirical success rule selects each arm with probability \(1/2\), regardless of their variances, because the difference between the sample means is symmetric about zero. For \(K\ge3\), an arm is selected only if its sample mean exceeds those of all other arms:
\[
\{\widehat k_m=j\}
=
\bigcap_{k\ne j}\{\widehat\mu_j-\widehat\mu_k\ge0\}.
\]
Suppose all arms have equal means, but one arm has a larger variance than the others, which share a common variance. As its variance grows relative to the others, its sample mean becomes increasingly likely to be either the largest or the smallest. Its selection probability therefore approaches \(1/2\), exceeding the equal-variance benchmark \(1/K\). 

This selection asymmetry explains why the uniform cost-benefit condition fails. Suppose the higher-variance arm has a slightly lower mean than the other arms, with a gap of \(\varepsilon>0\). The marginal experimentation cost is then \(\bar\Delta=\varepsilon/K\), while the expected rollout loss per unit approaches \(\varepsilon p_0\), where \(p_0>1/K\) is the selection probability of the higher-variance arm when all means are equal. For every fixed \(m\in(0,N)\), the marginal cost-benefit ratio therefore satisfies
\[
\lim_{\varepsilon\downarrow0}
\mathrm{CBR}_K(m,\mu,\sigma)
=
Kp_0>1.
\]

The same selection asymmetry also explains why every sample size is admissible. Consider any two sample sizes \(0<m_1<m_2<N\). When the higher-variance arm is slightly inferior, the larger sample size \(m_2\) yields lower regret than \(m_1\). Conversely, when that arm is slightly superior, the smaller sample size \(m_1\) yields lower regret. Thus, the regret ranking of any two sample sizes can be reversed by changing the treatment means, so neither sample size uniformly dominates the other.

\medskip
\begin{remark}[Bounded variances and mean gaps]\label{rem:bounded_variances}
The conclusions of Theorem~\ref{fm:6.2} and Corollary~\ref{fm:6.3} remain valid if the parameter space is further restricted by
\[
0<\bar\Delta\le B,
\qquad
\underline{\sigma}\le\sigma_k\le\overline{\sigma}
\quad (k=1,\ldots,K),
\]
for any \(B>0\) and \(0<\underline{\sigma}<\overline{\sigma}<\infty\). Hence, neither unbounded variances nor arbitrarily large mean gaps are required for the negative results. The parameter space must, however, allow arbitrarily small positive mean gaps and some departure from common variance. The proof is provided in Appendix~\ref{sec:appendix-E}.
\end{remark}

\subsection{Common Variance}
\label{sec:common-variance}

The preceding results show that even arbitrarily small variance heterogeneity precludes a nontrivial uniform upper bound on admissible sample sizes when \(K\ge3\). A common-variance model provides a useful benchmark for experiments in which treatments are modeled as affecting mean outcomes without changing outcome variability. We therefore consider a common-variance model with \(K\ge2\), in which \(\sigma_k=\sigma>0\) for every \(k\), and show that this restriction restores a nontrivial envelope equal to the largest admissible sample size. In this subsection, \(\sigma\) denotes the common standard deviation rather than the vector of arm-specific standard deviations used in Section~\ref{sec:heterogeneous-variances}.

Define the normalized gaps by \(d_k:=\Delta_k/\bar\Delta\) and write \(d=(d_1,\ldots,d_K)\). Since at least one arm is optimal and \(\bar\Delta=K^{-1}\sum_{k=1}^K\Delta_k\), we have
\[
d\in\mathcal D_K
:=
\left\{
d\in[0,\infty)^K:
\min_{1\le k\le K}d_k=0,\quad
\sum_{k=1}^K d_k=K
\right\}.
\]
The vector \(d\) describes the relative pattern of mean gaps across arms. Define the standardized signal at sample size \(m\) by
\[
\xi:=\frac{\bar\Delta\sqrt m}{\sigma\sqrt K}.
\]

For each fixed \(m>0\), let
$
Z_k:=\sqrt m(\widehat\mu_k-\mu_k)/(\sigma\sqrt K)
$
for $k=1,\ldots,K$.
Under Assumption~\ref{asm:6.1} and common variance, \(Z_1,\ldots,Z_K\) are independent standard normal variables, and
\[
\widehat\mu_k
=
\mu_*-\Delta_k+\sigma\sqrt{\frac Km}Z_k.
\]

Subtracting \(\mu_*\) and multiplying by the positive factor \(\sqrt m/(\sigma\sqrt K)\) leave the maximizing arm unchanged. Hence,
\[
\mathop{\mathrm{arg\,max}}_{1\le k\le K}\widehat\mu_k
=
\mathop{\mathrm{arg\,max}}_{1\le k\le K}
\{Z_k-\xi d_k\}.
\]
The decision of the empirical success rule therefore depends on the original parameters only through the normalized gap vector \(d\) and the standardized signal \(\xi\) at sample size \(m\).

Define the common-variance cost-benefit threshold by
\[
m_K^{\mathrm{CBR}} := \inf\left\{ m\in(0,N): \sup_{\substack{\mu\in\mathbb R^K,\,\sigma>0\\\bar\Delta>0}} \mathrm{CBR}_K(m,\mu,\sigma) \le1 \right\}.
\]
Thus, \(m_K^{\mathrm{CBR}}\) is the lower boundary of sample sizes at which the uniform marginal comparison no longer favors additional experimentation.

To characterize this boundary, it is useful to rewrite the marginal cost-benefit condition. Since \(\bar\Delta-r_K(m,\mu,\sigma)>0\), for \(m\in(0,N)\),
\begin{align}
\mathrm{CBR}_K(m,\mu,\sigma)\le1 \quad\Longleftrightarrow\quad \frac{-m\partial_m r_K(m,\mu,\sigma)} {\bar\Delta-r_K(m,\mu,\sigma)} \le \frac{m}{N-m}.
\label{eq:12}
\end{align}

Thus, verifying the uniform cost-benefit condition reduces to evaluating the parameter supremum of the normalized marginal ratio on the left-hand side of \eqref{eq:12}. To do so, define
\[
q_K(\xi,d) := \sum_{j=1}^K d_j \Pr\left( Z_j-\xi d_j = \max_{1\le k\le K}\{Z_k-\xi d_k\} \right).
\]
This is the expected normalized gap of the arm selected for rollout.

To interpret the normalized marginal ratio in \eqref{eq:12}, observe that \(\mathbb E[\mu_{\hat k_m}]-K^{-1}\sum_{k=1}^K\mu_k=(\mu_*-r_K)-(\mu_*-\bar\Delta)=\bar\Delta-r_K=\bar\Delta\{1-q_K(\xi,d)\}\). Thus, \(1-q_K(\xi,d)\) measures the expected gain per rollout unit relative to equal allocation across arms, normalized by \(\bar\Delta\). For fixed \((\mu,\sigma)\), both \(\bar\Delta\) and \(d\) remain constant as \(m\) varies. 
Since \(m\partial_m\xi=\xi/2\), we have
\[
\begin{aligned}
\frac{-m\,\partial_m r_K(m,\mu,\sigma)}
     {\bar\Delta-r_K(m,\mu,\sigma)}
&=
\frac{\partial\log\{\bar\Delta-r_K(m,\mu,\sigma)\}}
     {\partial\log m}
=
\frac{-(\xi/2)\partial_\xi q_K(\xi,d)}
     {1-q_K(\xi,d)}.
\end{aligned}
\]

The normalized marginal ratio in \eqref{eq:12} therefore has a natural interpretation as the sample-size elasticity of the expected per-unit gain from learning. It measures the percentage increase in the expected rollout gain, relative to equal allocation, induced locally by a 1\% increase in the experimental sample size. Thus, evaluating the uniform cost-benefit condition reduces to taking the supremum of this elasticity over \(\xi>0\) and \(d\in\mathcal D_K\).

The boundary \(m_K^{\mathrm{CBR}}\) is defined through a marginal comparison that holds uniformly over parameter values. As in Section~\ref{sec:envelope-admissible}, we also consider the parameter-specific regret-minimizing sample sizes. For a parameter value \((\mu,\sigma)\), a global minimizer is a point of \((0,N)\) at which \(m\mapsto\mathrm{Reg}_K(m,\mu,\sigma)\) attains its minimum over \((0,N)\). Define the common-variance envelope sample size by
\[
m_K^{\mathrm{Env}}
:=
\sup\left(
\bigcup_{\substack{\mu\in\mathbb R^K,\,\sigma>0\\
                  \bar\Delta>0}}
\mathop{\mathrm{arg\,min}}_{0<m<N}
\mathrm{Reg}_K(m,\mu,\sigma)
\right).
\]
Unlike the binary definition in Theorem~\ref{thm:3.1}, this definition does not presuppose a unique minimizer.

The following theorem characterizes the cost-benefit threshold through the supremum of the normalized marginal ratio and shows that it coincides with the envelope sample size.

\medskip
\begin{theorem}[Common variance with multiple treatments]
\label{fm:6.4}
Suppose Assumption~\ref{asm:6.1} holds, \(K\ge2\), and \(\sigma_k=\sigma>0\) for every \(k\). Let the parameter space consist of all \((\mu,\sigma)\in\mathbb R^K\times(0,\infty)\) with \(\bar\Delta>0\), and fix the rollout rule to the empirical success rule. Define
\[
\beta_K
:=
\sup_{\substack{\xi>0\\d\in\mathcal D_K}}
\frac{-(\xi/2)\partial_\xi q_K(\xi,d)}
     {1-q_K(\xi,d)}.
\]
Then \(1/2\le\beta_K<\infty\), and
\[
m_K^{\mathrm{Env}}
=
m_K^{\mathrm{CBR}}
=
\frac{\beta_K}{1+\beta_K}N
\in
\left[\frac N3,N\right).
\]
Moreover, regret has at least one global minimizer on \((0,N)\) for every parameter value \((\mu,\sigma)\).

If \(K\ge3\), then \(\beta_K>1/2\), and there exists a parameter value under which \(m_K^{\mathrm{CBR}}\) is the unique global minimizer on \((0,N)\). If \(K=2\), then \(\beta_2=1/2\) and \(m_2^{\mathrm{Env}}=m_2^{\mathrm{CBR}}=N/3\), in agreement with Theorem~\ref{thm:3.1}.
\end{theorem}
\medskip

Building on Theorem~\ref{fm:6.4}, the following corollary characterizes the admissible sample sizes under common variance.

\medskip
\begin{corollary}[Admissibility under common variance]
\label{cor:6.5}
Under the conditions of Theorem~\ref{fm:6.4}, every \(m\in(m_K^{\mathrm{CBR}},N)\) is strictly dominated by \(m_K^{\mathrm{CBR}}\). Moreover, the admissible sample sizes are exactly those in \((0,m_K^{\mathrm{CBR}}]\), and this set is the unique minimal complete class among deterministic sample-size choices.
\end{corollary}
\medskip

Theorem~\ref{fm:6.4} and Corollary~\ref{cor:6.5} extend the binary-treatment results to multiple arms under common variance. The envelope is a fixed fraction of the target population determined solely by \(K\) and is the largest admissible sample size, strictly dominating every larger sample size. For \(K=2\), this fraction is \(1/3\); for \(K\ge3\), it exceeds \(1/3\), and the envelope is attained as the unique global regret minimizer at some parameter value.

For any fixed \(K\), \(\beta_K\) can be approximated numerically by optimizing over \((\xi,d)\in(0,\infty)\times\mathcal D_K\). Table~\ref{tab:4} reports the corresponding envelope shares \(m_K^{\mathrm{CBR}}/N\). The numerical results indicate that this share increases gradually with the number of treatments, from \(1/3\) for \(K=2\) to approximately \(0.40\) for \(K=10\). Thus, the admissible sample-size boundary changes only moderately with the number of treatment arms over the range considered.

\begin{table}[H]
\centering
\caption{Common-variance envelope ratios \(m_K^{\mathrm{CBR}}/N\).}
\label{tab:4}
\begin{tabular}{@{}lrrrrrrrrr@{}}
\toprule
\(K\) & 2 & 3 & 4 & 5 & 6 & 7 & 8 & 9 & 10 \\
\midrule
\(m_K^{\mathrm{CBR}}/N\) & 0.333 & 0.341 & 0.351 & 0.361 & 0.370 & 0.378 & 0.385 & 0.392 & 0.397 \\
\bottomrule
\end{tabular}
\end{table}

\section{Conclusion}\label{sec:conclusion}

This paper studies experimental sample-size choice in finite-population test-and-roll experiments, where the decision-maker balances the welfare cost of experimentation against the risk of rolling out an inferior treatment. Direct minimax regret is uninformative when either the treatment-effect magnitude or the noise scale is unrestricted, and even under externally specified parameter bounds, the resulting sample-size choice can depend materially on those bounds. We therefore derive a sharp characterization of admissible sample sizes under the empirical success rule. Under a Gaussian approximation and continuous relaxation, the binary-treatment admissible set is exactly $(0,N/3]$, with $N/3$ strictly dominating every larger sample size. This characterization requires neither a prior nor externally specified parameter bounds, while partial signal-to-noise information further narrows the admissible set.

The extensions show how admissibility depends on the statistical structure of the problem. Under covariate targeting with the plug-in targeting rule, the envelope is determined by the convergence rate of the CATE estimator. Root-$m$ estimation recovers the $N/3$ envelope, while slower convergence yields a smaller envelope. With three or more treatment arms, even arbitrarily small variance heterogeneity makes every sample size admissible. Under common variance, by contrast, a $K$-specific envelope in $(N/3,N)$ is restored, which strictly dominates every larger sample size and defines the upper boundary of the admissible set. Taken together, these results show how parameter information, learning rates, and the variance structure across arms shape welfare-based sample-size choice.

\appendix

\part*{Appendix}

\section{Proof of Corollary~\ref{cor:3.2}}\label{sec:appendix-A}

\begin{proof}

The proof has three steps. First, we show that every sample size above \(N/3\) is strictly dominated by \(N/3\). We then compare any two sample sizes that remain in the complete class and show that their regret ranking depends on the parameter value. Finally, we combine the two stages to show that the admissible sample sizes are exactly those in \((0,N/3]\).

Weak dominance and complete classes are as defined in Section~\ref{sec:main-theorem}, with the parameter space of Theorem~\ref{thm:3.1}.

(i) We first show that every sample size above \(N/3\) is strictly dominated by the boundary sample size.

Fix \(m\in(N/3,N)\) and \((|\tau|,\sigma)\in(0,\infty)^3\), and write \(\rho:=|\tau|/\sqrt{\sigma_1^2+\sigma_0^2}>0\). By \eqref{eq:11}, \(\sup_{\rho'>0}\mathrm{CBR}(N/3,\rho')\le1\), and hence \(\mathrm{CBR}(N/3,\rho)\le1\). By \eqref{eq:6}, \(\mathrm{CBR}(\cdot,\rho)\) is strictly decreasing on \((0,N)\). Therefore, for every \(u\in(N/3,m]\), \(\mathrm{CBR}(u,\rho)<\mathrm{CBR}(N/3,\rho)\le1.\)

Hence
\[
\begin{aligned}
\mathrm{Reg}(m,\tau,\sigma)-\mathrm{Reg}(N/3,\tau,\sigma)
&=
\int_{N/3}^{m}
\frac{\partial}{\partial u}\mathrm{Reg}(u,\tau,\sigma)\,du\\
&=
\frac{|\tau|}{2}
\int_{N/3}^{m}\{1-\mathrm{CBR}(u,\rho)\}\,du
>0.
\end{aligned}
\]

Since the parameter space of Theorem~\ref{thm:3.1} excludes \(\tau=0\), \(N/3\) strictly dominates every larger sample size, which proves the complete-class conclusion. It remains to determine whether that conclusion ranks the sample sizes inside the class.

(ii) We next show that the regret ranking of two sample sizes within the complete class depends on the parameter value.

Fix \(0<m_1<m_2\le N/3\). To compare their regrets exactly under Assumption~\ref{asm:2.1}, for \(\rho>0\)
define
\[
\begin{aligned}
D(\rho)
&:=
\frac{
\mathrm{Reg}(m_1,\tau,\sigma)-\mathrm{Reg}(m_2,\tau,\sigma)
}{|\tau|}\\
&=
\frac{m_1-m_2}{2}
+(N-m_1)\Phi\left(-\sqrt{\frac{m_1}{2}}\rho\right) -(N-m_2)\Phi\left(-\sqrt{\frac{m_2}{2}}\rho\right).
\end{aligned}
\]

The second equality follows from the Gaussian regret formula under Assumption~\ref{asm:2.1}. Thus \(D(\rho)>0\) means that \(m_2\) has strictly lower regret than \(m_1\), whereas \(D(\rho)<0\) means that \(m_1\) has strictly lower regret than \(m_2\). Define \(D(0)\) by the continuous extension of the displayed expression. Since \(\Phi(0)=1/2\),
\[
\begin{aligned}
D(0)
&=
\frac{m_1-m_2}{2}
+\frac{N-m_1}{2}
-\frac{N-m_2}{2}
=0.
\end{aligned}
\]

To establish the regret ranking for small positive values of \(\rho\), first differentiate the closed-form expression and evaluate the derivative at zero:
\[
D'(0)
=
\frac{\varphi(0)}{\sqrt{2}}
\left\{
(N-m_2)\sqrt{m_2}
-
(N-m_1)\sqrt{m_1}
\right\}.
\]

The coefficient in braces is strictly positive. Indeed, \(\frac{\partial}{\partial m}\{(N-m)\sqrt m\} = \frac{N-3m}{2\sqrt m},\)
which is positive on \((0,N/3)\). In particular, including the boundary case \(m_2=N/3\),
\[
\begin{aligned}
(N-m_2)\sqrt{m_2}-(N-m_1)\sqrt{m_1}
&=
\int_{m_1}^{m_2}\frac{N-3u}{2\sqrt u}\,du
>0.
\end{aligned}
\]

Consequently \(D'(0)>0\). The derivative is continuous at zero, so it remains positive on a sufficiently small right-neighborhood of zero. Because \(D(0)=0\), the mean value theorem then implies \(D(\rho)>0\) for all sufficiently small positive \(\rho\). Thus \(m_2\) has strictly lower regret than \(m_1\) at such parameter values.

For the opposite ranking, the exact Gaussian tail limit gives \(\lim_{\rho\to\infty}D(\rho) = \frac{m_1-m_2}{2} <0,\)
because both Gaussian tail probabilities in the closed-form expression converge to zero. Hence \(D(\rho)<0\) for all sufficiently large finite values of \(\rho\), and \(m_1\) has strictly lower regret at those parameter values.

(iii) We combine the comparisons to establish admissibility.

It remains to show that no sample size above \(N/3\) weakly dominates a sample size in \((0,N/3]\). Let \(m\in(0,N/3]\) and \(m'\in(N/3,N)\). The first stage gives \(\mathrm{Reg}(N/3,\tau,\sigma)<\mathrm{Reg}(m',\tau,\sigma)\) whenever \(\tau\ne0\). If \(m=N/3\), this strict inequality already shows that \(m'\) does not weakly dominate \(m\). If \(m<N/3\), the second stage with \(m_1=m\) and \(m_2=N/3\) gives a parameter value with \(\tau\ne0\) at which \(\mathrm{Reg}(m,\tau,\sigma)<\mathrm{Reg}(N/3,\tau,\sigma)\), and hence \(\mathrm{Reg}(m,\tau,\sigma)<\mathrm{Reg}(m',\tau,\sigma)\) at that parameter value. Together with the second stage, this shows that no other sample size weakly dominates any \(m\in(0,N/3]\), so every such \(m\) is admissible. By the first stage, every \(m'\in(N/3,N)\) is weakly dominated by \(N/3\) and is therefore not admissible. Hence the admissible sample sizes are exactly those in \((0,N/3]\), and this set is also a complete class. By Section~\ref{sec:main-theorem}, it is therefore the unique minimal complete class. \end{proof}

\section{Proof of Corollary~\ref{cor:3.6}}\label{sec:appendix-B}

\begin{proof}

Section~\ref{sec:envelope-proof} shows that, for each \(\rho>0\), \(\mathrm{CBR}(\cdot,\rho)\) is continuous and strictly decreasing in \(m\) and crosses one exactly once at \(m^*(\rho)\), and that regret is strictly decreasing before this crossing and strictly increasing after it; see \eqref{eq:4}--\eqref{eq:7}. The discussion leading to \eqref{eq:11} also gives \(\mathrm{CBR}(N/3,\rho)<1\) for every \(\rho>0\). Hence \(0<m^*(\rho)<N/3.\)

It remains to establish the dependence of \(m^*(\rho)\) on \(\rho\). Let \(F(m,\rho):=\mathrm{CBR}(m,\rho)-1.\) By \eqref{eq:6}, \(\partial_mF<0\). At \(m=m^*(\rho)\), write
\[
x:=\sqrt{\frac{m^*(\rho)}{2}}\,\rho,
\qquad
\lambda:=\frac{N-m^*(\rho)}{m^*(\rho)}.
\]

The crossing condition \(F(m^*(\rho),\rho)=0\) implies \(\lambda x\varphi(x)=2\Phi(x)-1.\) Using \eqref{eq:10}, \(\partial_\rho F(m^*(\rho),\rho)<0\). Indeed,
\[
\lambda(1-x^2)-2
=
-\frac{
2x\varphi(x)-\{2\Phi(x)-1\}(1-x^2)
}{
x\varphi(x)
}
<0,
\]
because the numerator is zero at \(x=0\) and has derivative \(2x\{2\Phi(x)-1\}>0\) for \(x>0\). The implicit function theorem therefore implies that \(m^*(\rho)\) is continuously differentiable and
\[
\frac{d m^*(\rho)}{d\rho}
=
-\frac{\partial_\rho F(m^*(\rho),\rho)}
{\partial_m F(m^*(\rho),\rho)}
<0.
\]

Thus, \(m^*(\rho)\) is continuous and strictly decreasing on \((0,\infty)\).

Now write \(a:=m^*(\overline{\rho}),\) and \(b:=m^*(\underline{\rho}).\) For every \(\rho\in[\underline{\rho},\overline{\rho}]\),
\[
0<a\le m^*(\rho)\le b<\frac N3.
\]

If \(m<a\), then \(m<a\le m^*(\rho)\) for every admissible \(\rho\). Since regret is strictly decreasing before \(m^*(\rho)\), \(a\) strictly dominates \(m\). Similarly, if \(m>b\), then \(m>b\ge m^*(\rho)\) for every admissible \(\rho\), and the strict increase of regret after \(m^*(\rho)\) implies that \(b\) strictly dominates \(m\).

Finally, fix \(m_0\in[a,b]\). By continuity and strict monotonicity, there is a unique \(\rho_0\in[\underline{\rho},\overline{\rho}]\) such that \(m_0=m^*(\rho_0)\). For example, the parameter value \(\tau=\rho_0\) and \(\sigma_0=\sigma_1=1/\sqrt{2}\) has signal-to-noise ratio \(\rho_0\). At this parameter value, \(m_0\) is the unique regret minimizer, so no other sample size can dominate it. Thus every \(m_0\in[a,b]\) is admissible.

Consequently, \([m^*(\overline{\rho}),\,m^*(\underline{\rho})]\) is exactly the set of admissible continuously relaxed sample sizes. Every sample size outside it is strictly dominated by one of its endpoints, so it is also a complete class and hence, by Section~\ref{sec:main-theorem}, the unique minimal complete class over the restricted parameter space. \end{proof}

\subsection*{Declaration of generative AI and AI-assisted technologies in the manuscript preparation process}

During the preparation of this work, the authors used Claude Code (Anthropic), Codex (OpenAI), and ChatGPT (OpenAI) to improve the language and terminology of the manuscript, check the consistency of notation and cross-references, assist in checking and reorganizing mathematical derivations and proofs, convert the manuscript into LaTeX, and generate and refine the illustration in Figure~\ref{fig:1}. After using these tools, the authors reviewed and edited the content as needed and take full responsibility for the content of the published article.

\bibliographystyle{ecta}
\bibliography{mybib_natbib}

\newpage

\part*{Supplementary Appendix}

\section{Minimax Optimality of the Empirical Success Rule}
\label{app:optimality-empirical-success-rule}

This appendix establishes minimax optimality of the empirical success rule under the Gaussian sampling model in Assumption~\ref{asm:2.1}. We maintain the iid potential-outcome setup and balanced randomization of Section~\ref{sec:test-and-roll-setup} and hold the experimental sample size fixed. We consider all data-generating distributions satisfying these restrictions and the parameter bounds specified below. These bounds are imposed solely for the minimax analysis and are not required for the admissible sample-size results in the main text.

\begin{proposition}[Minimax rollout rule]
\label{prop:fixed-m-empirical-success}
Fix $m\in\mathcal M_N$ and finite constants $\bar\tau,\bar v>0$.
Let $\mathcal P_G(\bar\tau,\bar v)$ be the class of all potential-outcome
laws $P$ in the setup of Section~\ref{sec:test-and-roll-setup} for which
the Gaussian sampling restriction in Assumption~\ref{asm:2.1} holds
at every feasible sample size and
\[
|\tau(P)|\le\bar\tau,
\qquad
\sigma_1^2(P),\sigma_0^2(P)>0,
\qquad
v(P):=\sigma_1^2(P)+\sigma_0^2(P)\le\bar v.
\]
Let $\mathcal D_m$ be the class of all measurable rollout rules
$\delta(\omega_m)\in[0,1]$ based on the full experimental data.
Then $\widehat\delta_m(\omega_m)=\mathbf 1\{\widehat\tau_m>0\}$
is minimax for finite-population regret:
\[
\widehat\delta_m
\in
\argmin_{\delta\in\mathcal D_m}
\sup_{P\in\mathcal P_G(\bar\tau,\bar v)}
\Reg_P(m,\delta).
\]
The minimax value is
\[
\max_{0\le t\le\bar\tau}
\left\{
\frac m2t+(N-m)t\Phi\left(-t\sqrt{\frac{m}{2\bar v}}\right)
\right\}.
\]
\end{proposition}

\begin{proof}
Write $\mathcal P_G=\mathcal P_G(\bar\tau,\bar v)$.
For any $P\in\mathcal P_G$, Assumption~\ref{asm:2.1} and the regret
decomposition in \eqref{eq:1} give
\[
\Reg_P(m,\widehat\delta_m)
=
\frac m2|\tau(P)|
+(N-m)|\tau(P)|
\Phi\left(-|\tau(P)|\sqrt{\frac{m}{2v(P)}}\right).
\]
For fixed $|\tau(P)|$, this expression is nondecreasing in $v(P)$.
Define
\[
B_m(t):=
\frac m2t+(N-m)t\Phi\left(-t\sqrt{\frac{m}{2\bar v}}\right),
\qquad 0\le t\le\bar\tau.
\]
Since $B_m$ is continuous on the compact interval $[0,\bar\tau]$, its maximum exists. For every $P\in\mathcal P_G$, the monotonicity of $\Phi$ and the restrictions $v(P)\le\bar v$ and $|\tau(P)|\le\bar\tau$ imply
\[
\Reg_P(m,\widehat\delta_m)
\le B_m\bigl(|\tau(P)|\bigr)
\le \max_{0\le t\le\bar\tau}B_m(t).
\]
Taking the supremum over $P\in\mathcal P_G$ yields
\[
\sup_{P\in\mathcal P_G}\Reg_P(m,\widehat\delta_m)
\le \max_{0\le t\le\bar\tau}B_m(t).
\]

For the lower bound, fix $t\in(0,\bar\tau]$ and consider two
potential-outcome laws $P_t^+$ and $P_t^-$ under which the coordinates
are independent and
\[
Y_i(1)\sim\mathcal N(\pm t/2,\bar v/2),
\qquad
Y_i(0)\sim\mathcal N(\mp t/2,\bar v/2),
\]
where the upper signs correspond to $P_t^+$.
Both laws belong to $\mathcal P_G$, as they satisfy the parameter bounds and Assumption~\ref{asm:2.1} under balanced randomization.

Let $f_t^+$ and $f_t^-$ denote the densities of the full experimental
data under these laws. 
Since the randomization mechanism is identical under both laws, the log-likelihood ratio is
\[
\begin{aligned}
\log\frac{f_t^+(\omega_m)}{f_t^-(\omega_m)}
&=
\frac{2t}{\bar v}
\left(
\sum_{i:T_i=1}Y_i-\sum_{i:T_i=0}Y_i
\right)=
\frac{mt}{\bar v}\widehat\tau_m.
\end{aligned}
\]

Under the prior assigning probability $1/2$ to each law, the Bayes regret of a rollout rule $\delta$ is
\[
\begin{aligned}
\bar R_t(\delta)
&:=
\frac12\Reg_{P_t^+}(m,\delta)
+\frac12\Reg_{P_t^-}(m,\delta)\\
&=
\frac{mt}{2}
+\frac{(N-m)t}{2}
\int
\left[
f_t^+(\omega_m)\{1-\delta(\omega_m)\}
+f_t^-(\omega_m)\delta(\omega_m)
\right]d\omega_m.
\end{aligned}
\]
Since the experimentation cost does not depend on $\delta$, minimizing Bayes regret amounts to minimizing the integrand pointwise. Thus, a Bayes rule selects arm 1 when $f_t^+(\omega_m)>f_t^-(\omega_m)$ and arm 0 when the reverse inequality holds. The likelihood ratio derived above implies that this rule coincides with the empirical success rule $\widehat\delta_m$. Since $t>0$, ties have probability zero. Moreover, the regret of $\widehat\delta_m$ equals $B_m(t)$ under either law.

Consequently, for every $\delta\in\mathcal D_m$,
\[
\begin{aligned}
\sup_{P\in\mathcal P_G}\Reg_P(m,\delta)
&\ge
\frac12\Reg_{P_t^+}(m,\delta)
+\frac12\Reg_{P_t^-}(m,\delta)\\
&\ge
\frac12\Reg_{P_t^+}(m,\widehat\delta_m)
+\frac12\Reg_{P_t^-}(m,\widehat\delta_m)\\
&=B_m(t).
\end{aligned}
\]
Taking the supremum over $t\in(0,\bar\tau]$ and using $B_m(0)=0$
gives a lower bound that matches the preceding upper bound.
This proves both claims.
\end{proof}

\section{Proofs of Proposition~\ref{prop:5.2} and Corollary~\ref{cor:5.3}}\label{sec:appendix-D}

\begin{proof}[Proof of Proposition~\ref{prop:5.2}]

The proof has four stages. First, we express the derivative of regret in terms of an integrated \emph{marginal cost-benefit ratio}. Second, for a fixed admissible pair \((\tau,v)\), we vary \(m\) and show that this ratio crosses one exactly once. Third, we fix \(m\) and vary the pointwise signal to obtain a bound that is uniform over admissible pairs. Finally, a parameter value with a constant signal shows that the uniform bound is sharp.

(i) We express the regret derivative through the integrated marginal cost-benefit ratio.

The first step identifies the quantity whose crossing determines the regret-minimizing sample size. Write \(\rho_{\mathrm{cov}}(x):=|\tau(x)|/v(x)\). Differentiating the experimentation and rollout terms gives the normalized \emph{marginal cost-benefit ratio}
\begin{align}
\mathrm{CBR}_{\mathrm{cov}}(m,\tau,v)
:=
\frac{
\int f(x)|\tau(x)|
\mathrm{CBR}_\alpha\{m,\rho_{\mathrm{cov}}(x)\}\,dx
}{
\int f(x)|\tau(x)|\,dx
},
\label{eq:14}
\end{align}
where
\[
\mathrm{CBR}_\alpha(m,\rho_{\mathrm{cov}})
=
2\Phi(-m^\alpha\rho_{\mathrm{cov}})
+2\alpha\frac{N-m}{m}m^\alpha\rho_{\mathrm{cov}}
\varphi(m^\alpha\rho_{\mathrm{cov}}).
\]

Indeed, the experimentation term has derivative \(\frac12\int f(x)|\tau(x)|\,dx,\)
whereas minus the derivative of the rollout term is \(\int f(x)|\tau(x)| \left[ \Phi\{-m^\alpha\rho_{\mathrm{cov}}(x)\} +\frac{\alpha(N-m)}m m^\alpha\rho_{\mathrm{cov}}(x) \varphi\{m^\alpha\rho_{\mathrm{cov}}(x)\} \right]dx.\)

On every compact subinterval of \((0,N)\), these derivatives are bounded by a constant times \(f(x)|\tau(x)|\), so differentiation under the integral is valid. It follows in particular that
\begin{align}
\partial_m\mathrm{Reg}_{\mathrm{cov}}(m,\tau,v)
=
\frac12\left\{\int f(x)|\tau(x)|\,dx\right\}
\{1-\mathrm{CBR}_{\mathrm{cov}}(m,\tau,v)\}.
\label{eq:15}
\end{align}

Thus regret decreases precisely when the integrated ratio exceeds one and increases when it is below one. The next stage uses this identity after establishing a unique crossing for each fixed \((\tau,v)\).

(ii) We show that the integrated ratio crosses one exactly once for each fixed admissible parameter pair.

Fix an admissible pair \((\tau,v)\) and vary \(m\). To show that its regret-minimizing sample size is unique, it is enough to show that the integrated ratio decreases continuously from above one to below one. For a fixed positive \(\rho_{\mathrm{cov}}\), set \(z=m^\alpha\rho_{\mathrm{cov}}\). Direct differentiation gives
\begin{align}
\partial_m\mathrm{CBR}_\alpha(m,\rho_{\mathrm{cov}})
=
\frac{2\alpha z\varphi(z)}m
\left[
-1-\frac Nm
+\frac{\alpha(N-m)}m(1-z^2)
\right]
<0.
\label{eq:16}
\end{align}

If \(1-z^2\le0\), the bracket is at most \(-1-N/m\). If \(1-z^2>0\), then \(0<\alpha\le1/2\) gives \(\alpha(N-m)(1-z^2)<N+m\), so the bracket is again negative. Dominated convergence makes the integrated ratio continuous in \(m\). Equation \eqref{eq:16} gives a strict pointwise ordering whenever the regret weight is positive. Integrating this ordering against weights with positive total mass shows that the integrated ratio is strictly decreasing in \(m\).

For $f$-almost every $x$ with positive regret weight, \(\mathrm{CBR}_\alpha\{m,\rho_{\mathrm{cov}}(x)\}\to\infty\) as \(m\downarrow0\). All terms are nonnegative, so Fatou's lemma applied after normalizing the weights in Equation \eqref{eq:14} yields \(\lim_{m\downarrow0}\mathrm{CBR}_{\mathrm{cov}}(m,\tau,v)=\infty.\)

As \(m\uparrow N\), dominated convergence gives
\[
\lim_{m\uparrow N}\mathrm{CBR}_{\mathrm{cov}}(m,\tau,v)
=
\frac{
\int f(x)|\tau(x)|
2\Phi\{-N^\alpha\rho_{\mathrm{cov}}(x)\}\,dx
}{
\int f(x)|\tau(x)|\,dx
}
<1.
\]

For $f$-almost every \(x\) with \(f(x)|\tau(x)|>0\), \(0<v(x)<\infty\) gives \(0<2\Phi\{-N^\alpha\rho_{\mathrm{cov}}(x)\}<1.\)

Consequently,
\[
\int f(x)|\tau(x)|2\Phi\{-N^\alpha\rho_{\mathrm{cov}}(x)\}\,dx
<
\int f(x)|\tau(x)|\,dx.
\]

The denominator is positive, so this weighted-integral inequality gives the strict inequality in the displayed limit. Thus the integrated ratio crosses one exactly once. By Equation \eqref{eq:15}, that crossing is the unique global minimizer \(m^*(\tau,v)\). This establishes the parameter-specific minimization step. To bound these crossings uniformly, the next stage fixes a candidate value of \(m\) and asks whether the pointwise ratio can exceed one at any signal value.

(iii) We derive an upper bound on the crossing that is uniform over admissible parameter pairs.

Equation \eqref{eq:16} varies \(m\) while holding the positive signal \(\rho_{\mathrm{cov}}\) fixed. For the uniform boundary, instead fix \(m\) and vary the pointwise signal. For \(u\in(0,N)\), write \(b(u):=\alpha(N-u)/u\), and set \(r:=m^\alpha\rho_{\mathrm{cov}}\). The pointwise ratio introduced in Equation \eqref{eq:14} satisfies \(\mathrm{CBR}_\alpha(m,\rho_{\mathrm{cov}}) =2\Phi(-r)+2b(m)r\varphi(r)\).

If \(m\ge\alpha N/(1+\alpha)\), then \(\alpha(N-m)\le m \Longleftrightarrow b(m)\le1.\)
Moreover,
\[
\begin{aligned}
\frac{\partial}{\partial r}
\{2\Phi(-r)+2b(m)r\varphi(r)\}
&=-2\varphi(r)+2b(m)\{\varphi(r)-r^2\varphi(r)\}\\
&=2\{b(m)-1-b(m)r^2\}\varphi(r)\le0.
\end{aligned}
\]

Its value at \(r=0\) is one, so \(\mathrm{CBR}_\alpha(m,\rho_{\mathrm{cov}})\le1\) pointwise. Equation \eqref{eq:14} is a weighted average of these pointwise ratios, with nonnegative weights, and hence \(\mathrm{CBR}_{\mathrm{cov}}(m,\tau,v)\le1\). Equation \eqref{eq:15} therefore gives \(\partial_m\mathrm{Reg}_{\mathrm{cov}}(m,\tau,v)\ge0\) at every such \(m\). Together with the single-crossing result above, this shows that every parameter-specific minimizer is no larger than \(\alpha N/(1+\alpha)\). We have therefore proved the uniform upper bound \(\sup_{\tau,v}m^*(\tau,v)\le\alpha N/(1+\alpha)\).

(iv) We construct a parameter value with a constant signal to show that the upper bound is sharp.

It remains to show that the preceding upper bound cannot be reduced. To do so, fix a sample size below the boundary and construct a constant signal for which the unique crossing lies to its right. Specifically, fix any \(m_0<\alpha N/(1+\alpha)\) and consider \(b(m_0)\). Then \(b(m_0)>1\), and
\[
\left.\frac{\partial}{\partial r}
\{2\Phi(-r)+2b(m_0)r\varphi(r)\}\right|_{r=0}
=2\{b(m_0)-1\}\varphi(0)>0.
\]
Moreover, \(2\Phi(0)+2b(m_0)\cdot0\cdot\varphi(0)=1\).

Choose \(r>0\) sufficiently small that \(2\Phi(-r)+2b(m_0)r\varphi(r)>1\), and take the parameter value \(v(x)=1, \tau(x)=\frac{r}{m_0^\alpha}>0\) for $f$-almost every $x$.

Because \(f\) is a probability density or mass function, this parameter value has \(0<\int f(x)|\tau(x)|\,dx<\infty,\)
and its normalized signal is constant. Equation \eqref{eq:14} therefore reduces to the pointwise ratio \(2\Phi(-r)+2b(m_0)r\varphi(r)\) and exceeds one at \(m_0\). The single-crossing result implies \(m^*(\tau,v)>m_0\). Since this holds for every \(m_0\) below the boundary, the reverse inequality is \(\sup_{\tau,v}m^*(\tau,v)\ge\alpha N/(1+\alpha)\). Combining the two bounds gives
\[
\sup_{\tau,v}m^*(\tau,v)
=
\frac{\alpha}{1+\alpha}N.
\]

This constant-signal construction supplies sharpness, while the pointwise bound supplies uniformity, and together they prove the proposition. \end{proof}

\medskip

\begin{proof}[Proof of Corollary~\ref{cor:5.3}]

The corollary makes two claims. First, \(m_{\mathrm{cov}}^{\mathrm{Env}}=\alpha N/(1+\alpha)\) has strictly lower covariate-integrated regret than every larger sample size at every parameter value, so it strictly dominates every larger sample size. Second, no other sample size weakly dominates any \(m\in(0,m_{\mathrm{cov}}^{\mathrm{Env}}]\). For the second claim, we compare two fixed sample sizes at parameter values with a constant signal, where a small signal favors the larger sample size when both lie in \((0,m_{\mathrm{cov}}^{\mathrm{Env}}]\) and a large signal favors the smaller one. Throughout, a parameter value is a measurable pair \((\tau,v)\) with \(0<v(x)<\infty\) for \(f\)-almost every \(x\) and \(0<\int f(x)|\tau(x)|\,dx<\infty\), as in Proposition~\ref{prop:5.2}, and \(\rho_{\mathrm{cov}}(x)=|\tau(x)|/v(x)\).

(i) We show that the envelope sample size strictly dominates every larger sample size.

Fix a parameter value \((\tau,v)\) and \(u\in[m_{\mathrm{cov}}^{\mathrm{Env}},N)\). The inequality \(u\ge\alpha N/(1+\alpha)\) is equivalent to \(\alpha(N-u)\le u\), so \(b(u)\) in the proof of Proposition~\ref{prop:5.2} satisfies \(0<b(u)\le1\). With \(r=u^\alpha\rho_{\mathrm{cov}}(x)\), the pointwise ratio in Equation \eqref{eq:14} at \(u\) equals \(2\Phi(-r)+2b(u)r\varphi(r)\). As computed in the uniform upper bound step of the proof of Proposition~\ref{prop:5.2}, this function equals one at \(r=0\) and has derivative \(2\{b(u)-1-b(u)r^2\}\varphi(r)\) with respect to \(r\). For \(r>0\), this derivative is strictly negative because \(b(u)-1\le0\) and \(b(u)r^2>0\). Hence \(\mathrm{CBR}_\alpha\{u,\rho_{\mathrm{cov}}(x)\}<1\) whenever \(\rho_{\mathrm{cov}}(x)>0\). Substituting Equation \eqref{eq:14} into Equation \eqref{eq:15} gives
\[
\partial_u\mathrm{Reg}_{\mathrm{cov}}(u,\tau,v)
=\frac12\int f(x)|\tau(x)|\left[1-\mathrm{CBR}_\alpha\{u,\rho_{\mathrm{cov}}(x)\}\right]dx.
\]

For \(f\)-almost every \(x\) with \(f(x)|\tau(x)|>0\), the condition \(0<v(x)<\infty\) gives \(\rho_{\mathrm{cov}}(x)>0\), so the integrand is strictly positive there. This set has positive weight because \(\int f(x)|\tau(x)|\,dx>0\). Hence \(\partial_u\mathrm{Reg}_{\mathrm{cov}}(u,\tau,v)>0\) for every \(u\in[m_{\mathrm{cov}}^{\mathrm{Env}},N)\). By Equation \eqref{eq:15} and the continuity of the integrated ratio shown in the proof of Proposition~\ref{prop:5.2}, this derivative is continuous in \(u\). Therefore, for every \(m\in(m_{\mathrm{cov}}^{\mathrm{Env}},N)\),
\[
\mathrm{Reg}_{\mathrm{cov}}(m,\tau,v)-\mathrm{Reg}_{\mathrm{cov}}(m_{\mathrm{cov}}^{\mathrm{Env}},\tau,v)
=\int_{m_{\mathrm{cov}}^{\mathrm{Env}}}^{m}\partial_u\mathrm{Reg}_{\mathrm{cov}}(u,\tau,v)\,du>0.
\]

The parameter value was arbitrary, so \(m_{\mathrm{cov}}^{\mathrm{Env}}\) strictly dominates every larger sample size.

(ii) We compare two fixed sample sizes at parameter values with a constant signal.

Fix \(0<m_1<m_2<N\). For \(\varepsilon>0\), let \(\tau(x)=\varepsilon\) and \(v(x)=1\) for every \(x\). This parameter value is measurable, satisfies \(0<v(x)<\infty\), and has \(\int f(x)|\tau(x)|\,dx=\varepsilon\in(0,\infty)\) because \(f\) is a probability density or mass function. Under Assumption~\ref{asm:5.1}, \(e_x(m,\tau(x),v(x))=\Phi(-m^\alpha\varepsilon)\) for every \(x\), so the regret formula gives
\[
\begin{aligned}
\frac{\mathrm{Reg}_{\mathrm{cov}}(m_1,\tau,v)-\mathrm{Reg}_{\mathrm{cov}}(m_2,\tau,v)}{\varepsilon}
&=\frac{m_1-m_2}{2}+(N-m_1)\Phi(-m_1^\alpha\varepsilon)\\
&\quad-(N-m_2)\Phi(-m_2^\alpha\varepsilon).
\end{aligned}
\]

The right-hand side is a continuously differentiable function of \(\varepsilon\in\mathbb R\). Since \(\Phi(0)=1/2\), its value at \(\varepsilon=0\) is \((m_1-m_2)/2+(N-m_1)/2-(N-m_2)/2=0\). Differentiating the two Gaussian terms shows that its derivative at \(\varepsilon=0\) equals
\[
\varphi(0)\left\{(N-m_2)m_2^\alpha-(N-m_1)m_1^\alpha\right\}.
\]

Suppose now that \(m_2\le m_{\mathrm{cov}}^{\mathrm{Env}}\). Since \(\partial_m\{(N-m)m^\alpha\}=m^{\alpha-1}\{\alpha N-(1+\alpha)m\}\) is strictly positive on \((0,\alpha N/(1+\alpha))\),
\[
(N-m_2)m_2^\alpha-(N-m_1)m_1^\alpha
=\int_{m_1}^{m_2}u^{\alpha-1}\{\alpha N-(1+\alpha)u\}\,du>0.
\]

The difference therefore vanishes at \(\varepsilon=0\) and has a strictly positive derivative there, so it is strictly positive for all sufficiently small \(\varepsilon>0\). At such a parameter value, \(m_2\) has strictly lower regret than \(m_1\).

Without the restriction \(m_2\le m_{\mathrm{cov}}^{\mathrm{Env}}\), let \(\varepsilon\to\infty\) instead. Both Gaussian tail probabilities converge to zero, so the difference converges to \((m_1-m_2)/2<0\). For a sufficiently large finite \(\varepsilon\), \(m_1\) has strictly lower regret than \(m_2\).

(iii) We combine the comparisons to establish admissibility below the envelope.

Let \(m\in(0,m_{\mathrm{cov}}^{\mathrm{Env}}]\) and let \(m'\in(0,N)\) differ from \(m\). If \(m'<m\), the small-signal comparison with \(m_1=m'\) and \(m_2=m\) gives a parameter value at which \(m\) has strictly lower regret than \(m'\). If \(m'>m\), the large-signal comparison with \(m_1=m\) and \(m_2=m'\) gives such a parameter value. In both cases \(m'\) does not weakly dominate \(m\), so \(m\) is admissible. Combined with the first step, this shows that the admissible sample sizes are exactly those in \((0,m_{\mathrm{cov}}^{\mathrm{Env}}]\). The first step also shows that this set is a complete class, so by Section~\ref{sec:main-theorem} it is the unique minimal complete class. \end{proof}

\section{Proofs for the Multiple-Treatment Extension}\label{sec:appendix-E}

The two variance specifications require different proof strategies. Under heterogeneous variances, the proof of Theorem~\ref{fm:6.2} first establishes a selection asymmetry under any variance-ratio bound \(C>1\). It then uses this asymmetry to violate the uniform cost-benefit condition and to construct parameter values whose global regret minimizers approach \(N\). The proof of Corollary~\ref{fm:6.3} holds the standard-deviation vector fixed and constructs two mean vectors that reverse the regret ranking of any two sample sizes. We then prove that these conclusions remain valid under the additional variance and mean-gap bounds in the remark. Under common variance, a reparameterization reduces regret comparisons to the normalized per-unit rollout loss \(q_K\). The proof of Theorem~\ref{fm:6.4} bounds \(\beta_K\), locates the boundary through the derivative of regret, and constructs, for \(K\ge3\), one parameter value at which the boundary is the unique regret minimizer. The proof of Corollary~\ref{cor:6.5} compares each sample size up to the boundary with every competitor through a finite-difference inequality for \(q_K\).

\subsection{Heterogeneous variances}\label{sec:appendix-E1}

\begin{proof}[Proof of Theorem~\ref{fm:6.2}]

Fix \(K\ge3\) and \(1<C<\infty\). We first show that an arbitrarily
small variance difference makes one arm more likely to be selected
at the equal-mean boundary. We then use this asymmetry to establish
the failure of the uniform cost-benefit condition and construct
parameter values whose global regret minimizers approach \(N\).

\medskip
\noindent
\textit{(i) Selection probabilities under bounded variance ratios.}

Let \(Z_1,\ldots,Z_K\) be independent standard normal variables.
For \(L>0\) and \(t\in\mathbb R\), define
\[
\pi_L(t)
:=
\Pr\left(
LZ_1-t\ge\max_{2\le k\le K}Z_k
\right)
=
\int_{\mathbb R}
\varphi(z)\Phi(Lz-t)^{K-1}\,dz.
\]
The integral representation and dominated convergence show that
\((L,t)\mapsto\pi_L(t)\) is jointly continuous on
\((0,\infty)\times\mathbb R\).
At \(L=1\) and \(t=0\), symmetry gives \(\pi_1(0)=1/K\).
Differentiating with respect to \(L\) at \(t=0\) and combining
the contributions from positive and negative \(z\), we obtain
\[
\begin{aligned}
\frac{\partial}{\partial L}\pi_L(0)
&=
(K-1)\int_{\mathbb R}
z\varphi(z)\varphi(Lz)\Phi(Lz)^{K-2}\,dz\\
&=
(K-1)\int_0^\infty
z\varphi(z)\varphi(Lz)
\left[
\Phi(Lz)^{K-2}-\Phi(-Lz)^{K-2}
\right]dz
>0.
\end{aligned}
\]
The strict inequality follows because \(K\ge3\), \(L>0\), and
\(\Phi(Lz)>\Phi(-Lz)\) for every \(z>0\).
Differentiation under the integral sign is justified by the
integrable bound
\[
(K-1)|z|\varphi(z)\varphi(Lz)\Phi(Lz)^{K-2}
\le
\frac{K-1}{\sqrt{2\pi}}|z|\varphi(z).
\]
Consequently,
\[
\pi_L(0)>\frac1K
\qquad\text{for every }L>1.
\]

Choose and fix \(L\in(1,\sqrt C)\), and set
\[
\sigma=(L,1,\ldots,1).
\]
Then
\[
\frac{\max_k\sigma_k^2}{\min_k\sigma_k^2}
=
L^2<C.
\]
This standard-deviation vector is held fixed throughout the rest of the proof.

We also record that
\[
\pi_L'(t)
=
-(K-1)\int_{\mathbb R}
\varphi(z)\varphi(Lz-t)\Phi(Lz-t)^{K-2}\,dz
<0
\qquad(t\in\mathbb R).
\]
The integrand, including the factor \(K-1\), is bounded in absolute
value by \((K-1)\varphi(z)/\sqrt{2\pi}\). This bound justifies
differentiation and shows that \(\pi_L'\) is continuous and bounded.
Moreover, dominated convergence gives
\[
\lim_{t\to\infty}\pi_L(t)=0.
\]

\medskip
\noindent
\textit{(ii) Failure of the uniform cost-benefit condition.}

For \(\varepsilon>0\), define
\[
\mu^\varepsilon=(-\varepsilon,0,\ldots,0).
\]
Arm 1 is the only inferior arm, with gap \(\varepsilon\), and
\[
\bar\Delta=\frac{\varepsilon}{K}>0.
\]
Hence \((\mu^\varepsilon,\sigma)\in\Theta_K(C)\) for every
\(\varepsilon>0\).

For each fixed \(m>0\), Assumption~\ref{asm:6.1} gives the representation
\[
\widehat\mu_1
=
-\varepsilon+L\sqrt{\frac Km}Z_1,
\qquad
\widehat\mu_k
=
\sqrt{\frac Km}Z_k
\quad(k=2,\ldots,K).
\]
The empirical success rule therefore selects arm 1 with probability
\[
\Pr_{\mu^\varepsilon,\sigma}(\widehat k_m=1)
=
\pi_L\left(\varepsilon\sqrt{\frac mK}\right).
\]
Writing \(t_m:=\varepsilon\sqrt{m/K}\), we obtain
\[
r_K(m,\mu^\varepsilon,\sigma)
=
\varepsilon\pi_L(t_m),
\qquad
\partial_m r_K(m,\mu^\varepsilon,\sigma)
=
\varepsilon\pi_L'(t_m)\frac{t_m}{2m}.
\]
Since \(\bar\Delta=\varepsilon/K\), it follows that
\[
\mathrm{CBR}_K(m,\mu^\varepsilon,\sigma)
=
K\left[
\pi_L(t_m)
-
\frac{N-m}{2m}t_m\pi_L'(t_m)
\right].
\]

Fix any \(m\in(0,N)\). As \(\varepsilon\downarrow0\), we have
\(t_m\to0\), and the continuity of \(\pi_L\) and boundedness of
\(\pi_L'\) imply
\[
\mathrm{CBR}_K(m,\mu^\varepsilon,\sigma)
\longrightarrow
K\pi_L(0)>1.
\]
Thus, for all sufficiently small \(\varepsilon>0\),
\[
\sup_{(\mu,\sigma)\in\Theta_K(C)}
\mathrm{CBR}_K(m,\mu,\sigma)
\ge
\mathrm{CBR}_K(m,\mu^\varepsilon,\sigma)
>1.
\]
Since \(m\in(0,N)\) was arbitrary, no sample size below \(N\)
satisfies the uniform cost-benefit condition.

\medskip
\noindent
\textit{(iii) The envelope equals \(N\).}

Because \(\pi_L\) is continuous and strictly decreasing,
\(\pi_L(0)>1/K\), and \(\pi_L(t)\to0\) as \(t\to\infty\),
there is a unique \(t_L>0\) such that
\[
\pi_L(t_L)=\frac1K.
\]

Fix any \(\eta\in(0,N)\). Choose \(\varepsilon>0\) satisfying
\[
t_L\sqrt{\frac KN}
<
\varepsilon
<
t_L\sqrt{\frac K{N-\eta}}.
\]
This interval is nonempty because \(N-\eta<N\).
For this choice,
\[
\pi_L\left(\varepsilon\sqrt{\frac{N-\eta}{K}}\right)
>
\frac1K
>
\pi_L\left(\varepsilon\sqrt{\frac NK}\right).
\]

Retain the parameter value \((\mu^\varepsilon,\sigma)\in\Theta_K(C)\).
For \(m\in(0,N)\), regret and its derivative are
\[
\begin{aligned}
\mathrm{Reg}_K(m,\mu^\varepsilon,\sigma)
&=
\varepsilon\left[
\frac mK+(N-m)\pi_L(t_m)
\right],\\
\partial_m\mathrm{Reg}_K(m,\mu^\varepsilon,\sigma)
&=
\varepsilon\left[
\frac1K-\pi_L(t_m)
+
\frac{N-m}{2m}t_m\pi_L'(t_m)
\right].
\end{aligned}
\]
For every \(0<m\le N-\eta\),
\[
t_m
\le
\varepsilon\sqrt{\frac{N-\eta}{K}}
<
t_L,
\]
so \(\pi_L(t_m)>1/K\). Since \(\pi_L'(t_m)<0\), this gives
\[
\partial_m\mathrm{Reg}_K(m,\mu^\varepsilon,\sigma)<0
\qquad(0<m\le N-\eta).
\]
At the other endpoint,
\[
\lim_{m\uparrow N}
\partial_m\mathrm{Reg}_K(m,\mu^\varepsilon,\sigma)
=
\varepsilon\left[
\frac1K-
\pi_L\left(\varepsilon\sqrt{\frac NK}\right)
\right]
>0.
\]

The regret function extends continuously to \([0,N]\), with
endpoint values
\[
\mathrm{Reg}_K(0,\mu^\varepsilon,\sigma)
=
\varepsilon N\pi_L(0),
\qquad
\mathrm{Reg}_K(N,\mu^\varepsilon,\sigma)
=
\frac{\varepsilon N}{K}.
\]
The extension therefore attains a global minimum on \([0,N]\).
Strict decrease on \((0,N-\eta]\), together with continuity at zero,
rules out every point of \([0,N-\eta)\). The negative derivative
at \(N-\eta\) rules out \(N-\eta\), and the positive left derivative
at \(N\) rules out \(N\). Hence every global minimizer belongs to
\((N-\eta,N)\) and is also a global minimizer on \((0,N)\).

We have therefore constructed, for every \(\eta\in(0,N)\), a
parameter value in \(\Theta_K(C)\) with a nonempty set of global
minimizers contained in \((N-\eta,N)\). Since all sample sizes
entering the definition of the envelope lie in \((0,N)\),
\[
m_K^{\mathrm{Env}}(C)=N.
\]

Finally, taking instead
\[
\varepsilon_0=t_L\sqrt{\frac KN}
\]
gives
\[
\pi_L\left(\varepsilon_0\sqrt{\frac mK}\right)>\frac1K
\qquad(0<m<N).
\]
The same derivative formula then shows that regret is strictly
decreasing throughout \((0,N)\). Its infimum is approached only
as \(m\uparrow N\), so it has no global minimizer on \((0,N)\).
This illustrates why the definition of the envelope accommodates
parameter values with an empty set of global minimizers.
\end{proof}

\medskip

\begin{proof}[Proof of Corollary~\ref{fm:6.3}]

Fix \(K\ge3\), \(1<C<\infty\), and \(0<m_1<m_2<N\).
Choose \(L\in(1,\sqrt C)\) and fix
\(\sigma=(L,1,\ldots,1)\). Let \(\pi_L\) be the selection-probability
function defined in the proof of Theorem~\ref{fm:6.2}, which satisfies
\(\pi_L(0)>1/K\). We construct two mean vectors, keeping this same standard-deviation fixed, at which the regret rankings of \(m_1\) and
\(m_2\) are reversed.

\medskip
\noindent
\textit{(i) A parameter value favoring the larger sample size.}

For \(\varepsilon>0\), let
\(\mu^-_\varepsilon=(-\varepsilon,0,\ldots,0)\).
Then \((\mu^-_\varepsilon,\sigma)\in\Theta_K(C)\), arm 1 is the
only inferior arm, and \(\bar\Delta=\varepsilon/K\).
With \(t_m=\varepsilon\sqrt{m/K}\), the preceding proof gives
\[
\mathrm{Reg}_K(m,\mu^-_\varepsilon,\sigma)
=
\varepsilon\left[\frac mK+(N-m)\pi_L(t_m)\right].
\]
Consequently, continuity of \(\pi_L\) at zero implies
\[
\frac{
\mathrm{Reg}_K(m_2,\mu^-_\varepsilon,\sigma)
-
\mathrm{Reg}_K(m_1,\mu^-_\varepsilon,\sigma)
}{\varepsilon}
\longrightarrow
(m_2-m_1)\left[\frac1K-\pi_L(0)\right]<0
\]
as \(\varepsilon\downarrow0\). Hence \(m_2\) has strictly lower
regret than \(m_1\) for all sufficiently small \(\varepsilon>0\).

\medskip
\noindent
\textit{(ii) A parameter value favoring the smaller sample size.}

Now let \(\mu^+_\varepsilon=(\varepsilon,0,\ldots,0)\), retaining
\(\sigma=(L,1,\ldots,1)\). Arm 1 is the unique best arm, every other
arm has gap \(\varepsilon\), and
\(\bar\Delta=(K-1)\varepsilon/K>0\).
Thus \((\mu^+_\varepsilon,\sigma)\in\Theta_K(C)\).
The probability of selecting arm 1 is \(\pi_L(-t_m)\), so
\[
\begin{aligned}
r_K(m,\mu^+_\varepsilon,\sigma)
&=\varepsilon\{1-\pi_L(-t_m)\},\\
\mathrm{Reg}_K(m,\mu^+_\varepsilon,\sigma)
&=\varepsilon\left[
\frac{K-1}{K}m+(N-m)\{1-\pi_L(-t_m)\}
\right].
\end{aligned}
\]
Therefore,
\[
\frac{
\mathrm{Reg}_K(m_2,\mu^+_\varepsilon,\sigma)
-
\mathrm{Reg}_K(m_1,\mu^+_\varepsilon,\sigma)
}{\varepsilon}
\longrightarrow
(m_2-m_1)\left[\pi_L(0)-\frac1K\right]>0
\]
as \(\varepsilon\downarrow0\). Hence \(m_1\) has strictly lower
regret than \(m_2\) for all sufficiently small \(\varepsilon>0\).

These two parameter values show that neither sample size weakly
dominates the other over \(\Theta_K(C)\). Since \(m_1\) and \(m_2\)
were arbitrary, every deterministic sample size in \((0,N)\) is
admissible under the empirical success rule.
The full design space \((0,N)\) is a complete class. Every proper
subset omits an admissible sample size, which cannot be weakly
dominated by any sample size in that subset. No proper subset is
therefore a complete class. Thus \((0,N)\) is the unique minimal
complete class among deterministic sample-size choices.
\end{proof}

\medskip

\begin{proof}[Proof of Remark~\ref{rem:bounded_variances}]

Fix \(K\ge3\), \(1<C<\infty\), \(B>0\), and
\(0<\underline{\sigma}<\overline{\sigma}<\infty\).
Write the restricted parameter space as
\[
\Theta^{\mathrm b}
:=
\left\{
(\mu,\sigma)\in\Theta_K(C):
0<\bar\Delta\le B,\quad
\underline{\sigma}\le\sigma_k\le\overline{\sigma}
\ \text{for all }k
\right\}.
\]
All suprema and admissibility comparisons in this proof are taken
over \(\Theta^{\mathrm b}\).
Let \(s=\underline{\sigma}\) and choose
\[
1<L<\min\left\{\sqrt C,\frac{\overline{\sigma}}{s}\right\}.
\]
The interval is nonempty. The standard-deviation vector
\[
\sigma^{L,s}:=(Ls,s,\ldots,s)
\]
satisfies both the absolute bounds and the variance-ratio restriction.
We retain \(s\) throughout the proof; \(L\) may be chosen separately
for each construction.

\medskip
\noindent
\textit{(i) Failure of the uniform cost-benefit condition.}

Fix any such \(L\) and any \(m\in(0,N)\). Set
\(\mu^-_\varepsilon=(-\varepsilon,0,\ldots,0)\) and
\[
t_m:=\frac{\varepsilon}{s}\sqrt{\frac mK}.
\]
The probability of selecting arm 1 is \(\pi_L(t_m)\), where
\(\pi_L\) is defined in the proof of Theorem~\ref{fm:6.2}.
For \(0<\varepsilon\le KB\), we have
\((\mu^-_\varepsilon,\sigma^{L,s})\in\Theta^{\mathrm b}\).
The same calculation as in step (ii) of that proof gives
\[
\mathrm{CBR}_K(m,\mu^-_\varepsilon,\sigma^{L,s})
=
K\left[
\pi_L(t_m)-\frac{N-m}{2m}t_m\pi_L'(t_m)
\right]
\longrightarrow K\pi_L(0)>1
\]
as \(\varepsilon\downarrow0\).
Hence the uniform cost-benefit condition fails at every \(m\in(0,N)\)
even over \(\Theta^{\mathrm b}\).

\medskip
\noindent
\textit{(ii) The envelope over the restricted space equals \(N\).}

For each \(L>1\), let \(t_L>0\) be the unique solution to
\(\pi_L(t_L)=1/K\). We first show that
\[
t_L\longrightarrow0\qquad\text{as }L\downarrow1.
\]
Fix any \(a>0\). The strict decrease of \(\pi_1\) gives
\(\pi_1(a)<\pi_1(0)=1/K\). Joint continuity of \(\pi_L(t)\)
therefore implies \(\pi_L(a)<1/K\) for all \(L>1\) sufficiently
close to one. Since \(\pi_L(0)>1/K\), the unique root then
satisfies \(0<t_L<a\). As \(a>0\) was arbitrary, the claimed
limit follows.

Fix any \(\eta\in(0,N)\). Choose \(L>1\) sufficiently close to one
that
\[
L<\min\left\{\sqrt C,\frac{\overline{\sigma}}{s}\right\},
\qquad
s t_L\sqrt{\frac K{N-\eta}}<KB.
\]
This choice is possible by \(t_L\to0\).
Choose \(\varepsilon\) satisfying
\[
s t_L\sqrt{\frac KN}
<\varepsilon<
s t_L\sqrt{\frac K{N-\eta}}<KB,
\]
and take \(\mu^-_\varepsilon=(-\varepsilon,0,\ldots,0)\) and
\(\sigma^{L,s}=(Ls,s,\ldots,s)\).
Then \((\mu^-_\varepsilon,\sigma^{L,s})\in\Theta^{\mathrm b}\) and,
with \(t_m=(\varepsilon/s)\sqrt{m/K}\),
\[
t_{N-\eta}<t_L<t_N,
\qquad
\pi_L(t_{N-\eta})>\frac1K>\pi_L(t_N).
\]
The regret function and its derivative satisfy
\[
\begin{aligned}
\mathrm{Reg}_K(m,\mu^-_\varepsilon,\sigma^{L,s})
&=\varepsilon\left[\frac mK+(N-m)\pi_L(t_m)\right],\\
\partial_m\mathrm{Reg}_K(m,\mu^-_\varepsilon,\sigma^{L,s})
&=\varepsilon\left[
\frac1K-\pi_L(t_m)+\frac{N-m}{2m}t_m\pi_L'(t_m)
\right].
\end{aligned}
\]
The derivative is strictly negative for \(0<m\le N-\eta\), while
its limit as \(m\uparrow N\) is
\(\varepsilon\{1/K-\pi_L(t_N)\}>0\).
The function has a continuous extension to \([0,N]\).
The same endpoint argument as in step (iii) of the proof of
Theorem~\ref{fm:6.2} therefore shows that its set of global
minimizers is nonempty and contained in \((N-\eta,N)\).
Since \(\eta\in(0,N)\) was arbitrary,
\[
\sup\left(
\bigcup_{(\mu,\sigma)\in\Theta^{\mathrm b}}
\mathop{\mathrm{arg\,min}}_{0<m<N}
\mathrm{Reg}_K(m,\mu,\sigma)
\right)=N.
\]

\medskip
\noindent
\textit{(iii) Admissibility and the unique minimal complete class.}

Fix \(0<m_1<m_2<N\), and choose any
\(L\in(1,\min\{\sqrt C,\overline{\sigma}/s\})\).
Hold \(\sigma^{L,s}=(Ls,s,\ldots,s)\) fixed and consider
\[
\mu^-_\varepsilon=(-\varepsilon,0,\ldots,0),
\qquad
\mu^+_\varepsilon=(\varepsilon,0,\ldots,0).
\]
Their average gaps are \(\varepsilon/K\) and
\((K-1)\varepsilon/K\), respectively. Thus both parameter values
belong to \(\Theta^{\mathrm b}\) whenever
\(0<\varepsilon\le KB/(K-1)\).
The calculations in the proof of Corollary~\ref{fm:6.3} apply with
\(t_m=(\varepsilon/s)\sqrt{m/K}\) and give
\[
\begin{aligned}
\frac{
\mathrm{Reg}_K(m_2,\mu^-_\varepsilon,\sigma^{L,s})
-\mathrm{Reg}_K(m_1,\mu^-_\varepsilon,\sigma^{L,s})
}{\varepsilon}
&\longrightarrow
(m_2-m_1)\left[\frac1K-\pi_L(0)\right]<0,\\
\frac{
\mathrm{Reg}_K(m_2,\mu^+_\varepsilon,\sigma^{L,s})
-\mathrm{Reg}_K(m_1,\mu^+_\varepsilon,\sigma^{L,s})
}{\varepsilon}
&\longrightarrow
(m_2-m_1)\left[\pi_L(0)-\frac1K\right]>0
\end{aligned}
\]
as \(\varepsilon\downarrow0\).
Choosing \(\varepsilon>0\) sufficiently small makes both inequalities
strict while satisfying the gap bound. Neither sample size weakly
dominates the other over \(\Theta^{\mathrm b}\).
Since the pair was arbitrary, every deterministic sample size in
\((0,N)\) is admissible over the restricted space. The complete-class
argument in the proof of Corollary~\ref{fm:6.3} then shows that
\((0,N)\) is again the unique minimal complete class.
\end{proof}

\subsection{Common variance}
\label{sec:appendix-E2}

We first state a lemma on \(q_K\) and \(\beta_K\) and then prove
Theorem~\ref{fm:6.4} and Corollary~\ref{cor:6.5}. All sample-size
comparisons below are between deterministic choices in \((0,N)\)
under the same empirical success rule.

Throughout, \(\sigma_1=\cdots=\sigma_K=\sigma>0\) and
\(\bar\Delta>0\), and we use the notation of
Section~\ref{sec:common-variance}, so that
\(d_k=\Delta_k/\bar\Delta\),
\(d=(d_1,\ldots,d_K)\in\mathcal D_K\),
\(\xi=\bar\Delta\sqrt m/(\sigma\sqrt K)\), and
\(r_K(m,\mu,\sigma)=\bar\Delta q_K(\xi,d)\).
We call a pair \((\xi,d)\in(0,\infty)\times\mathcal D_K\) a state.
Under a parameter value \((\mu,\sigma)\), the state at sample size
\(u\in(0,N)\) is
\((\bar\Delta\sqrt u/(\sigma\sqrt K),d)\).
The normalized marginal ratio at a state \((\xi,d)\) is
\(\{-(\xi/2)\partial_\xi q_K(\xi,d)\}/\{1-q_K(\xi,d)\}\),
which is well defined because \(1-q_K(\xi,d)>0\) by
Lemma~\ref{lem:E.1}(b) below. Thus \(\beta_K\) in
Theorem~\ref{fm:6.4} is the supremum of the normalized marginal
ratio over all states. Since
\(\mathrm{Reg}_K(m,\mu,\sigma)
=m\bar\Delta+(N-m)\bar\Delta q_K(\xi,d)\),
\begin{align}
\mathrm{Reg}_K(m,\mu,\sigma)
=
\bar\Delta\bigl[N-(N-m)\{1-q_K(\xi,d)\}\bigr].
\label{eq:E.1}
\end{align}

Lemma~\ref{lem:E.1} concerns \(q_K\) and \(\beta_K\) only, and its
proof uses neither Theorem~\ref{fm:6.4} nor
Corollary~\ref{cor:6.5}. Parts (e) and (f) assume
\(0<\beta_K<\infty\), which step (i) of the proof of
Theorem~\ref{fm:6.4} derives from parts (a), (b), and (d) before
either part is used. Lemma~\ref{lem:E.1} is proved in
Appendix~\ref{sec:appendix-E3}.

\medskip
\begin{lemma}[Properties of \(q_K\) and \(\beta_K\)]
\label{lem:E.1}

Fix any finite \(K\ge2\). The following hold.

\begin{enumerate}
    \item[(a)] For every \(d\in\mathcal D_K\), the map
\(\xi\mapsto q_K(\xi,d)\) is continuously differentiable on
\([0,\infty)\), with a one-sided derivative at \(\xi=0\), and
twice continuously differentiable on \((0,\infty)\).

\item[(b)] For every \(d\in\mathcal D_K\), \(q_K(0,d)=1\), and
\(-\partial_\xi q_K(\xi,d)>0\) for every \(\xi\ge0\).
Consequently, for every \(\xi>0\),
\[
1-q_K(\xi,d)
=
\int_0^\xi\{-\partial_sq_K(s,d)\}\,ds>0.
\]

\item[(c)] For every \(d\in\mathcal D_K\),
\(q_K(\xi,d)\to0\) as \(\xi\to\infty\).

\item[(d)] For every \(d\in\mathcal D_K\), every \(\xi>0\), and every
\(\gamma\ge0\),
\[
\frac{\xi\,\partial_\xi^2q_K(\xi,d)}
     {\partial_\xi q_K(\xi,d)}
\le
\frac13+\frac{\gamma^2}{4}
+\frac{K-2}{3}
\sup_{z\in\mathbb R}
(z-\gamma)\frac{\varphi(z)}{\Phi(z)}.
\]

\item[(e)] Suppose \(0<\beta_K<\infty\). For every \(a\in(0,1)\), there
exists a state \((\xi,d)\in(0,\infty)\times\mathcal D_K\) such that
\[
\frac{1-q_K(\xi,d)}
     {1-q_K(\sqrt a\,\xi,d)}
>
1+\beta_K(1-a).
\]

\item[(f)] Suppose \(K\ge3\) and \(0<\beta_K<\infty\). Then there exists a
state \((\xi,d)\in(0,\infty)\times\mathcal D_K\), not depending on
\(a\), at which the inequality in (e) holds for every \(a\in(0,1)\).
\end{enumerate}
\end{lemma}
\medskip

\begin{proof}[Proof of Theorem~\ref{fm:6.4}]

Step (i) proves \(1/2\le\beta_K<\infty\). Step (ii) derives the
supremum identity and the formula for \(m_K^{\mathrm{CBR}}\) from
the derivative of the regret, and step (iii) shows that
\(m_K^{\mathrm{CBR}}\) strictly dominates every larger sample size.
Step (iv) shows that every parameter value has a global minimizer
and that all global minimizers lie in \((0,m_K^{\mathrm{CBR}}]\),
so \(m_K^{\mathrm{Env}}\le m_K^{\mathrm{CBR}}\).
For the reverse inequality with \(K\ge3\), step (v) uses
Lemma~\ref{lem:E.1}(f) to construct a parameter value under which
\(m_K^{\mathrm{CBR}}\) is the unique global minimizer.
Step (vi) treats \(K=2\) through Theorem~\ref{thm:3.1}.

(i) We bound \(\beta_K\).

Fix \(d\in\mathcal D_K\). For \(\xi>0\), Lemma~\ref{lem:E.1}(b)
writes \(1-q_K(\xi,d)\) as the integral of
\(-\partial_sq_K(s,d)\) over \([0,\xi]\), so the normalized marginal
ratio at \((\xi,d)\) equals
\[
\frac{-(\xi/2)\partial_\xi q_K(\xi,d)}
     {1-q_K(\xi,d)}
=
\frac12\cdot
\frac{-\partial_\xi q_K(\xi,d)}
     {\xi^{-1}\int_0^\xi\{-\partial_sq_K(s,d)\}\,ds}.
\]

As \(\xi\downarrow0\), the numerator on the right converges to the
one-sided derivative \(-\partial_\xi q_K(0,d)\), because
\(\partial_\xi q_K(\cdot,d)\) is continuous on \([0,\infty)\) by
Lemma~\ref{lem:E.1}(a). The denominator is the average of the
continuous function \(s\mapsto-\partial_sq_K(s,d)\) over
\([0,\xi]\), so it lies between the minimum and the maximum of
this function on \([0,\xi]\) and converges to the same limit.
This limit is positive by Lemma~\ref{lem:E.1}(b), and therefore
\begin{align}
\lim_{\xi\downarrow0}
\frac{-(\xi/2)\partial_\xi q_K(\xi,d)}
     {1-q_K(\xi,d)}
=
\frac12
\qquad\text{for every }d\in\mathcal D_K.
\label{eq:E.2}
\end{align}

Since \(\beta_K\) is at least the normalized marginal ratio at
every state, \eqref{eq:E.2} gives \(\beta_K\ge1/2\).

For the upper bound, we first evaluate Lemma~\ref{lem:E.1}(d) at
\(\gamma=0\). If \(z\le0\), then
\(z\varphi(z)/\Phi(z)\le0\). If \(z>0\), then
\(\Phi(z)\ge1/2\) and \(z\varphi(z)\le\varphi(1)\), because
\(\{t\varphi(t)\}'=(1-t^2)\varphi(t)\) shows that
\(t\varphi(t)\) attains its maximum over \(t>0\) at \(t=1\).
Hence
\(\sup_{z\in\mathbb R}z\varphi(z)/\Phi(z)\le2\varphi(1)\),
and since \(K-2\ge0\), Lemma~\ref{lem:E.1}(d) with \(\gamma=0\)
gives
\(t\,\partial_t^2q_K(t,d)/\partial_tq_K(t,d)\le A_K\)
for every \(t>0\), where
\(A_K:=1/3+2(K-2)\varphi(1)/3\).
By Lemma~\ref{lem:E.1}(a) and (b), the map
\(t\mapsto\log\{-\partial_tq_K(t,d)\}\) is continuously
differentiable on \((0,\infty)\), and its derivative
\(\partial_t^2q_K(t,d)/\partial_tq_K(t,d)\) is at most \(A_K/t\).
For \(0<s\le\xi\), integrating this derivative over
\(t\in[s,\xi]\) yields
\[
\begin{aligned}
&\log\{-\partial_\xi q_K(\xi,d)\}
-\log\{-\partial_sq_K(s,d)\}\\
&\qquad=
\int_s^\xi
\frac{\partial_t^2q_K(t,d)}
     {\partial_tq_K(t,d)}\,dt
\le
A_K\log\frac{\xi}{s},
\end{aligned}
\]
and exponentiating gives
\begin{align}
-\partial_sq_K(s,d)
\ge
\Bigl(\frac{s}{\xi}\Bigr)^{A_K}
\{-\partial_\xi q_K(\xi,d)\}
\qquad(0<s\le\xi).
\label{eq:E.3}
\end{align}

Integrating \eqref{eq:E.3} over \(s\in(0,\xi]\) and using the
integral formula of Lemma~\ref{lem:E.1}(b) gives
\[
\begin{aligned}
1-q_K(\xi,d)
&=
\int_0^\xi\{-\partial_sq_K(s,d)\}\,ds\\
&\ge
\{-\partial_\xi q_K(\xi,d)\}
\int_0^\xi\Bigl(\frac{s}{\xi}\Bigr)^{A_K}\,ds\\
&=
\frac{\xi\{-\partial_\xi q_K(\xi,d)\}}{1+A_K}.
\end{aligned}
\]

Multiplying this inequality by the positive number
\((1+A_K)/[2\{1-q_K(\xi,d)\}]\) shows that the normalized
marginal ratio at \((\xi,d)\) is at most \((1+A_K)/2\).
Since \(d\in\mathcal D_K\) and \(\xi>0\) were arbitrary,
\[
\frac12
\le
\beta_K
\le
\frac{1+A_K}{2}
=
\frac23+\frac{K-2}{3}\varphi(1)
<
\infty.
\]

In particular, \(0<\beta_K<\infty\), which is the hypothesis of
Lemma~\ref{lem:E.1}(e) and (f).

(ii) We derive the supremum identity and the formula for the
cost-benefit threshold.

Fix a parameter value \((\mu,\sigma)\) with normalized gap vector
\(d\) and a sample size \(m\in(0,N)\), and let
\(\xi=\bar\Delta\sqrt m/(\sigma\sqrt K)\).
Then \(\partial_m\xi=\xi/(2m)\).
Since \(q_K(\cdot,d)\) is differentiable by
Lemma~\ref{lem:E.1}(a), differentiating \eqref{eq:E.1} in \(m\) by
the chain rule and then factoring out \(1-q_K(\xi,d)\), which is
positive by Lemma~\ref{lem:E.1}(b), gives
\begin{align}
\begin{aligned}
\partial_m\mathrm{Reg}_K(m,\mu,\sigma)
&=
\bar\Delta\Bigl[
\{1-q_K(\xi,d)\}
-(N-m)\frac{\xi}{2m}
\{-\partial_\xi q_K(\xi,d)\}
\Bigr]\\
&=
\bar\Delta\{1-q_K(\xi,d)\}
\Bigl[
1-\frac{N-m}{m}\cdot
\frac{-(\xi/2)\partial_\xi q_K(\xi,d)}
     {1-q_K(\xi,d)}
\Bigr].
\end{aligned}
\label{eq:E.4}
\end{align}

Differentiating the definition
\(\mathrm{Reg}_K=m\bar\Delta+(N-m)r_K\) instead gives
\(\partial_m\mathrm{Reg}_K(m,\mu,\sigma)
=\bar\Delta-r_K(m,\mu,\sigma)
+(N-m)\partial_mr_K(m,\mu,\sigma)
=\bar\Delta\{1-\mathrm{CBR}_K(m,\mu,\sigma)\}\).
Equating the two expressions and dividing by
\(\bar\Delta\{1-q_K(\xi,d)\}>0\) shows that
\(\{1-\mathrm{CBR}_K(m,\mu,\sigma)\}/\{1-q_K(\xi,d)\}\)
equals the bracket in \eqref{eq:E.4}.
Hence \(\mathrm{CBR}_K(m,\mu,\sigma)\le1\) if and only if this
bracket is nonnegative, that is, if and only if the normalized
marginal ratio at \((\xi,d)\) is at most \(m/(N-m)\).
By the elasticity identity displayed before Theorem~\ref{fm:6.4},
this normalized marginal ratio equals
\(-m\partial_mr_K(m,\mu,\sigma)/
\{\bar\Delta-r_K(m,\mu,\sigma)\}\).

Every state is the state at sample size \(m\) of some parameter
value. Given \((\xi,d)\in(0,\infty)\times\mathcal D_K\), set
\(\mu_k=-d_k\) for \(k=1,\ldots,K\) and
\(\sigma=\sqrt m/(\xi\sqrt K)\).
Then \(\mu_*=-\min_kd_k=0\), \(\Delta_k=d_k\), and
\(\bar\Delta=K^{-1}\sum_kd_k=1\), so \((\mu,\sigma)\) is a
parameter value with normalized gap vector \(d\).
Its standardized signal at sample size \(u\in(0,N)\) is
\(\sqrt u/(\sigma\sqrt K)=\sqrt{u/m}\,\xi\), which equals
\(\xi\) at \(u=m\). We say that this parameter value realizes
\((\xi,d)\) at sample size \(m\). Conversely, every parameter
value yields a state at sample size \(m\).
Taking the supremum of the elasticity identity over parameter
values therefore gives the supremum over all states, which is
\(\beta_K\). This is the supremum identity in
Theorem~\ref{fm:6.4}, and it holds for every \(m\in(0,N)\).
Moreover,
\(\sup_{(\mu,\sigma)}\mathrm{CBR}_K(m,\mu,\sigma)\le1\)
holds exactly when the normalized marginal ratio is at most
\(m/(N-m)\) at every state. Using \(\beta_K<\infty\) and
\(N-m>0\), we obtain
\[
\begin{aligned}
\sup_{\substack{\mu\in\mathbb R^K,\,\sigma>0\\
                \bar\Delta>0}}
\mathrm{CBR}_K(m,\mu,\sigma)\le1
&\iff
\beta_K\le\frac{m}{N-m}\\
&\iff
m\ge\frac{\beta_K}{1+\beta_K}N.
\end{aligned}
\]

The set in the definition of \(m_K^{\mathrm{CBR}}\) is therefore
\([N\beta_K/(1+\beta_K),N)\), which is nonempty, and its infimum
is \(m_K^{\mathrm{CBR}}=N\beta_K/(1+\beta_K)\).
Because \(x\mapsto x/(1+x)\) is increasing on \([0,\infty)\)
and \(1/2\le\beta_K<\infty\), we have
\(m_K^{\mathrm{CBR}}\in[N/3,N)\).
The formula also gives
\(N-m_K^{\mathrm{CBR}}=N/(1+\beta_K)\), and hence
\begin{align}
\frac{m_K^{\mathrm{CBR}}}
     {N-m_K^{\mathrm{CBR}}}
=
\beta_K.
\label{eq:E.5}
\end{align}

(iii) We show that the cost-benefit threshold strictly dominates
every larger sample size.

Let \(m_K^{\mathrm{CBR}}<m<N\) and fix a parameter value
\((\mu,\sigma)\). For \(u\in(m_K^{\mathrm{CBR}},N)\), let
\(\xi=\bar\Delta\sqrt u/(\sigma\sqrt K)\) be the standardized
signal at sample size \(u\).
The normalized marginal ratio at \((\xi,d)\) is at most
\(\beta_K\), the factor \((N-u)/u\) is positive and strictly
decreasing in \(u\), and \(\beta_K>0\), so \eqref{eq:E.5} gives
\[
\begin{aligned}
\frac{N-u}{u}\cdot
\frac{-(\xi/2)\partial_\xi q_K(\xi,d)}
     {1-q_K(\xi,d)}
&\le
\frac{N-u}{u}\beta_K\\
&<
\frac{N-m_K^{\mathrm{CBR}}}
     {m_K^{\mathrm{CBR}}}\beta_K
=
1.
\end{aligned}
\]

By \eqref{eq:E.4} at sample size \(u\), the derivative
\(\partial_u\mathrm{Reg}_K(u,\mu,\sigma)\) is the product of
the positive number \(\bar\Delta\{1-q_K(\xi,d)\}\) and a
positive bracket, so it is positive for every
\(u\in(m_K^{\mathrm{CBR}},N)\).
This derivative is continuous in \(u\) by \eqref{eq:E.4} and
Lemma~\ref{lem:E.1}(a), and integrating it over
\([m_K^{\mathrm{CBR}},m]\) gives
\[
\begin{aligned}
&\mathrm{Reg}_K(m,\mu,\sigma)
-\mathrm{Reg}_K(m_K^{\mathrm{CBR}},\mu,\sigma)\\
&\qquad=
\int_{m_K^{\mathrm{CBR}}}^m
\partial_u\mathrm{Reg}_K(u,\mu,\sigma)\,du
>0.
\end{aligned}
\]

The parameter value was arbitrary, so \(m_K^{\mathrm{CBR}}\)
strictly dominates \(m\), and \(m\) is not admissible.

(iv) We show existence of global minimizers and the upper bound
on the envelope.

Fix a parameter value \((\mu,\sigma)\) with normalized gap vector
\(d\). By \eqref{eq:E.1},
\[
\mathrm{Reg}_K(u,\mu,\sigma)
=
\bar\Delta\Bigl[
N-(N-u)\Bigl\{
1-q_K\Bigl(
\frac{\bar\Delta\sqrt u}{\sigma\sqrt K},d
\Bigr)
\Bigr\}
\Bigr],
\qquad u\in(0,N),
\]
which is continuous in \(u\) by Lemma~\ref{lem:E.1}(a).
As \(u\downarrow0\), the standardized signal
\(\bar\Delta\sqrt u/(\sigma\sqrt K)\) converges to \(0\), so
\(q_K(\bar\Delta\sqrt u/(\sigma\sqrt K),d)\) converges to
\(q_K(0,d)=1\) by Lemma~\ref{lem:E.1}(a) and (b), and the product
\((N-u)\{1-q_K(\bar\Delta\sqrt u/(\sigma\sqrt K),d)\}\)
converges to \(0\).
For every \(u\in(0,N)\), the factor
\(1-q_K(\bar\Delta\sqrt u/(\sigma\sqrt K),d)\) lies in
\((0,1]\), because it is positive by Lemma~\ref{lem:E.1}(b)
and \(q_K\ge0\) by its definition with \(d_j\ge0\).
Hence the same product lies in \((0,N-u]\), and it converges to
\(0\) as \(u\uparrow N\).
Therefore \(u\mapsto\mathrm{Reg}_K(u,\mu,\sigma)\) extends to a
continuous function on \([0,N]\) whose value at both endpoints
is \(N\bar\Delta\), while its value on \((0,N)\) is strictly
smaller than \(N\bar\Delta\) because the product is positive
there. The extension attains its minimum on the compact interval
\([0,N]\), this minimum is smaller than \(N\bar\Delta\), and so
every minimizing point lies in \((0,N)\) and is a global minimizer.
Hence the set of global minimizers on \((0,N)\) is nonempty.

By step (iii),
\(\mathrm{Reg}_K(m_K^{\mathrm{CBR}},\mu,\sigma)
<\mathrm{Reg}_K(u,\mu,\sigma)\)
for every parameter value and every
\(u\in(m_K^{\mathrm{CBR}},N)\).
Since \(m_K^{\mathrm{CBR}}\in(0,N)\), no such \(u\) is a global minimizer of regret at any parameter value, and every global minimizer
lies in \((0,m_K^{\mathrm{CBR}}]\).
The union in the definition of \(m_K^{\mathrm{Env}}\) is
therefore a nonempty subset of \((0,m_K^{\mathrm{CBR}}]\),
and \(m_K^{\mathrm{Env}}\le m_K^{\mathrm{CBR}}\) for every
finite \(K\ge2\).

(v) We construct, for \(K\ge3\), a parameter value at which the
cost-benefit threshold is the unique global minimizer.

Let \(K\ge3\). Since \(0<\beta_K<\infty\) by step (i),
Lemma~\ref{lem:E.1}(f) applies and gives a state \((\xi,d)\)
at which the inequality in Lemma~\ref{lem:E.1}(e) holds for every
\(a\in(0,1)\), and we fix this state for the rest of this step.
This state is the same for every \(a\in(0,1)\).

Let \((\mu,\sigma)\) realize this state at sample size
\(m_K^{\mathrm{CBR}}\), namely \(\mu_k=-d_k\) for
\(k=1,\ldots,K\) and
\(\sigma=\sqrt{m_K^{\mathrm{CBR}}}/(\xi\sqrt K)\).
As shown in step (ii), \(\bar\Delta=1\), the normalized gap
vector is \(d\), and the standardized signal at sample size
\(u\in(0,N)\) is
\(\sqrt{u/m_K^{\mathrm{CBR}}}\,\xi\),
which equals \(\xi\) at \(u=m_K^{\mathrm{CBR}}\).
Substituting \(\bar\Delta=1\) and this standardized signal
into \eqref{eq:E.1} gives, for every \(u\in(0,N)\),
\[
\mathrm{Reg}_K(u,\mu,\sigma)
=
N-(N-u)\Bigl\{
1-q_K\Bigl(\sqrt{u/m_K^{\mathrm{CBR}}}\,\xi,d\Bigr)
\Bigr\}.
\]

At \(u=m_K^{\mathrm{CBR}}\) the argument of \(q_K\) is \(\xi\),
so
\(\mathrm{Reg}_K(m_K^{\mathrm{CBR}},\mu,\sigma)
=N-(N-m_K^{\mathrm{CBR}})\{1-q_K(\xi,d)\}\).
Subtracting this from the previous display cancels the constant \(N\), yielding
\[
\begin{aligned}
&\mathrm{Reg}_K(u,\mu,\sigma)
-\mathrm{Reg}_K(m_K^{\mathrm{CBR}},\mu,\sigma)\\
&\qquad=
(N-m_K^{\mathrm{CBR}})\{1-q_K(\xi,d)\}\\
&\qquad\quad-
(N-u)\Bigl\{
1-q_K\Bigl(\sqrt{u/m_K^{\mathrm{CBR}}}\,\xi,d\Bigr)
\Bigr\}.
\end{aligned}
\]

Hence, for every \(u\in(0,N)\), the inequality
\(\mathrm{Reg}_K(m_K^{\mathrm{CBR}},\mu,\sigma)
<\mathrm{Reg}_K(u,\mu,\sigma)\)
holds exactly when the right-hand side is positive, that is,
it is equivalent to
\begin{align}
(N-u)\Bigl\{
1-q_K\Bigl(\sqrt{u/m_K^{\mathrm{CBR}}}\,\xi,d\Bigr)
\Bigr\}
<
(N-m_K^{\mathrm{CBR}})\{1-q_K(\xi,d)\}.
\label{eq:E.6}
\end{align}

Let \(0<u<m_K^{\mathrm{CBR}}\) and
\(a=u/m_K^{\mathrm{CBR}}\in(0,1)\), so that
\(\sqrt a\,\xi\) is the standardized signal at sample size \(u\).
By \eqref{eq:E.5},
\(\beta_K(N-m_K^{\mathrm{CBR}})=m_K^{\mathrm{CBR}}\),
and therefore
\[
\begin{aligned}
(N-m_K^{\mathrm{CBR}})\{1+\beta_K(1-a)\}
&=
N-m_K^{\mathrm{CBR}}
+m_K^{\mathrm{CBR}}
-m_K^{\mathrm{CBR}}a\\
&=
N-u.
\end{aligned}
\]

At the fixed state, the inequality in Lemma~\ref{lem:E.1}(e)
reads
\(\{1-q_K(\xi,d)\}/\{1-q_K(\sqrt a\,\xi,d)\}
>1+\beta_K(1-a)\).
The number
\((N-m_K^{\mathrm{CBR}})\{1-q_K(\sqrt a\,\xi,d)\}\)
is positive because \(m_K^{\mathrm{CBR}}<N\) and
\(1-q_K(\sqrt a\,\xi,d)>0\) by Lemma~\ref{lem:E.1}(b).
Multiplying both sides of this inequality by this number and
using the last identity gives
\[
\begin{aligned}
&(N-m_K^{\mathrm{CBR}})\{1-q_K(\xi,d)\}\\
&\qquad>
(N-m_K^{\mathrm{CBR}})
\{1+\beta_K(1-a)\}
\{1-q_K(\sqrt a\,\xi,d)\}\\
&\qquad=
(N-u)\{1-q_K(\sqrt a\,\xi,d)\},
\end{aligned}
\]
which is \eqref{eq:E.6} at this \(u\).
For \(m_K^{\mathrm{CBR}}<u<N\), the inequality
\(\mathrm{Reg}_K(m_K^{\mathrm{CBR}},\mu,\sigma)
<\mathrm{Reg}_K(u,\mu,\sigma)\)
is the strict dominance shown in step (iii).
Therefore
\[
\mathrm{Reg}_K(m_K^{\mathrm{CBR}},\mu,\sigma)
<
\mathrm{Reg}_K(u,\mu,\sigma)
\qquad
\text{for every }u\in(0,N)\setminus\{m_K^{\mathrm{CBR}}\},
\]
and \(m_K^{\mathrm{CBR}}\) is the unique global minimizer under
this parameter value.
Since \(m_K^{\mathrm{Env}}\) is at least every global minimizer of regret at
every parameter value,
\(m_K^{\mathrm{Env}}\ge m_K^{\mathrm{CBR}}\), and together with
step (iv) and the formula of step (ii) this gives
\(m_K^{\mathrm{Env}}=m_K^{\mathrm{CBR}}
=N\beta_K/(1+\beta_K)\) for \(K\ge3\).

(vi) We treat the binary case \(K=2\).

For \(K=2\), let \(\tau:=\mu_1-\mu_2\).
Under common variance, Assumption~\ref{asm:6.1} gives
\[
\widehat\mu_1-\widehat\mu_2
\sim
\mathcal N\left(\tau,\frac{4\sigma^2}{m}\right),
\qquad
\bar\Delta=\frac{|\tau|}{2}.
\]
Thus the empirical success rule and the regret function coincide
with those in the binary model:
\[
\mathrm{Reg}_2(m,\mu,\sigma)
=
\mathrm{Reg}(m,\tau,(\sigma,\sigma)).
\]
Conversely, for any binary-treatment parameter value with nonzero
mean difference \(\tau\) and arm-specific standard deviations
\((s_1,s_0)\), choose \(\mu_1=\tau\), \(\mu_2=0\), and the
common standard deviation
\[
\sigma=\sqrt{\frac{s_1^2+s_0^2}{2}}.
\]
This choice satisfies \(\bar\Delta>0\) and produces the same
distribution of the difference in sample means and hence the
same regret function for every \(m\in(0,N)\).
Therefore the common-variance problem with \(K=2\) has exactly
the same family of regret functions as the binary problem in
Section~\ref{sec:envelope-admissible}.
Equation~\eqref{eq:11} gives \(m_2^{\mathrm{CBR}}=N/3\), and
step (ii) then implies \(\beta_2=1/2\).
Theorem~\ref{thm:3.1} gives \(m_2^{\mathrm{Env}}=N/3\).
This completes the proof of Theorem~\ref{fm:6.4}.
\end{proof}

\medskip

\begin{proof}[Proof of Corollary~\ref{cor:6.5}]

Step (iii) of the proof of Theorem~\ref{fm:6.4} shows that
\(m_K^{\mathrm{CBR}}\) strictly dominates every
\(m\in(m_K^{\mathrm{CBR}},N)\), so no such \(m\) is admissible.
It remains to show that every sample size in
\((0,m_K^{\mathrm{CBR}}]\) is admissible.

Fix \(m_2\in(0,m_K^{\mathrm{CBR}}]\) and \(m_1\in(0,N)\)
with \(m_1\ne m_2\).
By definition, \(m_1\) weakly dominates \(m_2\) only if
\(\mathrm{Reg}_K(m_1,\mu,\sigma)
\le\mathrm{Reg}_K(m_2,\mu,\sigma)\)
for every parameter value.
It therefore suffices to find one parameter value \((\mu,\sigma)\),
which may depend on \(m_1\), such that
\(\mathrm{Reg}_K(m_2,\mu,\sigma)
<\mathrm{Reg}_K(m_1,\mu,\sigma)\).
We treat larger and smaller competitors separately.

Suppose first that \(m_1>m_2\).
By \eqref{eq:E.1},
\(\mathrm{Reg}_K(m,\mu,\sigma)/\bar\Delta
=m+(N-m)\,q_K(\xi,d)\)
with \(\xi\) the standardized signal at sample size \(m\), so
\[
\begin{aligned}
&\frac{
\mathrm{Reg}_K(m_1,\mu,\sigma)
-\mathrm{Reg}_K(m_2,\mu,\sigma)}
{\bar\Delta}\\
&\qquad=
m_1-m_2
+(N-m_1)\,q_K\Bigl(
\frac{\bar\Delta\sqrt{m_1}}{\sigma\sqrt K},d
\Bigr)\\
&\qquad\quad-
(N-m_2)\,q_K\Bigl(
\frac{\bar\Delta\sqrt{m_2}}{\sigma\sqrt K},d
\Bigr).
\end{aligned}
\]

Fix \(\mu\in\mathbb R^K\) with \(\bar\Delta>0\) and let
\(\sigma\downarrow0\).
The normalized gap vector \(d\) does not change, and both
standardized signals diverge to infinity.
By Lemma~\ref{lem:E.1}(c), both values of \(q_K\) converge to
zero, while the factors \(N-m_1\) and \(N-m_2\) stay fixed,
so the right-hand side converges to \(m_1-m_2>0\).
Hence
\(\mathrm{Reg}_K(m_2,\mu,\sigma)
<\mathrm{Reg}_K(m_1,\mu,\sigma)\)
for every sufficiently small \(\sigma>0\), and \(m_1\) does
not weakly dominate \(m_2\).

Suppose next that \(m_1<m_2\), and let
\(a=m_1/m_2\in(0,1)\).
Under a parameter value \((\mu,\sigma)\), let \(\xi\) be the
standardized signal at sample size \(m_2\).
The standardized signal at sample size \(m_1\) is then
\(\bar\Delta\sqrt{m_1}/(\sigma\sqrt K)=\sqrt a\,\xi\).
Subtracting the two regrets in \eqref{eq:E.1} from
\(N\bar\Delta\) and dividing by \(\bar\Delta>0\) gives
\[
\begin{aligned}
&\mathrm{Reg}_K(m_2,\mu,\sigma)
<
\mathrm{Reg}_K(m_1,\mu,\sigma)\\
&\qquad\iff
(N-m_1)\{1-q_K(\sqrt a\,\xi,d)\}
<
(N-m_2)\{1-q_K(\xi,d)\}.
\end{aligned}
\]

The numbers \(N-m_2\) and \(1-q_K(\sqrt a\,\xi,d)\) are
positive, the latter by Lemma~\ref{lem:E.1}(b).
Dividing by their product gives
\begin{align}
\begin{aligned}
&\mathrm{Reg}_K(m_2,\mu,\sigma)
<
\mathrm{Reg}_K(m_1,\mu,\sigma)\\
&\qquad\iff
\frac{1-q_K(\xi,d)}
     {1-q_K(\sqrt a\,\xi,d)}
>
\frac{N-m_1}{N-m_2}.
\end{aligned}
\label{eq:E.7}
\end{align}

We bound the right-hand side of \eqref{eq:E.7} by the
cost-benefit threshold.
Because \(x\mapsto x/(N-x)\) is strictly increasing on
\((0,N)\) and \(m_2\le m_K^{\mathrm{CBR}}\),
\eqref{eq:E.5} gives \(m_2/(N-m_2)\le\beta_K\), with equality
exactly when \(m_2=m_K^{\mathrm{CBR}}\).
Therefore
\begin{align}
\frac{N-m_1}{N-m_2}
=
1+\frac{m_2-m_1}{N-m_2}
=
1+\frac{m_2}{N-m_2}(1-a)
\le
1+\beta_K(1-a),
\label{eq:E.8}
\end{align}
with equality when \(m_2=m_K^{\mathrm{CBR}}\).

We now choose the parameter value.
Step (i) of the proof of Theorem~\ref{fm:6.4} gives
\(0<\beta_K<\infty\), so Lemma~\ref{lem:E.1}(e) applies with
\(a=m_1/m_2\) and provides a state \((\xi,d)\) that satisfies
its inequality.
Let \((\mu,\sigma)\) realize this state at sample size \(m_2\),
as in step (ii) of that proof.
Under this parameter value the standardized signal at sample
size \(m_2\) is \(\xi\) and that at sample size \(m_1\) is
\(\sqrt{m_1/m_2}\,\xi=\sqrt a\,\xi\), so
Lemma~\ref{lem:E.1}(e) and \eqref{eq:E.8} give
\[
\frac{1-q_K(\xi,d)}
     {1-q_K(\sqrt a\,\xi,d)}
>
1+\beta_K(1-a)
\ge
\frac{N-m_1}{N-m_2},
\]
and \eqref{eq:E.7} yields
\(\mathrm{Reg}_K(m_2,\mu,\sigma)
<\mathrm{Reg}_K(m_1,\mu,\sigma)\).
The state, and hence the parameter value, depends on \(m_1\)
through \(a\), which suffices for nondomination.
At the endpoint \(m_2=m_K^{\mathrm{CBR}}\),
\eqref{eq:E.8} is an equality, so a non-strict version of
the inequality in Lemma~\ref{lem:E.1}(e) would give only
\(\mathrm{Reg}_K(m_2,\mu,\sigma)
\le\mathrm{Reg}_K(m_1,\mu,\sigma)\),
which does not rule out weak dominance.
The strict inequality in Lemma~\ref{lem:E.1}(e) is what excludes
weak dominance there.

Combining the two cases, no sample size in \((0,N)\) weakly
dominates any \(m_2\in(0,m_K^{\mathrm{CBR}}]\), so every point
of \((0,m_K^{\mathrm{CBR}}]\) is admissible.
Together with step (iii) of the proof of Theorem~\ref{fm:6.4},
the admissible sample sizes are exactly those in
\((0,m_K^{\mathrm{CBR}}]\).
That step also shows that this set is a complete class among
deterministic sample-size choices.
Every complete class must contain every admissible sample size,
so no proper subset of \((0,m_K^{\mathrm{CBR}}]\) can be complete.
Moreover, any complete class strictly containing
\((0,m_K^{\mathrm{CBR}}]\) has this set as a proper complete
subclass and is therefore not minimal.
Hence \((0,m_K^{\mathrm{CBR}}]\) is the unique minimal complete
class, which proves Corollary~\ref{cor:6.5}.

\end{proof}

\subsection{Proof of Lemma~\ref{lem:E.1}}\label{sec:appendix-E3}

The proof uses only the definitions of \(q_K\), \(\mathcal D_K\) and \(\beta_K\), and it uses neither Theorem~\ref{fm:6.4} nor Corollary~\ref{cor:6.5}. Throughout, \(Z_1,\ldots,Z_K\) are independent standard normal random variables, as in the definition of \(q_K\). Section~\ref{sec:appendix-E31} derives an integral representation of \(\partial_\xi q_K\) and proves parts (a)--(c). Section~\ref{sec:appendix-E32} computes the pair integrals appearing in this representation, and Section~\ref{sec:appendix-E33} uses them to prove part (d). Sections~\ref{sec:appendix-E34}--\ref{sec:appendix-E36} establish two properties of \(\varphi/\Phi\), a representation of \(q_K\) at a normalized gap vector with a single best arm, and a scalar comparison at that normalized gap vector. Section~\ref{sec:appendix-E37} combines the scalar comparison with part (d) to prove part (e), and Section~\ref{sec:appendix-E38} proves part (f).

\subsubsection{\texorpdfstring{Integral representation and parts (a)--(c)}{Integral representation and parts (a)–(c)}}\label{sec:appendix-E31}

Fix \(d\in\mathcal D_K\). Up to an event of probability zero, arm \(j\) attains \(\max_k\{Z_k-\xi d_k\}\) if and only if \(Z_k<Z_j-\xi d_j+\xi d_k\) for every \(k\ne j\). Conditioning on the value \(x\) of \(Z_j-\xi d_j\), whose density is \(\varphi(x+\xi d_j)\), and using the independence of \(Z_1,\ldots,Z_K\), we obtain
\[
q_K(\xi,d)=\sum_{j=1}^Kd_j\int_{\mathbb R}\varphi(x+\xi d_j)\prod_{k\ne j}\Phi(x+\xi d_k)\,dx.
\]

Both sides are well defined for every real \(\xi\), and the same argument shows that they agree on \(\mathbb R\). Differentiating the integral representation with respect to \(\xi\), as justified below, gives a term in which the density of arm \(j\) is differentiated and terms in which the distribution function of another arm is differentiated:
\[
\begin{aligned}
\partial_\xi q_K(\xi,d)
&=\sum_{j=1}^Kd_j^2\int_{\mathbb R}\varphi'(x+\xi d_j)\prod_{k\ne j}\Phi(x+\xi d_k)\,dx\\
&\quad+\sum_{j=1}^K\sum_{k\ne j}d_jd_k\int_{\mathbb R}\varphi(x+\xi d_j)\varphi(x+\xi d_k)\prod_{\ell\ne j,k}\Phi(x+\xi d_\ell)\,dx.
\end{aligned}
\]

Integrating by parts in $x$ in the first sum and using the vanishing boundary terms, we obtain
\[
\int_{\mathbb R}\varphi'(x+\xi d_j)\prod_{k\ne j}\Phi(x+\xi d_k)\,dx=-\sum_{k\ne j}\int_{\mathbb R}\varphi(x+\xi d_j)\varphi(x+\xi d_k)\prod_{\ell\ne j,k}\Phi(x+\xi d_\ell)\,dx.
\]

For distinct arms \(j\) and \(k\), we call
\[
J_{jk}(\xi,d):=\int_{\mathbb R}\varphi(x+\xi d_j)\varphi(x+\xi d_k)\prod_{\ell\ne j,k}\Phi(x+\xi d_\ell)\,dx
\]
the pair integral of arms \(j\) and \(k\), where the empty product equals \(1\) if \(K=2\). The pair integral is positive and symmetric in \(j\) and \(k\). After the substitution, the coefficient of \(J_{jk}(\xi,d)\) for the ordered pair \((j,k)\) is \(d_jd_k-d_j^2\). Combining the ordered pairs \((j,k)\) and \((k,j)\) by symmetry gives the coefficient \((d_jd_k-d_j^2)+(d_kd_j-d_k^2)=-(d_j-d_k)^2\), and hence
\begin{align}
-\partial_\xi q_K(\xi,d)=\sum_{1\le j<k\le K}(d_j-d_k)^2J_{jk}(\xi,d).
\label{eq:E.9}
\end{align}

The interchange of differentiation and integration and the integration by parts are justified as follows. If \(d\in\mathcal D_K\), then \(0\le d_k\le K\) for every \(k\). When \(\xi\) is restricted to a bounded interval, there is a constant \(b>0\) such that \(|\xi d_k|\le b\) for every \(k\) and every \(d\in\mathcal D_K\), and \((x+v)^2\ge x^2/2-v^2\) implies \(\varphi(x+v)\le e^{b^2/2}e^{-x^2/4}/\sqrt{2\pi}\) whenever \(|v|\le b\). Combined with \(\varphi'(t)=-t\varphi(t)\) and \(0\le\Phi\le1\), this bound shows that the integrand and its first and second derivatives with respect to \(\xi\) are bounded by a constant multiple of the product of a polynomial in \(x\) and \(e^{-x^2/4}\). This integrable dominating function permits the interchange of differentiation and integration, and the boundary terms of the integration by parts vanish as \(|x|\to\infty\). The same bound and dominated convergence show that \(q_K\), \(\partial_\xi q_K\) and \(\partial_\xi^2q_K\) are continuous on \(\mathbb R\times\mathcal D_K\). In particular, \(q_K(\cdot,d)\) is twice continuously differentiable on \(\mathbb R\), and its restriction to \([0,\infty)\) has the properties stated in part (a), where \(\partial_\xi q_K(0,d)\) is the one-sided derivative at \(\xi=0\).

For part (b), the components of \(d\in\mathcal D_K\) are not all equal, because equal components would all vanish by \(\min_kd_k=0\), which contradicts \(\sum_kd_k=K\). Hence some term in \eqref{eq:E.9} has a positive coefficient, and since every pair integral is positive, \(-\partial_\xi q_K(\xi,d)>0\) for every \(\xi\ge0\). At \(\xi=0\), each arm is selected with probability \(1/K\), so \(q_K(0,d)=K^{-1}\sum_jd_j=1\). The fundamental theorem of calculus then gives, for every \(\xi>0\),
\begin{align}
1-q_K(\xi,d)=\int_0^\xi\{-\partial_sq_K(s,d)\}\,ds>0,
\label{eq:E.10}
\end{align}
which completes the proof of part (b). Since \(-\partial_\xi q_K(\cdot,d)\) is continuous on \([0,\infty)\) and positive at \(\xi=0\), \eqref{eq:E.10} gives \(1-q_K(\xi,d)=\xi\{-\partial_\xi q_K(0,d)\}+o(\xi)\) and \(-(\xi/2)\partial_\xi q_K(\xi,d)=(\xi/2)\{-\partial_\xi q_K(0,d)\}+o(\xi)\) as \(\xi\downarrow0\). Therefore, for every \(d\in\mathcal D_K\),
\begin{align}
\lim_{\xi\downarrow0}\frac{-(\xi/2)\partial_\xi q_K(\xi,d)}{1-q_K(\xi,d)}=\frac12.
\label{eq:E.11}
\end{align}

For part (c), choose an arm \(i\) with \(d_i=0\). An arm \(j\) with \(d_j>0\) can be selected only if \(Z_j-\xi d_j\ge Z_i\), that is, only if \(Z_j-Z_i\ge\xi d_j\), and \(Z_j-Z_i\) is normal with mean zero and variance \(2\). For \(t>0\), moreover, \(t\Phi(-t)=t\int_t^\infty\varphi(v)\,dv\le\int_t^\infty v\varphi(v)\,dv=\varphi(t)\le1/\sqrt{2\pi}\). Applying this bound with \(t=\xi d_j/\sqrt2\) gives \(d_j\Phi(-\xi d_j/\sqrt2)\le1/(\sqrt\pi\,\xi)\), and since at most \(K-1\) arms satisfy \(d_j>0\),
\begin{align}
0\le q_K(\xi,d)\le\sum_{j:d_j>0}d_j\Phi\Bigl(-\frac{\xi d_j}{\sqrt2}\Bigr)\le\frac{K}{\sqrt\pi\,\xi}\qquad(\xi>0).
\label{eq:E.12}
\end{align}

The right-hand side converges to zero as \(\xi\to\infty\), which proves part (c). The bound \eqref{eq:E.12} does not depend on \(d\).

\subsubsection{The pair integrals}\label{sec:appendix-E32}

Fix distinct arms \(j\) and \(k\), a real number \(\xi\) and \(d\in\mathcal D_K\), and let \(c_\ell:=\xi\{d_\ell-(d_j+d_k)/2\}\) for \(\ell\ne j,k\). Completing the square gives
\[
\begin{aligned}
(x+\xi d_j)^2+(x+\xi d_k)^2
&=2x^2+2\xi(d_j+d_k)x+\xi^2(d_j^2+d_k^2)\\
&=2\Bigl\{x+\frac{\xi(d_j+d_k)}{2}\Bigr\}^2+\xi^2\Bigl\{d_j^2+d_k^2-\frac{(d_j+d_k)^2}{2}\Bigr\}\\
&=2\Bigl\{x+\frac{\xi(d_j+d_k)}{2}\Bigr\}^2+\frac{\xi^2(d_j-d_k)^2}{2}.
\end{aligned}
\]

Since \(\varphi(t)=(2\pi)^{-1/2}e^{-t^2/2}\), it follows that
\[
\varphi(x+\xi d_j)\varphi(x+\xi d_k)=\frac{1}{2\pi}\exp\Bigl[-\Bigl\{x+\frac{\xi(d_j+d_k)}{2}\Bigr\}^2-\frac{\xi^2(d_j-d_k)^2}{4}\Bigr].
\]

The change of variables \(y=x+\xi(d_j+d_k)/2\) satisfies \(dx=dy\) and \(x+\xi d_\ell=y+c_\ell\) for \(\ell\ne j,k\), and therefore
\begin{align}
J_{jk}(\xi,d)=\frac{e^{-\xi^2(d_j-d_k)^2/4}}{2\pi}\int_{\mathbb R}e^{-y^2}\prod_{\ell\ne j,k}\Phi(y+c_\ell)\,dy.
\label{eq:E.13}
\end{align}

Let
\[
D_{jk}:=\int_{\mathbb R}e^{-v^2}\prod_{\ell\ne j,k}\Phi(v+c_\ell)\,dv,\qquad p_{jk}(y):=\frac{1}{D_{jk}}e^{-y^2}\prod_{\ell\ne j,k}\Phi(y+c_\ell)\quad(y\in\mathbb R).
\]

Since \(0<\Phi<1\), we have \(0<D_{jk}\le\sqrt\pi\), and \(p_{jk}\) is a positive probability density on \(\mathbb R\). We write \(\mathbb E_{jk}\) for the expectation of a function of \(y\) when \(y\) has the density \(p_{jk}\), and we suppress the dependence of \(D_{jk}\), \(p_{jk}\) and \(\mathbb E_{jk}\) on \(\xi\) and \(d\). Since \(p_{jk}>0\) and \(\Phi'=\varphi\), the logarithmic derivative of \(p_{jk}\) is
\[
\frac{d}{dy}\log p_{jk}(y)=-2y+\sum_{\ell\ne j,k}\frac{\varphi(y+c_\ell)}{\Phi(y+c_\ell)},
\]
and \(p_{jk}'(y)=p_{jk}(y)\,(d/dy)\log p_{jk}(y)\).

The expectations below are finite. The bound \(\Phi\le1\) gives \(p_{jk}(y)\le e^{-y^2}/D_{jk}\). Multiplying \(p_{jk}\) by \(\varphi/\Phi\) replaces one factor \(\Phi\) by \(\varphi\), so \(\varphi\le(2\pi)^{-1/2}\) and \(\Phi\le1\) give
\[
p_{jk}(y)\frac{\varphi(y+c_\ell)}{\Phi(y+c_\ell)}=\frac{1}{D_{jk}}e^{-y^2}\varphi(y+c_\ell)\prod_{m\ne j,k,\ell}\Phi(y+c_m)\le\frac{e^{-y^2}}{\sqrt{2\pi}\,D_{jk}}.
\]

Hence \(y^2p_{jk}(y)\), \(p_{jk}(y)\varphi(y+c_\ell)/\Phi(y+c_\ell)\) and \(|y|p_{jk}(y)\varphi(y+c_\ell)/\Phi(y+c_\ell)\) are all bounded by constant multiples of \((1+y^2)e^{-y^2}\) and are integrable. In particular, \(\mathbb E_{jk}[y^2]\), \(\mathbb E_{jk}[\varphi(y+c_\ell)/\Phi(y+c_\ell)]\) and \(\mathbb E_{jk}[|y|\varphi(y+c_\ell)/\Phi(y+c_\ell)]\) are finite, and \(p_{jk}'\) and \((yp_{jk})'=p_{jk}+yp_{jk}'\) are integrable. The bound \(p_{jk}(y)\le e^{-y^2}/D_{jk}\) also gives \(p_{jk}(\pm R)\to0\) and \(Rp_{jk}(\pm R)\to0\) as \(R\to\infty\). Since \(p_{jk}'\) and \((yp_{jk})'\) are integrable, dominated convergence and the fundamental theorem of calculus give
\[
\begin{aligned}
\int_{\mathbb R}p_{jk}'(y)\,dy&=\lim_{R\to\infty}\{p_{jk}(R)-p_{jk}(-R)\}=0,\\
\int_{\mathbb R}\{yp_{jk}(y)\}'\,dy&=\lim_{R\to\infty}\{Rp_{jk}(R)+Rp_{jk}(-R)\}=0.
\end{aligned}
\]

Substituting \(p_{jk}'=p_{jk}\,(d/dy)\log p_{jk}\) into the first equality gives
\[
0=\mathbb E_{jk}\Bigl[\frac{d}{dy}\log p_{jk}(y)\Bigr]=-2\mathbb E_{jk}[y]+\mathbb E_{jk}\Bigl[\sum_{\ell\ne j,k}\frac{\varphi(y+c_\ell)}{\Phi(y+c_\ell)}\Bigr],
\]
and the second equality together with \(\int_{\mathbb R}p_{jk}=1\) gives
\[
0=1+\mathbb E_{jk}\Bigl[y\frac{d}{dy}\log p_{jk}(y)\Bigr]=1-2\mathbb E_{jk}[y^2]+\mathbb E_{jk}\Bigl[y\sum_{\ell\ne j,k}\frac{\varphi(y+c_\ell)}{\Phi(y+c_\ell)}\Bigr].
\]

Consequently,
\begin{align}
\begin{aligned}
\mathbb E_{jk}\Bigl[\sum_{\ell\ne j,k}\frac{\varphi(y+c_\ell)}{\Phi(y+c_\ell)}\Bigr]&=2\mathbb E_{jk}[y],\\
\mathbb E_{jk}\Bigl[y\sum_{\ell\ne j,k}\frac{\varphi(y+c_\ell)}{\Phi(y+c_\ell)}\Bigr]&=2\mathbb E_{jk}[y^2]-1.
\end{aligned}
\label{eq:E.14}
\end{align}

We also record a pointwise bound. Since \(\{t\varphi(t)\}'=(1-t^2)\varphi(t)\), the function \(t\varphi(t)\) attains its maximum at \(t=1\) and its minimum at \(t=-1\), so \(|t|\varphi(t)\le\varphi(1)\). For \(\gamma\ge0\), let \(M_\gamma:=\sup_{z\in\mathbb R}(z-\gamma)\varphi(z)/\Phi(z)\), the supremum appearing in part (d). The function \((z-\gamma)\varphi(z)/\Phi(z)\) is nonpositive for \(z\le\gamma\), and for \(z>\gamma\ge0\) it is positive and satisfies \(\Phi(z)\ge1/2\) and \((z-\gamma)\varphi(z)\le z\varphi(z)\le\varphi(1)\). Hence \(0<M_\gamma\le2\varphi(1)\), and for every real \(t\) and every \(\gamma\ge0\),
\begin{align}
\frac{t\varphi(t)}{\Phi(t)}=\gamma\frac{\varphi(t)}{\Phi(t)}+\frac{(t-\gamma)\varphi(t)}{\Phi(t)}\le\gamma\frac{\varphi(t)}{\Phi(t)}+M_\gamma.
\label{eq:E.15}
\end{align}

\subsubsection{Proof of part (d)}\label{sec:appendix-E33}

Fix \(d\in\mathcal D_K\), \(\xi>0\) and \(\gamma\ge0\). Permuting the arms permutes the components of \(d\), leaves \(\mathcal D_K\) invariant and, because \(Z_1,\ldots,Z_K\) are independent and identically distributed, does not change \(q_K(\xi,d)\) or the left-hand side of part (d). We may therefore number the arms in increasing order of \(d_k\), with tied arms kept in their original order, so that \(d_j\le d_k\) whenever \(j<k\). For distinct arms \(j\), \(k\) and \(\ell\), we call
\[
J_{jk\ell}(\xi,d):=\int_{\mathbb R}\varphi(x+\xi d_j)\varphi(x+\xi d_k)\varphi(x+\xi d_\ell)\prod_{i\ne j,k,\ell}\Phi(x+\xi d_i)\,dx
\]
the three-arm integral of arms \(j\), \(k\) and \(\ell\). It is positive and symmetric in the three arms.

We first express the derivative of a pair integral through pair and three-arm integrals. Fix \(j\ne k\). Differentiation under the integral sign is justified as in Section~\ref{sec:appendix-E31}. Since \(\partial_\xi\varphi(x+\xi d_j)=-d_j(x+\xi d_j)\varphi(x+\xi d_j)\), \(\partial_\xi\varphi(x+\xi d_k)=-d_k(x+\xi d_k)\varphi(x+\xi d_k)\) and \(\partial_\xi\Phi(x+\xi d_\ell)=d_\ell\varphi(x+\xi d_\ell)\), the product rule and \(d_j(x+\xi d_j)+d_k(x+\xi d_k)=(d_j+d_k)x+\xi(d_j^2+d_k^2)\) give
\[
\begin{aligned}
\partial_\xi J_{jk}(\xi,d)
&=-\int_{\mathbb R}\{(d_j+d_k)x+\xi(d_j^2+d_k^2)\}\varphi(x+\xi d_j)\varphi(x+\xi d_k)\prod_{\ell\ne j,k}\Phi(x+\xi d_\ell)\,dx \\
&\qquad +\sum_{\ell\ne j,k}d_\ell J_{jk\ell}(\xi,d)\\
&=-(d_j+d_k)\int_{\mathbb R}x\,\varphi(x+\xi d_j)\varphi(x+\xi d_k)\prod_{\ell\ne j,k}\Phi(x+\xi d_\ell)\,dx\\
&\quad-\xi(d_j^2+d_k^2)J_{jk}(\xi,d)+\sum_{\ell\ne j,k}d_\ell J_{jk\ell}(\xi,d).
\end{aligned}
\]

The integral containing the factor \(x\) is neither a pair integral nor a three-arm integral, and we rewrite it by differentiating the integrand of \(J_{jk}(\xi,d)\) with respect to \(x\). Using \(\varphi'(t)=-t\varphi(t)\), \(\Phi'=\varphi\) and \((x+\xi d_j)+(x+\xi d_k)=2x+\xi(d_j+d_k)\), we obtain
\[
\begin{aligned}
&\frac{d}{dx}\Bigl\{\varphi(x+\xi d_j)\varphi(x+\xi d_k)\prod_{\ell\ne j,k}\Phi(x+\xi d_\ell)\Bigr\}\\
&\qquad=-\{2x+\xi(d_j+d_k)\}\varphi(x+\xi d_j)\varphi(x+\xi d_k)\prod_{\ell\ne j,k}\Phi(x+\xi d_\ell)\\
&\qquad\quad+\sum_{\ell\ne j,k}\varphi(x+\xi d_j)\varphi(x+\xi d_k)\varphi(x+\xi d_\ell)\prod_{i\ne j,k,\ell}\Phi(x+\xi d_i).
\end{aligned}
\]

By the fundamental theorem of calculus, the integral of the left-hand side over \([-R,R]\) is the difference between the values of the function in braces at \(x=R\) and at \(x=-R\). This function is at most \(\varphi(x+\xi d_j)/\sqrt{2\pi}\) because \(\Phi\le1\), so the difference converges to zero as \(R\to\infty\). Each term on the right-hand side is bounded by the product of a linear function of \(|x|\) and a Gaussian function and is therefore integrable. Integrating over \(\mathbb R\) thus gives
\[
0=-2\int_{\mathbb R}x\,\varphi(x+\xi d_j)\varphi(x+\xi d_k)\prod_{\ell\ne j,k}\Phi(x+\xi d_\ell)\,dx-\xi(d_j+d_k)J_{jk}(\xi,d)+\sum_{\ell\ne j,k}J_{jk\ell}(\xi,d),
\]
and solving for the integral containing the factor \(x\) gives
\[
\int_{\mathbb R}x\,\varphi(x+\xi d_j)\varphi(x+\xi d_k)\prod_{\ell\ne j,k}\Phi(x+\xi d_\ell)\,dx=\frac12\sum_{\ell\ne j,k}J_{jk\ell}(\xi,d)-\frac{\xi(d_j+d_k)}{2}J_{jk}(\xi,d).
\]

Substituting this expression into the formula for \(\partial_\xi J_{jk}(\xi,d)\) yields
\[
\begin{aligned}
\partial_\xi J_{jk}(\xi,d)
&=-\frac{d_j+d_k}{2}\sum_{\ell\ne j,k}J_{jk\ell}(\xi,d)+\frac{\xi(d_j+d_k)^2}{2}J_{jk}(\xi,d)\\
&\quad-\xi(d_j^2+d_k^2)J_{jk}(\xi,d)+\sum_{\ell\ne j,k}d_\ell J_{jk\ell}(\xi,d).
\end{aligned}
\]

The coefficient of the pair integral is
\[
\frac{\xi(d_j+d_k)^2}{2}-\xi(d_j^2+d_k^2)=-\frac{\xi}{2}\bigl(d_j^2-2d_jd_k+d_k^2\bigr)=-\frac{\xi(d_j-d_k)^2}{2},
\]
and the coefficient of the three-arm integral of arms \(j\), \(k\) and \(\ell\) is \(d_\ell-(d_j+d_k)/2\). Hence
\begin{align}
\partial_\xi J_{jk}(\xi,d)=-\frac{\xi(d_j-d_k)^2}{2}J_{jk}(\xi,d)+\sum_{\ell\ne j,k}\Bigl(d_\ell-\frac{d_j+d_k}{2}\Bigr)J_{jk\ell}(\xi,d).
\label{eq:E.16}
\end{align}

The argument holds for every real \(\xi\). Differentiating \eqref{eq:E.9} with respect to \(\xi\) and substituting \eqref{eq:E.16} for each pair integral gives
\begin{align}
\begin{aligned}
-\partial_\xi^2q_K(\xi,d)
&=-\sum_{j<k}\frac{\xi(d_k-d_j)^4}{2}J_{jk}(\xi,d)\\
&\quad+\sum_{j<k}(d_k-d_j)^2\sum_{\ell\ne j,k}\Bigl(d_\ell-\frac{d_j+d_k}{2}\Bigr)J_{jk\ell}(\xi,d).
\end{aligned}
\label{eq:E.17}
\end{align}

We next bound the coefficients of the three-arm integrals in \eqref{eq:E.17}. A set of three distinct arms appears three times in the double sum, once for each choice of the arm that plays the role of \(\ell\). Since the three-arm integral is symmetric, its total coefficient is the sum over these three choices of the squared difference between the two arms of the pair multiplied by the difference between the component of the third arm and the midpoint of the pair. This coefficient depends only on differences of components and is unchanged when a common constant is subtracted from the three components. Subtracting the smallest of them, we write the resulting values in increasing order as \(0\le v\le w\). The pair with values \(0\) and \(v\) contributes \(v^2(w-v/2)\), the pair with values \(0\) and \(w\) contributes \(w^2(v-w/2)\), and the pair with values \(v\) and \(w\) contributes \(-(w-v)^2(v+w)/2\). Using \((w-v)^2(v+w)=v^3+w^3-v^2w-vw^2\), the total coefficient is
\[
v^2\Bigl(w-\frac{v}{2}\Bigr)+w^2\Bigl(v-\frac{w}{2}\Bigr)-\frac{(w-v)^2(v+w)}{2}=\frac32vw(v+w)-v^3-w^3.
\]

Orient each pair so that \(d_j\le d_k\), and replace \(d_\ell-(d_j+d_k)/2\) by \(\{d_\ell-d_j+(d_k-d_j)/2\}/3\). The three contributions then become \(v^2(w+v/2)/3\), \(w^2(v+w/2)/3\) and \((w-v)^2\{-v+(w-v)/2\}/3\), and subtracting the original coefficient from their sum and expanding gives
\[
\begin{aligned}
&\frac13\Bigl[v^2\Bigl(w+\frac{v}{2}\Bigr)+w^2\Bigl(v+\frac{w}{2}\Bigr)+(w-v)^2\Bigl(-v+\frac{w-v}{2}\Bigr)\Bigr]-\Bigl\{\frac32vw(v+w)-v^3-w^3\Bigr\}\\
&\qquad=\frac23v^3+\frac43w^3-2vw^2=\frac23(w-v)^2(v+2w)\ge0.
\end{aligned}
\]

Since the three-arm integrals are positive, the coefficient of each of them may be replaced by this larger value. Under the numbering fixed above, \(j<k\) implies \(d_j\le d_k\), and the replaced coefficient \(d_\ell-d_j+(d_k-d_j)/2\) uses only this orientation of the pair and no restriction on the third arm \(\ell\). For a pair with \(d_j=d_k\), the factor \((d_k-d_j)^2\) vanishes, so the orientation is immaterial. The bound can therefore be written again as a sum over \(j<k\) and \(\ell\ne j,k\),
\begin{align}
\begin{aligned}
-\partial_\xi^2q_K(\xi,d)
&\le-\sum_{j<k}\frac{\xi(d_k-d_j)^4}{2}J_{jk}(\xi,d)\\
&\quad+\frac13\sum_{j<k}(d_k-d_j)^2\sum_{\ell\ne j,k}\Bigl(d_\ell-d_j+\frac{d_k-d_j}{2}\Bigr)J_{jk\ell}(\xi,d).
\end{aligned}
\label{eq:E.18}
\end{align}

By \eqref{eq:E.9} and part (b), \(-\partial_\xi q_K(\xi,d)=\sum_{j<k}(d_k-d_j)^2J_{jk}(\xi,d)>0\), so the left-hand side of part (d) equals \(\xi\{-\partial_\xi^2q_K(\xi,d)\}/\{-\partial_\xi q_K(\xi,d)\}\). For \(j<k\), let \(\theta_{jk}:=(d_k-d_j)^2J_{jk}(\xi,d)/\{-\partial_\xi q_K(\xi,d)\}\). These weights are nonnegative and sum to one over \(j<k\). Multiplying \eqref{eq:E.18} by \(\xi\), dividing by \(-\partial_\xi q_K(\xi,d)\) and factoring out the weights gives
\begin{align}
\frac{\xi\,\partial_\xi^2q_K(\xi,d)}{\partial_\xi q_K(\xi,d)}\le\sum_{j<k}\theta_{jk}\Bigl[-\frac{\xi^2(d_k-d_j)^2}{2}+\frac{\xi}{3}\sum_{\ell\ne j,k}\Bigl(d_\ell-d_j+\frac{d_k-d_j}{2}\Bigr)\frac{J_{jk\ell}(\xi,d)}{J_{jk}(\xi,d)}\Bigr].
\label{eq:E.19}
\end{align}

Because the weights are nonnegative and sum to one, it suffices to show that the expression in square brackets is at most the right-hand side of part (d) for every pair \(j<k\).

Fix \(j<k\) and let \(c_\ell\), \(p_{jk}\) and \(\mathbb E_{jk}\) be as in Section~\ref{sec:appendix-E32}. The completion of the square and the change of variables \(y=x+\xi(d_j+d_k)/2\) that led to \eqref{eq:E.13} transform the three-arm integral of arms \(j\), \(k\) and \(\ell\) into the right-hand side of \eqref{eq:E.13} with the factor \(\Phi(y+c_\ell)\) replaced by \(\varphi(y+c_\ell)\). The common factor \(e^{-\xi^2(d_j-d_k)^2/4}/(2\pi)\) cancels in the ratio, and
\[
\frac{J_{jk\ell}(\xi,d)}{J_{jk}(\xi,d)}=\mathbb E_{jk}\Bigl[\frac{\varphi(y+c_\ell)}{\Phi(y+c_\ell)}\Bigr].
\]

Since \(d_\ell-d_j+(d_k-d_j)/2=\{d_\ell-(d_j+d_k)/2\}+(d_k-d_j)\), we have \(\xi\{d_\ell-d_j+(d_k-d_j)/2\}=(y+c_\ell)-y+\xi(d_k-d_j)\) for every \(y\). The left-hand side does not depend on \(y\), so it can be moved inside the expectation and split into these three parts. Applying the second identity in \eqref{eq:E.14} to the part with \(-y\) and the first identity to the part with \(\xi(d_k-d_j)\) gives
\[
\begin{aligned}
\frac{\xi}{3}\sum_{\ell\ne j,k}\Bigl(d_\ell-d_j+\frac{d_k-d_j}{2}\Bigr)\mathbb E_{jk}\Bigl[\frac{\varphi(y+c_\ell)}{\Phi(y+c_\ell)}\Bigr]
&=\frac13\mathbb E_{jk}\Bigl[\sum_{\ell\ne j,k}(y+c_\ell)\frac{\varphi(y+c_\ell)}{\Phi(y+c_\ell)}\Bigr]\\
&\quad-\frac13\bigl(2\mathbb E_{jk}[y^2]-1\bigr)+\frac23\xi(d_k-d_j)\mathbb E_{jk}[y].
\end{aligned}
\]

The expression in square brackets in \eqref{eq:E.19} therefore equals
\begin{align}
\begin{aligned}
&-\frac{\xi^2(d_k-d_j)^2}{2}+\frac13+\frac23\xi(d_k-d_j)\mathbb E_{jk}[y]-\frac23\mathbb E_{jk}[y^2]\\
&\quad+\frac13\mathbb E_{jk}\Bigl[\sum_{\ell\ne j,k}(y+c_\ell)\frac{\varphi(y+c_\ell)}{\Phi(y+c_\ell)}\Bigr].
\end{aligned}
\label{eq:E.20}
\end{align}

We bound \eqref{eq:E.20} in three places. First, summing \eqref{eq:E.15} with \(t=y+c_\ell\) over the \(K-2\) arms \(\ell\ne j,k\), taking expectations and using the first identity in \eqref{eq:E.14} gives
\[
\mathbb E_{jk}\Bigl[\sum_{\ell\ne j,k}(y+c_\ell)\frac{\varphi(y+c_\ell)}{\Phi(y+c_\ell)}\Bigr]\le2\gamma\,\mathbb E_{jk}[y]+(K-2)M_\gamma.
\]

Second, \(\mathbb E_{jk}[y^2]\ge(\mathbb E_{jk}[y])^2\). With these two bounds, \eqref{eq:E.20} is at most
\[
-\frac{\xi^2(d_k-d_j)^2}{2}-\frac23(\mathbb E_{jk}[y])^2+\frac23\{\xi(d_k-d_j)+\gamma\}\mathbb E_{jk}[y]+\frac13+\frac{K-2}{3}M_\gamma.
\]

Third, completing the square in the first three terms gives
\[
\begin{aligned}
&-\frac{\xi^2(d_k-d_j)^2}{2}-\frac23(\mathbb E_{jk}[y])^2+\frac23\{\xi(d_k-d_j)+\gamma\}\mathbb E_{jk}[y]\\
&\qquad=-\frac12\Bigl\{\xi(d_k-d_j)-\frac23\mathbb E_{jk}[y]\Bigr\}^2-\frac49\Bigl(\mathbb E_{jk}[y]-\frac{3\gamma}{4}\Bigr)^2+\frac{\gamma^2}{4}.
\end{aligned}
\]

Indeed, after the right-hand side is expanded, the coefficient of \((\mathbb E_{jk}[y])^2\) is \(-2/9-4/9=-2/3\) and the constant term is \(-(4/9)(9\gamma^2/16)+\gamma^2/4=0\). The first two terms on the right-hand side are nonpositive, so the first three terms of the preceding bound are at most \(\gamma^2/4\). Hence \eqref{eq:E.20} is at most \(1/3+\gamma^2/4+(K-2)M_\gamma/3\), which is the right-hand side of part (d), and \eqref{eq:E.19} proves part (d). For \(K=2\), there are no three-arm integrals, \eqref{eq:E.17} reduces to \(-\partial_\xi^2q_K(\xi,d)=-\xi(d_2-d_1)^4J_{12}(\xi,d)/2\), and the left-hand side of part (d) equals \(-\xi^2(d_2-d_1)^2/2\le0<1/3+\gamma^2/4\). The interchanges of differentiation and integration in this section are justified by the dominating function of Section~\ref{sec:appendix-E31}, and the expectations are finite by Section~\ref{sec:appendix-E32}.

\subsubsection{Two properties of \(\varphi/\Phi\)}\label{sec:appendix-E34}

Sections~\ref{sec:appendix-E36} and~\ref{sec:appendix-E37} use the facts that \((\varphi/\Phi)'\) takes values in \((-1,0)\) and that \(\varphi/\Phi\) is strictly convex on \(\mathbb R\). The substitution \(t=z-u\) in \(\Phi(z)=\int_{-\infty}^z\varphi(t)\,dt\) gives \(\Phi(z)=\int_0^\infty\varphi(z-u)\,du\), and completing the square gives \(\varphi(z-u)=\varphi(z)e^{zu-u^2/2}\). For \(r=0,1,2\), let
\[
I_r(z):=\int_0^\infty u^re^{zu-u^2/2}\,du,
\]
so that \(\Phi(z)/\varphi(z)=I_0(z)\) and \(\varphi(z)/\Phi(z)=1/I_0(z)\). Each \(I_r\) is differentiable with \(I_r'=I_{r+1}\) for \(r=0,1\). Indeed, \(\partial_z\{u^re^{zu-u^2/2}\}=u^{r+1}e^{zu-u^2/2}\), and if \(|z|\le b\), then \(bu-u^2/4\le b^2\) for \(u\ge0\) implies \(u^{r+1}e^{zu-u^2/2}\le e^{b^2}u^{r+1}e^{-u^2/4}\), which is integrable and does not depend on \(z\), so differentiation under the integral sign is permitted. The quotient rule then gives
\[
\begin{aligned}
\frac{d}{dz}\log I_0(z)&=\frac{I_1(z)}{I_0(z)},\\
\frac{d^2}{dz^2}\log I_0(z)&=\frac{I_2(z)I_0(z)-I_1(z)^2}{I_0(z)^2}=\frac{I_2(z)}{I_0(z)}-\Bigl(\frac{I_1(z)}{I_0(z)}\Bigr)^2.
\end{aligned}
\]

These two quantities are the mean and the variance of a probability distribution. For fixed \(z\), let \(g_z(u):=e^{zu-u^2/2}/I_0(z)=\varphi(z-u)/\Phi(z)\) for \(u\ge0\). Since \(\varphi(z-u)\) is the density of \(\mathcal N(z,1)\) at \(u\) and \(\Phi(z)\) is the probability that \(\mathcal N(z,1)\) falls in \([0,\infty)\), \(g_z\) is the density of \(\mathcal N(z,1)\) truncated to \([0,\infty)\). The ratio \(I_1/I_0=\int_0^\infty ug_z(u)\,du\) is the mean of this distribution, and \(I_2/I_0-(I_1/I_0)^2\) is its variance. The distribution has a positive density on \((0,\infty)\), so its mean is positive, and since it is not a point mass, its variance is also positive. On the other hand, \(\log I_0=\log\Phi-\log\varphi\) and \((\log\varphi)'(z)=-z\) give
\[
\frac{d}{dz}\log I_0(z)=z+\frac{\varphi(z)}{\Phi(z)},\qquad\frac{d^2}{dz^2}\log I_0(z)=1+\Bigl(\frac{\varphi}{\Phi}\Bigr)'(z).
\]

The positivity of the mean and of the variance therefore gives \(z+\varphi(z)/\Phi(z)>0\) and \(1+(\varphi/\Phi)'(z)>0\). Moreover, \(\varphi'(z)=-z\varphi(z)\) gives
\[
\Bigl(\frac{\varphi}{\Phi}\Bigr)'(z)=\frac{\varphi'(z)\Phi(z)-\varphi(z)^2}{\Phi(z)^2}=-\frac{\varphi(z)}{\Phi(z)}\Bigl\{z+\frac{\varphi(z)}{\Phi(z)}\Bigr\}<0.
\]

Hence \((\varphi/\Phi)'(z)\in(-1,0)\) for every real \(z\). In addition, \((\log\Phi)''=(\varphi/\Phi)'<0\), so \(\log\Phi\) is strictly concave.

For strict convexity, \(\varphi/\Phi=1/I_0\), \(I_0'=I_1\) and \(I_1'=I_2\) give
\[
\begin{aligned}
\Bigl(\frac{\varphi}{\Phi}\Bigr)'(z)&=-\frac{I_1(z)}{I_0(z)^2},\\
\Bigl(\frac{\varphi}{\Phi}\Bigr)''(z)&=-\frac{I_2(z)}{I_0(z)^2}+\frac{2I_1(z)^2}{I_0(z)^3}=\frac{2I_1(z)^2-I_0(z)I_2(z)}{I_0(z)^3}.
\end{aligned}
\]

Since \(I_0>0\), it suffices to show that \(I_0(z)I_2(z)<2I_1(z)^2\). Fix \(s,t>0\). Then
\[
z-s=\frac{t}{s+t}z+\frac{s}{s+t}(z-s-t),\qquad z-t=\frac{s}{s+t}z+\frac{t}{s+t}(z-s-t),
\]
and \(z-s-t<z\). The strict concavity of \(\log\Phi\) gives
\[
\begin{aligned}
\log\Phi(z-s)&>\frac{t}{s+t}\log\Phi(z)+\frac{s}{s+t}\log\Phi(z-s-t),\\
\log\Phi(z-t)&>\frac{s}{s+t}\log\Phi(z)+\frac{t}{s+t}\log\Phi(z-s-t),
\end{aligned}
\]
and adding the two inequalities and exponentiating gives \(\Phi(z)\Phi(z-s-t)<\Phi(z-s)\Phi(z-t)\). We integrate this inequality over \(s,t\in(0,\infty)\). For \(w\ge0\), \(\Phi(z-w)=\int_w^\infty\varphi(z-u)\,du\). The integrands are nonnegative, so Tonelli's theorem permits interchanging the order of integration, and for fixed \(u\ge0\) the set \(\{(s,t):s,t\ge0,\ s+t\le u\}\) has area \(u^2/2\). Hence
\[
\begin{aligned}
\int_0^\infty\int_0^\infty\Phi(z-s-t)\,ds\,dt&=\int_0^\infty\frac{u^2}{2}\varphi(z-u)\,du=\frac{\varphi(z)}{2}I_2(z),\\
\int_0^\infty\Phi(z-s)\,ds&=\int_0^\infty u\varphi(z-u)\,du=\varphi(z)I_1(z),
\end{aligned}
\]
and both integrals are finite. The difference \(\Phi(z-s)\Phi(z-t)-\Phi(z)\Phi(z-s-t)\) is continuous and positive on \((0,\infty)^2\), and the integrals of both terms are finite, so its integral is positive and
\[
\Phi(z)\cdot\frac{\varphi(z)}{2}I_2(z)<\varphi(z)^2I_1(z)^2.
\]

Substituting \(\Phi(z)=\varphi(z)I_0(z)\) and dividing both sides by \(\varphi(z)^2/2\) gives \(I_0(z)I_2(z)<2I_1(z)^2\). Therefore \((\varphi/\Phi)''(z)>0\) for every real \(z\), and \(\varphi/\Phi\) is strictly convex.

\subsubsection{A normalized gap vector with a single best arm}\label{sec:appendix-E35}

Sections~\ref{sec:appendix-E36}--\ref{sec:appendix-E38} use the normalized gap vector \(d^\circ:=(0,K/(K-1),\ldots,K/(K-1))\), which belongs to \(\mathcal D_K\) because its components other than the first sum to \((K-1)\cdot K/(K-1)=K\). We do not assume that \(d^\circ\) maximizes the normalized marginal ratio. As a local scale we use \(\xi^\circ:=K\xi/(K-1)\), which equals \(\xi d^\circ_k\) for \(k\ge2\) and is in one-to-one correspondence with \(\xi>0\), and \(\partial_\xi=\{K/(K-1)\}\partial_{\xi^\circ}\) gives \(\xi\partial_\xi=\xi^\circ\partial_{\xi^\circ}\).

Arm \(1\) is selected if and only if \(Z_k-\xi d^\circ_k<Z_1-\xi d^\circ_1\) for every \(k\ge2\), up to an event of probability zero. Since \(\xi(d^\circ_k-d^\circ_1)=\xi^\circ\), this condition is equivalent to \(Z_k<Z_1+\xi^\circ\) for every \(k\ge2\). Given \(Z_1=z\), the independence of \(Z_2,\ldots,Z_K\) gives the conditional probability \(\Phi(z+\xi^\circ)^{K-1}\), and integrating against the density \(\varphi(z)\) and substituting \(u=z+\xi^\circ\) shows that arm \(1\) is selected with probability \(\int_{\mathbb R}\varphi(z)\Phi(z+\xi^\circ)^{K-1}\,dz=\int_{\mathbb R}\varphi(u-\xi^\circ)\Phi(u)^{K-1}\,du\). Every other arm has normalized mean difference \(K/(K-1)\), so
\[
q_K(\xi,d^\circ)=\frac{K}{K-1}\Bigl\{1-\int_{\mathbb R}\varphi(u-\xi^\circ)\Phi(u)^{K-1}\,du\Bigr\},
\]
and hence
\begin{align}
\begin{aligned}
1-q_K(\xi,d^\circ)
&=1-\frac{K}{K-1}+\frac{K}{K-1}\int_{\mathbb R}\varphi(u-\xi^\circ)\Phi(u)^{K-1}\,du\\
&=\frac{K}{K-1}\Bigl\{\int_{\mathbb R}\varphi(u-\xi^\circ)\Phi(u)^{K-1}\,du-\frac1K\Bigr\}.
\end{aligned}
\label{eq:E.21}
\end{align}

At \(\xi^\circ=0\), \(\Phi'=\varphi\) gives \(\int_{\mathbb R}\varphi(u)\Phi(u)^{K-1}\,du=[\Phi(u)^K/K]_{-\infty}^{\infty}=1/K\), so arm \(1\) is then selected with probability \(1/K\).

\subsubsection{\texorpdfstring{A scalar comparison}{A scalar comparison}}\label{sec:appendix-E36}

We show that
\begin{align}
\inf_{\gamma\ge0}\Bigl\{1+\frac{\gamma^2}{2}+(K-2)M_\gamma\Bigr\}<1+4\sup_{\xi>0}\frac{-(\xi/2)\partial_\xi q_K(\xi,d^\circ)}{1-q_K(\xi,d^\circ)}.
\label{eq:E.22}
\end{align}

The left-hand side does not depend on any normalized gap vector, and \(d^\circ\) enters only through \(q_K(\xi,d^\circ)\) on the right-hand side. By part (b), \(1-q_K(\xi,d^\circ)>0\) for \(\xi>0\), and \(-\partial_\xi q_K=\partial_\xi(1-q_K)\) and \(\xi\partial_\xi=\xi^\circ\partial_{\xi^\circ}\) give
\[
\frac{-(\xi/2)\partial_\xi q_K(\xi,d^\circ)}{1-q_K(\xi,d^\circ)}=\frac{\xi}{2}\partial_\xi\log\{1-q_K(\xi,d^\circ)\}=\frac{\xi^\circ}{2}\partial_{\xi^\circ}\log\{1-q_K(\xi,d^\circ)\}.
\]

By \eqref{eq:E.11}, the supremum on the right-hand side of \eqref{eq:E.22} is at least \(1/2\), so the right-hand side of \eqref{eq:E.22} is at least \(1+4\cdot(1/2)=3\).

Suppose first that \(K\le7\). We bound the infimum from above by the value at \(\gamma=1\). For \(z\le1\), \((z-1)\varphi(z)/\Phi(z)\le0\). For \(z>1\), \(\Phi(z)\ge1/2\), and the identity \(z^2/2-1/2=(z-1)(z+1)/2=(z-1)+(z-1)^2/2\) gives \(\varphi(z)=\varphi(1)e^{-(z-1)-(z-1)^2/2}\le\varphi(1)e^{-(z-1)}\). Together with \(\sup_{t\ge0}te^{-t}=1/e\), this gives
\[
(z-1)\frac{\varphi(z)}{\Phi(z)}\le2\varphi(1)(z-1)e^{-(z-1)}\le\frac{2\varphi(1)}{e}=\sqrt{\frac2\pi}\,e^{-3/2}<\frac14,
\]
where the last inequality follows from \(\sqrt{2/\pi}<1\) and \(e^{3/2}>4\), which in turn follows from \(e^3>2.7^3>16\). Hence \(M_1<1/4\). Since \(0\le K-2\le5\), the expression in braces at \(\gamma=1\) satisfies
\[
1+\frac12+(K-2)M_1<1+\frac12+\frac54=\frac{11}{4},
\]
and for \(K=2\) it equals \(3/2\). In either case the left-hand side of \eqref{eq:E.22} is smaller than \(3\), while the right-hand side is at least \(3\), so \eqref{eq:E.22} holds.

Suppose next that \(K\ge8\). We construct a positive number \(z_0\)
such that the left-hand side of \eqref{eq:E.22} is smaller than \(1+z_0^2/2\) and the right-hand side is larger than \(1+z_0^2/2\). We first locate, for fixed \(\gamma\ge0\), the maximizer of \(t\mapsto(t-\gamma)\varphi(t)/\Phi(t)\). This function is nonpositive for \(t\le\gamma\) and positive for \(t>\gamma\). By the formula for \((\varphi/\Phi)'\) in Section~\ref{sec:appendix-E34},
\[
\begin{aligned}
\frac{d}{dt}\Bigl\{(t-\gamma)\frac{\varphi(t)}{\Phi(t)}\Bigr\}
&=\frac{\varphi(t)}{\Phi(t)}+(t-\gamma)\Bigl(\frac{\varphi}{\Phi}\Bigr)'(t)\\
&=\frac{\varphi(t)}{\Phi(t)}\Bigl[1-(t-\gamma)\Bigl\{t+\frac{\varphi(t)}{\Phi(t)}\Bigr\}\Bigr].
\end{aligned}
\]

For \(t>\gamma\), the factor \(t-\gamma\) is positive and strictly increasing, and by Section~\ref{sec:appendix-E34} the factor \(t+\varphi(t)/\Phi(t)\) is positive with positive derivative \(1+(\varphi/\Phi)'(t)\). Their product is therefore strictly increasing on \((\gamma,\infty)\). It tends to zero as \(t\downarrow\gamma\), and since \(t+\varphi(t)/\Phi(t)>t\), it diverges as \(t\to\infty\). Hence there is exactly one \(t>\gamma\) at which the product equals one, the derivative above is positive to its left and negative to its right, and, since the function is nonpositive for \(t\le\gamma\), \((t-\gamma)\varphi(t)/\Phi(t)\) attains its maximum over \(\mathbb R\) at this point.

We now construct \(z_0\). The derivative of \(z-(K-2)\varphi(z)/\Phi(z)\) is \(1-(K-2)(\varphi/\Phi)'(z)\ge1\) because \((\varphi/\Phi)'<0\), so this function is strictly increasing. It is at most \(z\) because \(\varphi/\Phi>0\), so it diverges to \(-\infty\) as \(z\to-\infty\), and since its derivative is at least one, it diverges to \(\infty\) as \(z\to\infty\). It therefore has exactly one zero, and to the right of the zero it is positive and strictly increasing. On this range \(z+\varphi(z)/\Phi(z)\) is also positive and strictly increasing, so the product of the two functions increases strictly from zero at the zero to infinity. Hence there is exactly one \(z_0\) satisfying
\[
\Bigl\{z_0-(K-2)\frac{\varphi(z_0)}{\Phi(z_0)}\Bigr\}\Bigl\{z_0+\frac{\varphi(z_0)}{\Phi(z_0)}\Bigr\}=1,\qquad z_0>(K-2)\frac{\varphi(z_0)}{\Phi(z_0)}.
\]

Let \(\gamma_0:=(K-2)\varphi(z_0)/\Phi(z_0)\). Then \(\gamma_0>0\) because \(K\ge8\), and \(z_0>\gamma_0>0\). The equality above reads \((z_0-\gamma_0)\{z_0+\varphi(z_0)/\Phi(z_0)\}=1\), so by the preceding paragraph
\[
M_{\gamma_0}=\sup_{t\in\mathbb R}(t-\gamma_0)\frac{\varphi(t)}{\Phi(t)}=(z_0-\gamma_0)\frac{\varphi(z_0)}{\Phi(z_0)}.
\]

We claim that \(\gamma_0\ge z_0/2\). Suppose instead that \(\gamma_0<z_0/2\). Then \(z_0-\gamma_0>z_0/2>0\) and \(z_0+\varphi(z_0)/\Phi(z_0)>z_0>0\), so
\[
1=(z_0-\gamma_0)\Bigl\{z_0+\frac{\varphi(z_0)}{\Phi(z_0)}\Bigr\}>\frac{z_0}{2}\cdot z_0=\frac{z_0^2}{2},
\]
and \(z_0<\sqrt2\). The function \(z/\varphi(z)=\sqrt{2\pi}\,ze^{z^2/2}\) is strictly increasing on \(z>0\), and \(\Phi(z_0)<1\), so
\[
K-2=\frac{\gamma_0\Phi(z_0)}{\varphi(z_0)}<\frac{z_0\Phi(z_0)}{2\varphi(z_0)}<\frac{z_0}{2\varphi(z_0)}<\frac{\sqrt2}{2\varphi(\sqrt2)}=\frac{\sqrt2\cdot\sqrt{2\pi}\,e}{2}=\sqrt\pi\,e<5,
\]
where the last inequality follows from \(\sqrt\pi\,e<1.78\times2.72<5\). This contradicts \(K\ge8\), and hence \(\gamma_0\ge z_0/2\).

For the left-hand side of \eqref{eq:E.22}, \((K-2)\varphi(z_0)/\Phi(z_0)=\gamma_0\) shows that the expression in braces at \(\gamma=\gamma_0\) equals
\[
\begin{aligned}
1+\frac{\gamma_0^2}{2}+(K-2)(z_0-\gamma_0)\frac{\varphi(z_0)}{\Phi(z_0)}
&=1+\frac{\gamma_0^2}{2}+\gamma_0(z_0-\gamma_0)=1+\gamma_0z_0-\frac{\gamma_0^2}{2}\\
&=1+\frac{z_0^2}{2}-\frac{(z_0-\gamma_0)^2}{2}<1+\frac{z_0^2}{2},
\end{aligned}
\]
where the last inequality uses \(z_0\ne\gamma_0\). Hence the left-hand side of \eqref{eq:E.22} is smaller than \(1+z_0^2/2\).

For the right-hand side, it suffices to show that the supremum in \eqref{eq:E.22} exceeds \(z_0^2/8\). We evaluate the normalized marginal ratio at \(\xi^\circ=z_0/2\), which corresponds to \(\xi=(K-1)z_0/(2K)>0\). Let \(f\) be the probability density on \(\mathbb R\) proportional to \(\varphi(u-z_0/2)\Phi(u)^{K-1}\), that is,
\[
f(u):=\frac{\varphi(u-z_0/2)\Phi(u)^{K-1}}{\int_{\mathbb R}\varphi(v-z_0/2)\Phi(v)^{K-1}\,dv}.
\]

Since \((\log\varphi)'(t)=-t\) and \((\log\Phi)'=\varphi/\Phi\),
\[
(\log f)'(u)=-\Bigl(u-\frac{z_0}{2}\Bigr)+(K-1)\frac{\varphi(u)}{\Phi(u)}.
\]

By Section~\ref{sec:appendix-E34}, \((\log f)''(u)=-1+(K-1)(\varphi/\Phi)'(u)<0\), so \((\log f)'\) is strictly decreasing, and \((\log f)'''=(K-1)(\varphi/\Phi)''>0\), so \((\log f)'\) is strictly convex. As \(u\to\infty\), \(\varphi(u)/\Phi(u)\to0\) and \((\log f)'(u)\to-\infty\). As \(u\to-\infty\), the second term is positive and \((\log f)'(u)\to\infty\). At \(u=z_0\), the relations \((K-2)\varphi(z_0)/\Phi(z_0)=\gamma_0\) and \(\gamma_0\ge z_0/2\) give
\[
(\log f)'(z_0)=-\frac{z_0}{2}+(K-1)\frac{\varphi(z_0)}{\Phi(z_0)}=-\frac{z_0}{2}+\gamma_0+\frac{\varphi(z_0)}{\Phi(z_0)}>0.
\]

Hence \((\log f)'\) has exactly one zero \(u_0\), which is the unique mode of \(f\), and \(u_0>z_0\) because \((\log f)'\) is decreasing.

We show that the mean of \(f\) exceeds \(u_0\). For \(s>0\), the chain rule gives \((d/ds)\log f(u_0+s)=(\log f)'(u_0+s)\) and \((d/ds)\log f(u_0-s)=-(\log f)'(u_0-s)\), so
\[
\frac{d}{ds}\log\frac{f(u_0+s)}{f(u_0-s)}=\frac{d}{ds}\bigl\{\log f(u_0+s)-\log f(u_0-s)\bigr\}=(\log f)'(u_0+s)+(\log f)'(u_0-s).
\]

Since \(u_0\) is the midpoint of \(u_0+s\) and \(u_0-s\) and \((\log f)'\) is strictly convex, \((\log f)'(u_0)<\{(\log f)'(u_0+s)+(\log f)'(u_0-s)\}/2\). Together with \((\log f)'(u_0)=0\), this shows that the derivative above is positive for every \(s>0\). The logarithmic ratio vanishes at \(s=0\) and is strictly increasing on \(s\ge0\), so it is positive for every \(s>0\), and \(f(u_0+s)>f(u_0-s)\). Since \(\Phi\le1\), \(f(u)\) is bounded by a constant multiple of \(\varphi(u-z_0/2)\), and \(\int_{\mathbb R}|u|f(u)\,du\) is finite. Using \(\int_{\mathbb R}f(u)\,du=1\), splitting the integral at \(u_0\) and substituting \(u=u_0+s\) on \((u_0,\infty)\) and \(u=u_0-s\) on \((-\infty,u_0)\), we obtain
\[
\begin{aligned}
\int_{\mathbb R}uf(u)\,du-u_0&=\int_{\mathbb R}(u-u_0)f(u)\,du=\int_0^\infty sf(u_0+s)\,ds-\int_0^\infty sf(u_0-s)\,ds\\
&=\int_0^\infty s\{f(u_0+s)-f(u_0-s)\}\,ds.
\end{aligned}
\]

The integrand on the right is continuous and positive for \(s>0\), so the integral is positive. Hence the mean of \(f\) exceeds \(u_0\), and since \(u_0>z_0\), it exceeds \(z_0\).

We connect this bound with the logarithmic derivative of the selection probability of arm \(1\). Differentiating under the integral sign as in Section~\ref{sec:appendix-E31} and using \(\partial_{\xi^\circ}\varphi(u-\xi^\circ)=(u-\xi^\circ)\varphi(u-\xi^\circ)\) gives
\[
\begin{aligned}
\frac{d}{d\xi^\circ}\int_{\mathbb R}\varphi(u-\xi^\circ)\Phi(u)^{K-1}\,du
&=\int_{\mathbb R}(u-\xi^\circ)\varphi(u-\xi^\circ)\Phi(u)^{K-1}\,du\\
&=\int_{\mathbb R}u\,\varphi(u-\xi^\circ)\Phi(u)^{K-1}\,du-\xi^\circ\int_{\mathbb R}\varphi(u-\xi^\circ)\Phi(u)^{K-1}\,du,
\end{aligned}
\]
and dividing by \(\int_{\mathbb R}\varphi(u-\xi^\circ)\Phi(u)^{K-1}\,du\) gives
\[
\partial_{\xi^\circ}\log\int_{\mathbb R}\varphi(u-\xi^\circ)\Phi(u)^{K-1}\,du=\int_{\mathbb R}u\,\frac{\varphi(u-\xi^\circ)\Phi(u)^{K-1}}{\int_{\mathbb R}\varphi(v-\xi^\circ)\Phi(v)^{K-1}\,dv}\,du-\xi^\circ.
\]

At \(\xi^\circ=z_0/2\), the fraction inside the integral equals \(f(u)\), so the first term on the right is the mean of \(f\), and
\[
\partial_{\xi^\circ}\log\int_{\mathbb R}\varphi(u-\xi^\circ)\Phi(u)^{K-1}\,du\Big|_{\xi^\circ=z_0/2}=\int_{\mathbb R}uf(u)\,du-\frac{z_0}{2}>z_0-\frac{z_0}{2}=\frac{z_0}{2}.
\]

In particular, the derivative of the selection probability of arm \(1\) is positive at \(\xi^\circ=z_0/2\). By \eqref{eq:E.21},
\[
\log\{1-q_K(\xi,d^\circ)\}=\log\frac{K}{K-1}+\log\Bigl\{\int_{\mathbb R}\varphi(u-\xi^\circ)\Phi(u)^{K-1}\,du-\frac1K\Bigr\}.
\]

For \(\xi^\circ>0\), part (b) and \eqref{eq:E.21} give \(0<\int_{\mathbb R}\varphi(u-\xi^\circ)\Phi(u)^{K-1}\,du-1/K<\int_{\mathbb R}\varphi(u-\xi^\circ)\Phi(u)^{K-1}\,du\). Since the derivative of the selection probability is positive at \(\xi^\circ=z_0/2\), it follows that, at this point,
\[
\partial_{\xi^\circ}\log\{1-q_K(\xi,d^\circ)\}=\frac{\partial_{\xi^\circ}\int_{\mathbb R}\varphi(u-\xi^\circ)\Phi(u)^{K-1}\,du}{\int_{\mathbb R}\varphi(u-\xi^\circ)\Phi(u)^{K-1}\,du-1/K}>\frac{\partial_{\xi^\circ}\int_{\mathbb R}\varphi(u-\xi^\circ)\Phi(u)^{K-1}\,du}{\int_{\mathbb R}\varphi(u-\xi^\circ)\Phi(u)^{K-1}\,du}.
\]

The right-hand side is the logarithmic derivative of the selection probability of arm \(1\). Combining this inequality with the representation of the normalized marginal ratio at \(d^\circ\) obtained at the beginning of this section, at \(\xi^\circ=z_0/2\), that is, at \(\xi=(K-1)z_0/(2K)\), we obtain
\[
\frac{-(\xi/2)\partial_\xi q_K(\xi,d^\circ)}{1-q_K(\xi,d^\circ)}=\frac{\xi^\circ}{2}\partial_{\xi^\circ}\log\{1-q_K(\xi,d^\circ)\}>\frac12\cdot\frac{z_0}{2}\cdot\frac{z_0}{2}=\frac{z_0^2}{8}.
\]

Hence the supremum in \eqref{eq:E.22} exceeds \(z_0^2/8\), and the right-hand side of \eqref{eq:E.22} exceeds \(1+4\cdot z_0^2/8=1+z_0^2/2\). Since the left-hand side is smaller than \(1+z_0^2/2\), \eqref{eq:E.22} holds for \(K\ge8\). Together with the case \(K\le7\), this proves \eqref{eq:E.22} for every \(K\ge2\).

\subsubsection{Proof of part (e)}\label{sec:appendix-E37}

Suppose that \(0<\beta_K<\infty\), and fix \(a\in(0,1)\). We first show that every state satisfies
\begin{align}
\frac{\xi\,\partial_\xi^2q_K(\xi,d)}{\partial_\xi q_K(\xi,d)}<1+2\beta_K,
\label{eq:E.23}
\end{align}
then convert \eqref{eq:E.23} into a statewise finite-difference bound, and finally evaluate this bound along states whose normalized marginal ratios approach \(\beta_K\).

The normalized marginal ratio at every state \((\xi,d^\circ)\) is at most \(\beta_K\), so the right-hand side of \eqref{eq:E.22} is at most \(1+4\beta_K\). Since the infimum on the left-hand side of \eqref{eq:E.22} is smaller than \(1+4\beta_K\), there exists \(\gamma\ge0\) such that \(1+\gamma^2/2+(K-2)M_\gamma<1+4\beta_K\), and we fix this \(\gamma\). The right-hand side of part (d) at this \(\gamma\) can be written as
\[
\frac13+\frac{\gamma^2}{4}+\frac{K-2}{3}M_\gamma=\frac12\Bigl\{1+\frac{\gamma^2}{2}+(K-2)M_\gamma\Bigr\}-\frac16-\frac{K-2}{6}M_\gamma.
\]

Since \(M_\gamma\ge0\) by Section~\ref{sec:appendix-E32}, the last term is nonpositive, and the right-hand side is smaller than \((1+4\beta_K)/2-1/6=1/3+2\beta_K<1+2\beta_K\). Part (d) holds at every state, so \eqref{eq:E.23} holds at every state \((\xi,d)\in(0,\infty)\times\mathcal D_K\).

Fix \(d\in\mathcal D_K\). By parts (a) and (b), \(t\mapsto\log\{-\partial_tq_K(t,d)\}\) is continuously differentiable on \((0,\infty)\), and since \(\partial_tq_K<0\), \eqref{eq:E.23} gives
\[
\frac{d}{dt}\log\{-\partial_tq_K(t,d)\}=\frac{\partial_t^2q_K(t,d)}{\partial_tq_K(t,d)}<\frac{1+2\beta_K}{t}.
\]

For \(0<s\le\xi\), integrating over \(t\in[s,\xi]\) gives \(\log\{-\partial_\xi q_K(\xi,d)\}-\log\{-\partial_sq_K(s,d)\}\le(1+2\beta_K)\log(\xi/s)\), and exponentiating gives
\begin{align}
-\partial_sq_K(s,d)\ge\Bigl(\frac{s}{\xi}\Bigr)^{1+2\beta_K}\{-\partial_\xi q_K(\xi,d)\}\qquad(0<s\le\xi).
\label{eq:E.24}
\end{align}

Fix a state \((\xi,d)\) and integrate \eqref{eq:E.24} over \(s\in[\sqrt a\,\xi,\xi]\). The substitution \(t=s/\xi\) gives
\[
\int_{\sqrt a\,\xi}^{\xi}\Bigl(\frac{s}{\xi}\Bigr)^{1+2\beta_K}ds=\xi\int_{\sqrt a}^1t^{1+2\beta_K}\,dt=\frac{\xi}{2(1+\beta_K)}\bigl\{1-(\sqrt a)^{2+2\beta_K}\bigr\}=\frac{\xi}{2(1+\beta_K)}\bigl(1-a^{1+\beta_K}\bigr),
\]
and by the fundamental theorem of calculus,
\[
q_K(\sqrt a\,\xi,d)-q_K(\xi,d)=\int_{\sqrt a\,\xi}^{\xi}\{-\partial_sq_K(s,d)\}\,ds\ge\frac{\xi\{-\partial_\xi q_K(\xi,d)\}}{2(1+\beta_K)}\bigl(1-a^{1+\beta_K}\bigr).
\]

Writing \(1-q_K(\sqrt a\,\xi,d)=1-q_K(\xi,d)-\{q_K(\sqrt a\,\xi,d)-q_K(\xi,d)\}\) and dividing by \(1-q_K(\xi,d)>0\), we obtain, for every state \((\xi,d)\),
\begin{align}
\frac{1-q_K(\sqrt a\,\xi,d)}{1-q_K(\xi,d)}\le1-\frac{-(\xi/2)\partial_\xi q_K(\xi,d)}{1-q_K(\xi,d)}\cdot\frac{1-a^{1+\beta_K}}{1+\beta_K}.
\label{eq:E.25}
\end{align}

The right-hand side of \eqref{eq:E.25} decreases as the normalized marginal ratio at the state increases. Since the normalized marginal ratio is at most \(\beta_K\) and \(1-a^{1+\beta_K}>0\), the right-hand side is at least
\[
1-\beta_K\frac{1-a^{1+\beta_K}}{1+\beta_K}=\frac{1+\beta_Ka^{1+\beta_K}}{1+\beta_K}>0.
\]

The left-hand side of \eqref{eq:E.25} is positive by part (b), so taking reciprocals reverses the inequality and gives
\begin{align}
\frac{1-q_K(\xi,d)}{1-q_K(\sqrt a\,\xi,d)}\ge\Bigl[1-\frac{-(\xi/2)\partial_\xi q_K(\xi,d)}{1-q_K(\xi,d)}\cdot\frac{1-a^{1+\beta_K}}{1+\beta_K}\Bigr]^{-1}.
\label{eq:E.26}
\end{align}

By the definition of \(\beta_K\) as a supremum, there is a sequence of states \((\xi_n,d_n)\) whose normalized marginal ratios converge to \(\beta_K\), although the supremum need not be attained at any state. Since the expression in square brackets in \eqref{eq:E.26} is bounded below by the positive number displayed above, the right-hand side of \eqref{eq:E.26} at \((\xi_n,d_n)\) converges to
\[
\Bigl[1-\beta_K\frac{1-a^{1+\beta_K}}{1+\beta_K}\Bigr]^{-1}=\frac{1+\beta_K}{1+\beta_Ka^{1+\beta_K}}.
\]

This limit is strictly larger than \(1+\beta_K(1-a)\). Indeed, \(t\mapsto t^{-\beta_K}\) is strictly convex on \((0,\infty)\) because \(\beta_K>0\), and comparison with its tangent line at \(t=1\) gives \(a^{-\beta_K}>1+\beta_K(1-a)\) for \(0<a<1\). Multiplying by \(a^{\beta_K}>0\) gives \(a^{\beta_K}\{1+\beta_K(1-a)\}<1\), and therefore
\begin{align}
\frac{1+\beta_Ka^{1+\beta_K}}{1+\beta_K}\{1+\beta_K(1-a)\}-1=\frac{\beta_Ka}{1+\beta_K}\bigl[a^{\beta_K}\{1+\beta_K(1-a)\}-1\bigr]<0.
\label{eq:E.27}
\end{align}

Since \(1+\beta_Ka^{1+\beta_K}\), \(1+\beta_K\) and \(1+\beta_K(1-a)\) are positive, \eqref{eq:E.27} is equivalent to
\[
\frac{1+\beta_K}{1+\beta_Ka^{1+\beta_K}}>1+\beta_K(1-a).
\]

Hence, for all sufficiently large \(n\), the left-hand side of \eqref{eq:E.26} at \((\xi_n,d_n)\) exceeds \(1+\beta_K(1-a)\), and any such state satisfies the inequality in part (e). The strict inequality comes from the strict convexity of \(t\mapsto t^{-\beta_K}\), and the argument does not require \(\beta_K\) to be attained.

\subsubsection{Proof of part (f)}\label{sec:appendix-E38}

Suppose that \(K\ge3\) and \(0<\beta_K<\infty\). We show that the normalized marginal ratio attains its supremum \(\beta_K\) at some state and that this single state satisfies the inequality in part (e) for every \(a\in(0,1)\). The set \(\mathcal D_K\) is closed and bounded and hence compact. By Section~\ref{sec:appendix-E31} and part (b), the normalized marginal ratio is continuous on \((0,\infty)\times\mathcal D_K\), and \(-\partial_\xi q_K\) is continuous on \([0,\infty)\times\mathcal D_K\). By \eqref{eq:E.10}, the normalized marginal ratio at a state \((\xi,d)\) equals
\begin{align}
\frac{-(\xi/2)\partial_\xi q_K(\xi,d)}{1-q_K(\xi,d)}=\frac{(\xi/2)\{-\partial_\xi q_K(\xi,d)\}}{\int_0^\xi\{-\partial_sq_K(s,d)\}\,ds}.
\label{eq:E.28}
\end{align}

We first show that the normalized marginal ratio is uniformly close to \(1/2\) for small \(\xi\). At \(\xi=0\), every pair integral equals \(\int_{\mathbb R}\varphi(x)^2\Phi(x)^{K-2}\,dx\), so \eqref{eq:E.9} gives \(-\partial_\xi q_K(0,d)=\sum_{j<k}(d_j-d_k)^2\int_{\mathbb R}\varphi(x)^2\Phi(x)^{K-2}\,dx\). Expanding gives \(\sum_{j<k}(d_j-d_k)^2=K\sum_jd_j^2-(\sum_jd_j)^2\). A normalized gap vector \(d\in\mathcal D_K\) has a zero component, and its remaining \(K-1\) components sum to \(K\), so \(\sum_jd_j=K\), and the Cauchy--Schwarz inequality applied to these \(K-1\) components gives \(\sum_jd_j^2\ge K^2/(K-1)\). Therefore \(\sum_{j<k}(d_j-d_k)^2\ge K^3/(K-1)-K^2=K^2/(K-1)\), and for every \(d\in\mathcal D_K\),
\[
-\partial_\xi q_K(0,d)\ge\lambda_K:=\frac{K^2}{K-1}\int_{\mathbb R}\varphi(x)^2\Phi(x)^{K-2}\,dx>0.
\]

The constant \(\lambda_K\) does not depend on \(d\). Fix \(\varepsilon\in(0,1)\). Since \(-\partial_\xi q_K\) is continuous on the compact set \([0,1]\times\mathcal D_K\), it is uniformly continuous there, and there exists \(\delta\in(0,1]\) such that \(|\partial_sq_K(s,d)-\partial_\xi q_K(0,d)|\le\varepsilon\lambda_K\) for every \(s\in[0,\delta]\) and every \(d\in\mathcal D_K\). Since \(\lambda_K\le-\partial_\xi q_K(0,d)\),
\[
(1-\varepsilon)\{-\partial_\xi q_K(0,d)\}\le-\partial_sq_K(s,d)\le(1+\varepsilon)\{-\partial_\xi q_K(0,d)\}\qquad(0\le s\le\delta).
\]

For \(0<\xi\le\delta\), this inequality at \(s=\xi\) bounds the numerator of \eqref{eq:E.28}, and its integral over \(s\in[0,\xi]\) gives \((1-\varepsilon)\xi\{-\partial_\xi q_K(0,d)\}\le\int_0^\xi\{-\partial_sq_K(s,d)\}\,ds\le(1+\varepsilon)\xi\{-\partial_\xi q_K(0,d)\}\) for the denominator. Substituting both bounds into \eqref{eq:E.28} and cancelling \(\xi\) and the positive number \(-\partial_\xi q_K(0,d)\), we obtain, for every \(0<\xi\le\delta\) and every \(d\in\mathcal D_K\),
\begin{align}
\frac{1-\varepsilon}{2(1+\varepsilon)}\le\frac{-(\xi/2)\partial_\xi q_K(\xi,d)}{1-q_K(\xi,d)}\le\frac{1+\varepsilon}{2(1-\varepsilon)}.
\label{eq:E.29}
\end{align}

We next bound the normalized marginal ratio uniformly for large \(\xi\), treating the denominator and the numerator separately. By \eqref{eq:E.12}, \(1-q_K(\xi,d)\ge1/2\) whenever \(\xi\ge2K/\sqrt\pi\). For the numerator, bounding the distribution functions of the other arms by one in \eqref{eq:E.13} and using \(\int_{\mathbb R}e^{-y^2}\,dy=\sqrt\pi\) give \(J_{jk}(\xi,d)\le e^{-\xi^2(d_j-d_k)^2/4}\sqrt\pi/(2\pi)=e^{-\xi^2(d_j-d_k)^2/4}/(2\sqrt\pi)\). Substituting this bound into \eqref{eq:E.9} and multiplying by \(\xi\) gives
\[
0<\xi\{-\partial_\xi q_K(\xi,d)\}\le\frac{1}{2\sqrt\pi\,\xi}\sum_{j<k}\xi^2(d_j-d_k)^2e^{-\xi^2(d_j-d_k)^2/4}.
\]

The derivative of \(te^{-t/4}\) is \(e^{-t/4}(1-t/4)\), so \(\sup_{t\ge0}te^{-t/4}\) is attained at \(t=4\) and equals \(4/e\). Applying this bound with \(t=\xi^2(d_j-d_k)^2\) to each of the \(K(K-1)/2\) pairs gives
\[
\xi\{-\partial_\xi q_K(\xi,d)\}\le\frac{1}{2\sqrt\pi\,\xi}\cdot\frac{K(K-1)}{2}\cdot\frac4e=\frac{K(K-1)}{e\sqrt\pi\,\xi}.
\]

The normalized marginal ratio is at most one half of this bound divided by the lower bound \(1/2\) for the denominator. For \(\xi\ge2K^2/\sqrt\pi\), which implies \(\xi\ge2K/\sqrt\pi\), we therefore obtain, for every \(d\in\mathcal D_K\),
\begin{align}
\frac{-(\xi/2)\partial_\xi q_K(\xi,d)}{1-q_K(\xi,d)}\le\frac{K(K-1)}{e\sqrt\pi\,\xi}\le\frac{K(K-1)}{e\sqrt\pi}\cdot\frac{\sqrt\pi}{2K^2}=\frac{K-1}{2eK}<\frac1{2e}.
\label{eq:E.30}
\end{align}

We now show that \(\beta_K>1/2\) by using the normalized gap vector \(d^\circ\) of Section~\ref{sec:appendix-E35} and computing \(-\partial_\xi^2q_K(0,d^\circ)\). Differentiating \eqref{eq:E.9} gives \(-\partial_\xi^2q_K(\xi,d^\circ)=\sum_{j<k}(d^\circ_j-d^\circ_k)^2\partial_\xi J_{jk}(\xi,d^\circ)\), and by \eqref{eq:E.16}, which holds for every real \(\xi\), the first term of \(\partial_\xi J_{jk}(\xi,d^\circ)\) vanishes at \(\xi=0\) because of its factor \(\xi\). At \(\xi=0\), every three-arm integral equals \(\int_{\mathbb R}\varphi(x)^3\Phi(x)^{K-3}\,dx\). The pairs with \((d^\circ_j-d^\circ_k)^2>0\) are the \(K-1\) pairs consisting of arm \(1\) and an arm \(k\ge2\), and each satisfies \((d^\circ_1-d^\circ_k)^2=\{K/(K-1)\}^2\). Each such pair has \(K-2\) third arms \(\ell\), and each of them satisfies \(d^\circ_\ell-(d^\circ_1+d^\circ_k)/2=K/(K-1)-K/\{2(K-1)\}=K/\{2(K-1)\}\). Therefore
\[
\begin{aligned}
-\partial_\xi^2q_K(0,d^\circ)
&=(K-1)\Bigl(\frac{K}{K-1}\Bigr)^2(K-2)\frac{K}{2(K-1)}\int_{\mathbb R}\varphi(x)^3\Phi(x)^{K-3}\,dx\\
&=\frac{K^3(K-2)}{2(K-1)^2}\int_{\mathbb R}\varphi(x)^3\Phi(x)^{K-3}\,dx,
\end{aligned}
\]
which is positive because \(K\ge3\). Since \(-\partial_\xi^2q_K(\cdot,d^\circ)\) is continuous by Section~\ref{sec:appendix-E31} and positive at \(\xi=0\), there exists \(\eta>0\) such that it is positive on \([0,\eta)\), and \(-\partial_\xi q_K(\cdot,d^\circ)\) is strictly increasing on \([0,\eta)\). For \(0<\xi<\eta\), every \(s\in[0,\xi)\) satisfies \(-\partial_sq_K(s,d^\circ)<-\partial_\xi q_K(\xi,d^\circ)\), so \(\int_0^\xi\{-\partial_sq_K(s,d^\circ)\}\,ds<\xi\{-\partial_\xi q_K(\xi,d^\circ)\}\). Using this bound for the denominator of \eqref{eq:E.28} gives
\[
\frac{(\xi/2)\{-\partial_\xi q_K(\xi,d^\circ)\}}{\int_0^\xi\{-\partial_sq_K(s,d^\circ)\}\,ds}>\frac{(\xi/2)\{-\partial_\xi q_K(\xi,d^\circ)\}}{\xi\{-\partial_\xi q_K(\xi,d^\circ)\}}=\frac12,
\]
and since \(\beta_K\) is at least this normalized marginal ratio, \(\beta_K>1/2\).

Since \(\beta_K>1/2\) and \((1+\varepsilon)/\{2(1-\varepsilon)\}\to1/2\) as \(\varepsilon\downarrow0\), we can choose \(\varepsilon\in(0,1)\) such that \((1+\varepsilon)/\{2(1-\varepsilon)\}<\beta_K\), and we let \(\delta\in(0,1]\) be the corresponding number in \eqref{eq:E.29}. By \eqref{eq:E.29}, the normalized marginal ratio is at most \((1+\varepsilon)/\{2(1-\varepsilon)\}\) on \((0,\delta)\times\mathcal D_K\), and by \eqref{eq:E.30} it is at most \(1/(2e)\) on \((2K^2/\sqrt\pi,\infty)\times\mathcal D_K\). Both constants are smaller than \(\beta_K\). The supremum over \((0,\infty)\times\mathcal D_K\) is the largest of the suprema over these two sets and over \([\delta,2K^2/\sqrt\pi]\times\mathcal D_K\), where \(\delta\le1<2K^2/\sqrt\pi\), and the suprema over the two outer sets do not reach \(\beta_K\). Hence
\[
\beta_K=\sup_{(\xi,d)\in[\delta,\,2K^2/\sqrt\pi]\times\mathcal D_K}\frac{-(\xi/2)\partial_\xi q_K(\xi,d)}{1-q_K(\xi,d)}.
\]

The set \([\delta,2K^2/\sqrt\pi]\times\mathcal D_K\) is compact and the normalized marginal ratio is continuous on it, so the supremum is attained. Hence there is a state \((\xi^*,d^*)\) at which the normalized marginal ratio equals \(\beta_K\).

Inequality \eqref{eq:E.25} holds at every state, and in particular at \((\xi^*,d^*)\). Substituting \(\beta_K\) for the normalized marginal ratio gives, for every \(a\in(0,1)\),
\[
\frac{1-q_K(\sqrt a\,\xi^*,d^*)}{1-q_K(\xi^*,d^*)}\le1-\beta_K\frac{1-a^{1+\beta_K}}{1+\beta_K}=\frac{1+\beta_Ka^{1+\beta_K}}{1+\beta_K}.
\]

By \eqref{eq:E.27}, which uses only \(\beta_K>0\) and \(0<a<1\), dividing by \(1+\beta_K(1-a)>0\) gives
\[
\frac{1+\beta_Ka^{1+\beta_K}}{1+\beta_K}<\frac{1}{1+\beta_K(1-a)}.
\]

Combining the two displays shows that \(\{1-q_K(\sqrt a\,\xi^*,d^*)\}/\{1-q_K(\xi^*,d^*)\}<1/\{1+\beta_K(1-a)\}\) for every \(a\in(0,1)\). Both sides are positive by part (b), and taking reciprocals gives
\[
\frac{1-q_K(\xi^*,d^*)}{1-q_K(\sqrt a\,\xi^*,d^*)}>1+\beta_K(1-a)\qquad\text{for every }a\in(0,1).
\]

The state \((\xi^*,d^*)\) does not depend on \(a\), in contrast to the state constructed in Section~\ref{sec:appendix-E37} for each \(a\). This proves part (f) and completes the proof of Lemma~\ref{lem:E.1}. \hfill\qedsymbol

\end{document}